\documentclass[12pt,oneside,reqno,titlepage, hyperfootnotes=false]{amsart}
\usepackage{amsmath}
\usepackage{lmodern}
\usepackage{lineno}
\usepackage[T1]{fontenc}
\usepackage[utf8]{inputenc}
\usepackage[a4paper]{geometry}
\usepackage{color}
\usepackage{mathtools}
\usepackage{dsfont}
\usepackage{amstext}
\usepackage{booktabs}
\usepackage{enumerate} 
\usepackage{amsthm}
\usepackage{amssymb}
\usepackage{bm}
\usepackage{setspace}
\usepackage{appendix}
\usepackage[authoryear]{natbib}
\usepackage[unicode=true,pdfusetitle,
 bookmarks=true,bookmarksnumbered=false,bookmarksopen=false,
 breaklinks=false,pdfborder={0 0 0},pdfborderstyle={},backref=false,colorlinks=true]
 {hyperref}
\usepackage{ragged2e}
\makeatletter
\numberwithin{equation}{section}
\numberwithin{figure}{section}

\usepackage{amsfonts}
\usepackage{amstext}
\usepackage{amsthm}

\usepackage{hyperref}
\hypersetup{
    colorlinks = true,
    citecolor = black
}

\usepackage{threeparttable}
\usepackage{graphicx}
\usepackage{enumerate}
\usepackage{listings}
\usepackage{subcaption}

\DeclareMathOperator*{\argmax}{arg\,max}

\usepackage{float}

\newtheorem{thm}{Theorem}

\newtheorem{asm}{Assumption}
\newtheorem{lem}{Lemma}
\newtheorem*{lem*}{Lemma}
\newtheorem*{thm*}{Theorem}

\theoremstyle{definition}

\author{Karun Adusumilli$^\dagger$}
\author{Abhi Vemulapati$^\dagger$}

\makeatother
\begin{document}
\title{Designing persuasive experiments}

\begin{abstract}
     Incentives in experimental design are often misaligned: experimenters design and finance experiments to seek regulatory approval, while regulators seek to maximize social welfare. This incentive misalignment is especially relevant for Bayesian and adaptive trial designs, which have gained substantial popularity in recent years and have been actively encouraged by the US Food and Drug Administration. We propose a framework to resolve this conflict, wherein regulators set a minimum welfare threshold, and experimenters optimize designs subject to this constraint. It requires no knowledge of experimenters' private preferences or costs and mitigates strategic Bayesian persuasion. By varying the welfare threshold, we can map the Pareto frontier between the experimenter’s and regulator's welfare. Under normal priors, Neyman allocation is always the optimal sampling strategy, regardless of specific objectives. We also characterize the optimal stopping rule. A numerical study calibrated to clinical trial data shows sample size reductions of over 48\% relative to classical designs attaining the same social welfare. 
\end{abstract}

\thanks{\textit{This version}: \today{}\\
We would like to thank Frank Diebold, Kevin He, Frank Schorfheide and seminar participants at UC San Diego and University of Pennsylvania for helpful comments.\\
$^\dagger$Department of Economics, University of Pennsylvania}

\maketitle

\section{Introduction}

In many scientific domains, researchers design experiments with the
goal of persuading a regulator to take a specific action. For
instance, in clinical trials, pharmaceutical companies conduct studies
to convince regulators---such as the FDA---to approve a new drug.
Similarly, in development economics, researchers run experiments to
influence policymakers to adopt particular interventions. In these
settings, the preferences of the experimenter can often diverge significantly from those of the regulators. A pharmaceutical company, for example,
may have a strong incentive to secure drug approval regardless of
the drug's actual efficacy.

At the same time, in such strategic environments, the researcher typically bears
the full cost of experimentation. The experimental design is proposed
to the regulator before the experiment begins, and the researcher
is able to fully commit to this design. The regulator, in turn, can
either approve the proposed design or request a new one. 

This misalignment of incentives becomes particularly relevant as regulators move away from traditional fixed-sample trial designs and toward Bayesian and adaptive methods. For example, the FDA has recently promoted a transition towards Bayesian clinical trials. This stance builds on the commitments to Bayesian designs established by the Prescription Drug User Fee Act and the 21st Century Cures Act, both enacted by the US Congress in 2016, and has been codified in new draft guidance (\citealp{us2026guidance}) intended “to encourage the use of Bayesian statistics in clinical trial design and the readout of results” (\citealp{doshi2026fda}).\footnote{In an accompanying 1.09-minute video (\citealp{Makary2026}) posted on X on 12 January 2026, the FDA commissioner at the time, Marty Makary, said: “I'm pleased to announce the FDA is now open to Bayesian statistics. We are putting out new guidance to encourage the use of Bayesian statistics in clinical trial design and readout of results... it is a leap forward beyond the frequentist model of analyzing data".} In parallel, the European Medicines Agency has also committed to releasing a draft `reflection paper' outlining its stance on Bayesian clinical trials by 2027. Meanwhile, encouraging the use of adaptive trials has been another major FDA objective since 2004, when the agency launched the `Critical Path Initiative' to make drug development more efficient (see, for instance, the draft guidance on adaptive trials in \citealp{us2019guidance}). 

Since pharmaceutical companies propose the trial design, adding both Bayesian and adaptive components expands the scope for Bayesian persuasion—an inherent advantage experimenters possess in strategic settings because they can shape the underlying information structure. In the absence of restrictions on allowable designs, experimenters would prefer to stop the trial exactly when regulators are indifferent between approval and rejection, i.e., as soon as the estimated treatment effect becomes (barely) positive. As a result, without such constraints, the overall social value the FDA derives from these trials could, in fact, be lower than what is already achieved by conventional, non-adaptive Randomized Controlled Trial (RCT) designs.

This article studies the optimal design of Bayesian and adaptive experiments when the interests of researchers
and regulators may not be aligned, yet the cost and control over the
design lie entirely with the researcher. 

Experimental design has traditionally been developed within the hypothesis-testing framework. Yet this paradigm is not well suited to Bayesian and adaptive experiments.\footnote{This article does not adopt a position on whether the hypothesis testing framework is best suited to classical fixed-sample trial designs; see, however, \cite{manski2016sufficient} for a well-known critique.} Bayesian methods emerged in part because conventional approaches were often seen as inefficient and not explicitly aligned with social welfare objectives. Typical trial designs, centered on fixed size and power calculations, typically entail large enrollments and extended timelines, all of which can substantially raise development costs. These issues are especially acute for rare diseases and pediatric settings, where eligible participants are scarce. Bayesian techniques and informative priors are viewed as a way for sponsors to “borrow strength” by incorporating external evidence—such as findings from adult studies, external control information, or data shared across related diseases—to support primary regulatory conclusions.\footnote{As then FDA Commissioner Marty Makary stated in an agency statement: ``Bayesian methodologies help address two of the biggest problems of drug development: high costs and long timelines. Providing clarity around modern statistical methods will help sponsors bring more cures and meaningful treatments to patients faster and more affordably.''} Moreover, routinely applying a 5\% significance level to determine trial size across all studies, regardless of setting or the stakes involved, may not be socially optimal. Over the past 15 years, the FDA has authorized several drugs—including Makena (Hydroxyprogesterone Caproate), Exondys 51 (Eteplirsen), and, more controversially, Aduhelm (Aducanumab)—despite limited evidence of clinical benefit under conventional standards. In each case, the rationale pointed to the pressing need for new treatments, underscoring how welfare considerations are often invoked to motivate departures from traditional hypothesis-testing thresholds. Lastly, adaptive designs introduce further complications for experimental design, since size and power metrics alone are not sufficient to uniquely characterize optimal stopping rules. 

Within the move toward Bayesian inference, the statistical threshold for regulatory approval is sometimes taken to be the posterior probability that the treatment effect exceeds a pre-specified value. For example, the Pfizer-BioNTech COVID-19 Vaccine (Comirnaty) was authorized because the posterior probability that its vaccine efficacy was greater than 30\% surpassed a 98.6\% threshold. Nevertheless, it remains unclear how such thresholds ought to be selected, and it is also unlikely that a universal threshold of this kind would be socially optimal. 

In this study, we advocate using welfare
maximization as the primary criterion for approving Bayesian experimental designs.
We propose that the regulator impose a minimum requirement
on the expected social welfare that a proposed design must generate
in order to be approved. The experimenter, in turn, selects a design
that maximizes their own objective, subject to this welfare constraint.
We characterize the form of optimal experimental designs
within this framework, explicitly allowing for misaligned preferences between experimenters
and regulators.

Our proposal contains several key features that make it well suited
to practical, real-world contexts. First, introducing a welfare threshold curbs the
scope for Bayesian persuasion by limiting an experimenter’s ability to skew outcomes via strategic experiment design. Notably, by varying the welfare threshold, we can trace out the Pareto frontier between the experimenter’s welfare and the regulator’s.\footnote{See Figure \ref{fig:welfare_v_meantime}, which plots the expected number of observations required as a function of the welfare threshold.} In fact, several of our findings, such as the characterization of the optimal sampling strategy, are insensitive to the specific threshold. We also provide a natural benchmark for selecting the threshold: the expected welfare a regulator would achieve if the experimenter were limited to running a conventional RCT.

Moreover, and just as importantly, our method does not depend on the regulator having information about the experimenter’s preferences or the costs of carrying out the study. This matters because such information is typically not verifiable, and experimenters have strong incentives
to misstate it. Indeed, experimenters usually possess an informational advantage over regulators regarding the logistics and costs of running trials. Many components of the optimal experimental design, such as optimal stopping rules and optimal approval and rejection cutoffs, hinge on these experimenter-specific parameters and private information. It is for this reason that we recommend regulators set welfare targets rather than contract directly on the experiment design.\footnote{Indeed, we suspect that
the widespread adoption of the 5\% significance level in hypothesis testing may arise precisely from the difficulty of eliciting true preferences and cost structures. While recent decision theoretic
rationales (e.g., Tetenov, 2016; Viviano et al., 2021) for hypothesis testing in classical contexts tie the choice of test size directly to experimental costs, such links are not often observed in practice.} 

In our framework, we allow the experimenter to jointly optimize the sampling allocation and the stopping rule; the latter endogenously determines the total sample size. Utilizing a continuous-time framework under Gaussian priors, we characterize an optimal experimentation policy that exhibits several notable properties.

First, the optimal sampling rule is always the Neyman allocation, just as in classical designs, regardless of the specific utility functions of the experimenter or the regulator. This result is based on an extension of \cite{liang-mu-syrgkanis-ecta-2022} to a multi-agent setting. The optimality of Neyman allocation arises from a fundamental alignment of interests regarding informational efficiency: despite their divergent ultimate objectives, both the experimenter and the regulator seek to minimize the posterior variance of the treatment effect relative to the cost of experimentation.

Second, the optimal stopping strategy diverges from the single-agent benchmark if and only if the experimenter derives an additive benefit from approval that is independent of the treatment effect's magnitude. In the presence of this benefit, the optimal stopping boundaries become asymmetric: the threshold for approval is shifted downward relative to the rejection threshold, reflecting the experimenter's bias toward treatment adoption. 

Third, the entire experimental strategy is invariant to the specific weighting of Type I (approving the treatment when the true treatment effect is negative) and Type II (rejecting the treatment when the true treatment effect is positive) errors. The optimal strategy thus remains unchanged whether the regulator and experimenter are expected utility maximizers or regret-minimizers.

In a numerical analysis with priors calibrated using clinical trial data, we show that our proposed strategies reduce expected sample sizes by more than 48\% compared with traditional RCT designs that achieve the same welfare level. For pharmaceutical firms, this implies nearly \$6 million in savings on experimentation costs.\footnote{Conversely, we find that holding the expected sample size constant would lead to only a 3.2\% increase in social welfare.} 

While the preceding results are derived within a continuous-time framework, we demonstrate that the finite-sample counterparts of these strategies remain asymptotically optimal under `small-cost asymptotics' (\citealp{adusumilli2025}).

\subsection{Related literature}
If we remove the welfare constraint, the framework presented here collapses to the classic model of costly Bayesian persuasion (\citealp{kamenica-gentzkow}). It is precisely to emphasize this link that we refer to it as persuasive experimental design.

Our results add to the broad literature on information acquisition and optimal stopping with costly sampling. The seminal contributions of \cite{wald} and \cite{arrow} derive optimal stopping rules for welfare maximization in binary state spaces. \cite{Bather_Walker_1962} study the more general problem of selecting an optimal stopping time for a Brownian motion, where the cost of any path depends solely on its terminal time and position. More recent work has shifted toward richer distributional frameworks and more flexible experimentation.
\cite{fudenberg-strack-strzalecki2018} characterize the welfare-maximizing stopping rule for choosing between two treatments under Gaussian priors. \cite{liang-mu-syrgkanis-ecta-2022} extend these insights by endogenizing the sampling strategy and demonstrate that Neyman allocation is optimal under symmetric Gaussian priors, regardless of the experimenter's utility function. A common element across these contributions, however, is the focus on a single decision-maker. Our work adds to this literature by determining the optimal experimentation policy in a multi-agent environment in which the experimenter and the regulator have conflicting objectives. 

This article thus also contributes to a growing line of economic research that examines how incentives and the presence of multiple agents shape empirical work. \cite{tetenov2016}, \cite{viviano2021model}, and \cite{bates2023incentive} investigate hypothesis testing in classical, i.e., non-Bayesian fixed sample settings when incentives are misaligned and emphasize how controlling size can help `screen out' experimenters who aim to put forward low-quality treatments. This study focuses on Bayesian and adaptive trial design, where the imposition of a prior and a welfare constraint also indirectly functions as a screening mechanism for undesirable participants. In addition, we study the full experimental design problem---including sampling and stopping rules---rather than focusing exclusively on approval thresholds. 

In more closely related work, \cite{yoder2022designing} analyzes environments where a principal employs an experimenter to conduct research but is unable to contract on the experimenter's private type (such as their utility function, research costs, and so forth). This differs from our setup, where the experimenter conducts research with the goal of persuading a regulator; however, we do share the same contracting constraints. \cite{henry2019research} examine a related setting in which the experimentation strategy is determined as the Nash equilibrium of a game between the experimenter and the regulator. \cite{mcclellan2022experimentation} studies a similar environment where the regulator selects the welfare-maximizing experimental design, subject to a dynamic participation constraint for the experimenter. Both articles differ from our framework: in our model, as in \cite{yoder2022designing}, the regulator is either unable or not allowed to directly optimize against any particular experimenter. The same constraints must apply to all experimenters, and the regulator cannot write contracts contingent on experimenter-specific characteristics—such as costs and benefits—that may be private information. As a result, the regulator in our setting can only impose a minimum welfare requirement. That said, by selecting an appropriate value of the welfare cutoff, we can recover the same results as in those papers. Moreover, our state space is considerably richer than in \cite{henry2019research} and \cite{mcclellan2022experimentation}, with two unknown treatment effects rather than binary states, and our decision space includes sampling strategies in addition to optimal stopping. 

Our analysis of the continuous-time experimental design problem makes use of the formalism introduced in \cite{adusumilli-continuous-arts}. We then rely on continuous-time asymptotic representations, also developed in that work, to translate the results to discrete-time environments. The particular asymptotic regime we employ is based on `small-cost asymptotics', introduced in \cite{adusumilli2025}. In this setting, one assumes that the marginal cost of each observation is negligible relative to the benefit of approval, with the cost-benefit ratio converging to zero at a specified rate. We demonstrate that this assumption is supported by empirical estimates from the clinical trial literature. More broadly, small-cost asymptotics fall under the umbrella of local asymptotics, which analyze local perturbations around reference parameter values. Other recent studies that use local asymptotic techniques to analyze adaptive experiments include \cite{wagerxu}, \cite{fanglynn}, and \cite{hirano2025asymptotic}.

\section{General Setup}\label{Sec:General_setup}

We start by presenting a general framework for studying experimental design when the experimenter and the regulator have divergent objectives. Since the canonical application of this framework is the design of Phase 3 clinical trials, we use this as a running example throughout this article.

The model involves two agents: a regulator (Alice) and an experimenter (Bob). Bob proposes a project (e.g., a new drug) whose utility depends on an underlying state of the world $s \in \mathcal{S}$. A simple example could be the binary state space $\mathcal{S} = \{0, 1\}$, where $s=1$ indicates an effective drug and $s=0$ indicates an ineffective one; but ultimately, we are interested in continuous, bivariate $s$.

Both agents are expected utility maximizers. Let $\Delta(\mathcal{S})$ denote the set of all beliefs, i.e., probability measures on $\mathcal{S}$, endowed with the metric of weak convergence. The space of distributions over beliefs (i.e., distributions over probabilities) is represented by $\Delta(\Delta(\mathcal{S}))$. 

\subsection{Experimental design}

We assume that both Alice and Bob agree on a prior $p_0$. To persuade Alice to approve his proposal, Bob submits an experimental protocol (which he is required to commit to). If Alice approves the protocol, the experiment is conducted, yielding a signal that allows both agents to update their beliefs, ex-post, to a posterior $p \in \Delta(\mathcal{S})$. Note that the experiment cannot be conducted without Alice's approval.

Employing the terminology of Blackwell experiments, we characterize an experiment by the ex-ante distribution over posterior beliefs $Q \in \Delta(\Delta(\mathcal{S}))$ that it induces. Posterior consistency demands that $Q$ satisfy the martingale constraint
\begin{equation}
\int p dQ(p)=p_{0}.\label{eq:martingale=000020constraint}
\end{equation}
This is the minimal and only constraint on $Q$ that experimentation
imposes (at this point, we leave the class of possible experiments
unrestricted and allow for garbling of signals).

\subsection{Utilities and welfare}

After the experiment ends, Alice selects an action: accept ($\delta = 1$) or reject ($\delta = 0$). Alice's utility, $u_A(s, \delta)$, is determined by the state $s$ and her chosen action. For a given posterior $p$, Alice selects $\delta^*(p) \in \arg\max_{\delta \in \{0,1\}} \int u_A(s, \delta) \, dp(s)$, which yields the following ex-post expected welfare:
$$
V_A(p) \coloneq \max_{\delta \in \{0, 1\}} \int_{\mathcal{S}} u_A(s, \delta) \, dp(s).
$$

Whenever Alice is indifferent between $\delta \in \{0,1\}$, we adopt the tie-breaking convention that she chooses Bob's preferred action, that is, $\delta = 1$. In fact, as shown in Section \ref{subsec:tie-breaking}, when information arrives gradually, we can, without mathematical loss of generality, suppose that Alice selects Bob's preferred action whenever she is indifferent, no matter what her actual tie-breaking rule may be.

Bob's overall utility, which is of the form 
$$
u_B(s,\delta) - C(Q),
$$
consists of two components. The first component, $u_{B}(s,\delta)$,
depends on the state of the world and Alice's action. For example, Bob may receive a constant benefit $B$ if the drug is approved: $u_B(s, \delta) = B \cdot \mathds{1}\{\delta = 1\}$. The second component captures the cost of the experiment, which is borne by Bob. Following \cite{denti2022experimental}, we model this ex-ante, as a function, $C:\Delta(\Delta(\mathcal{S})) \to [0,\infty)$, of the distribution over posteriors $Q$. Intuitively, $C(Q)$ represents the least-expensive way for Bob to run a specific experiment---which corresponds to some unique distribution over posteriors $Q$---with the convention that $C(Q) = \infty$ if the posterior distribution is unattainable, e.g., due to limitations on what experiments are allowed or are feasible. 

Often, the cost $C(Q)$ is posterior separable, i.e., it can be written as 
\[
C(Q) = \int\phi(p)dQ(p) - \phi(p_0),
\]
where $\phi(\cdot)$ is convex. Even though we measure information cost in an ex-ante
sense, it is known that posterior separable costs can be induced by specific running cost
functions (\citealp{mor}). Under a posterior separable cost, Bob's ex-ante expected welfare, for a given distribution over posteriors $Q$, is given by:
$$
\int [V_B(p) - \phi(p)] \, dQ(p), \quad \text{where} \\
\quad V_B(p) \coloneq \int u_B(s, \delta^*(p)) \, dp(s).
$$

\subsection{The regulator constraint and Bob's optimization problem}

Given the misalignment of incentives between Alice and Bob, it is suboptimal for Alice to unconditionally approve Bob's experimental protocols, as this allows Bob to leverage Bayesian persuasion to his advantage. At the same time, as in \cite{yoder2022designing}, contracting directly on Bob's utility or cost functions is infeasible for two reasons.  First, these parameters constitute private information and are subject to potential misrepresentation. For instance, Bob may claim that an experiment is more expensive than it really is in order to induce Alice to accept smaller sample sizes. Second, legal regulations often obligate Alice to impose identical conditions on all experimenters, which prevents her from tailoring her requirements specifically to Bob. 

Under these constraints, it would be neither practical nor efficient for Alice to prescribe a particular experimental design (as in, e.g.,  \citealp{mcclellan2022experimentation}), since the best option depends on Bob’s parameters. Instead, we propose that Alice impose a minimum threshold $V_0$ for the ex-ante expected welfare. Specifically, we introduce the regulator constraint:
\begin{equation}
\int V_{A}(p)dQ(p)\ge V_{0},\label{eq:regulator=000020constraint}
\end{equation}
and suggest that Alice commit to clearing any experimental proposal submitted by Bob that satisfies this condition. Importantly, by varying $V_0$, Alice can trace out the Pareto frontier between Bob's and her own welfare.

One could, following \cite{henry2019research}, require the experimentation strategy be determined as the Nash equilibrium of a game between Alice and Bob; this simply corresponds to selecting a particular value of $V_0$. But from Alice's point of view, computing the Nash equilibrium would require knowing Bob’s parameters, which, as noted above, are private. More broadly, we do not believe Alice should seek to capture the entire welfare surplus from experimentation. In reality, Alice and Bob play a repeated game, and it would be inefficient from a social standpoint to discourage Bob from undertaking future research. Accordingly, the appropriate value of $V_0$ should reflect how Alice weighs Bob’s profits, research incentives, and aggregate consumer welfare. Although it can always be varied, as a pragmatic benchmark for $V_0$, we suggest using the expected welfare Alice would obtain under the common prior $p_0$ if Bob were restricted to running a standard RCT. Notably, in Section \ref{sec:Calibration} we show that calibrating $V_0$ in this way requires only the prior variance and no other experimenter-specific information.

Given Alice's requirements, Bob's actions reduce to choosing an experiment $Q$ that maximizes his ex-ante expected welfare, subject to the martingale and regulator constraints (\ref{eq:martingale=000020constraint}) and (\ref{eq:regulator=000020constraint}):
\begin{align}
 & \max_{Q\in\Delta(\Delta(\mathcal{S}))}\int V_B(p) dQ(p) - C(Q) \label{eq:regulator problem}\\
 & \ \textrm{s.t.}\ \int p dQ(p)=p_{0}, \quad \int V_{A}(p)dQ(p)\ge V_{0}.
\end{align}
Note that when $V_0 = -\infty$ (no regulatory oversight), this model reduces to costly Bayesian persuasion (\citealp{kamenica-gentzkow}).

\subsection{A dual formulation of Bob's problem}

When experimentation costs are posterior separable, Bob's experimental design problem admits the dual representation
\begin{align*}
 & \min_{\{f(\cdot)\in L_{1}(\mathcal{S}),\gamma\ge0\}}\int f(s)dp_{0}(s) \\
 & \ \textrm{s.t.}\ [V_{B}(p) - \phi(p)] +\gamma [V_{A}(p) - V_0] \le\int f(s)dp(s)\ \forall\ p\in\Delta(\mathcal{S}),
\end{align*}
where $L_{1}(\mathcal{S})$ is the set of all Lipschitz continuous
functions on $\mathcal{S}$. The above is a slight extension of \citet[Lemma 1]{kolotilin2025persuasion}.

The function $f(s)$ can be interpreted as the shadow price associated with state $s$. We may view the dual problem as describing a decentralized economy in which Alice and Bob outsource the creation of welfare to an entrepreneur, Carol, who charges a price $f(s)$ for each latent state $s$. Alice and Bob pay Carol an ex-ante fee of $\int f(s)\,dp_0(s)$, and in return, Carol commits to delivering at least $[V_B(p) - \phi(p)]+ \gamma [V_A(p) - V_0]$ in every possible contingency where the posterior is $p$. The solution to the dual problem then characterizes the minimum ex-ante payment Carol can charge Alice and Bob without suffering an expected loss.   

The dual problem greatly simplifies the determination of the optimal strategy when $\mathcal{S}$ is discrete since it then reduces to a linear programming problem. In this article, however, our main focus is on the case where $s$ is continuous, and the functions $V_A(\cdot)$ and $V_B(\cdot)$ are also difficult to evaluate explicitly. As a result, the dual problem is of limited practical use for our purposes, but we present it here because it is of conceptual interest.

\section{Persuasive Experimental Design under Incremental Learning}\label{Sec:Incremental_learning}

We now focus on a specific class of problems that contrasts a proposed treatment with an active control or placebo, whose treatment effect is likewise unknown at the outset. Unfortunately, solving the experimental design problem (\ref{eq:regulator problem}) in general contexts is still infeasible, especially when $s$ is continuous. We therefore make three further restrictions: (1) information is constrained to arrive incrementally in the form of Gaussian processes, (2) the cost of sampling is proportional to the number of observations, and (3) the priors are Gaussian.

The first two restrictions are relatively mild. As we demonstrate in Section \ref{sec:local_asymptotics}, incremental signal processes can be formally motivated using local asymptotics. The assumption of linear sampling costs follows a long tradition in sequential analysis dating back to \cite{wald}; see also \cite{arrow}, \cite{fudenberg-strack-strzalecki2018}, and \cite{adusumilli2025}. In practice, clinical trials are often conducted in batches, and we can think of our framework as assuming that the marginal cost of each batch remains approximately constant. Furthermore, while firms do incur significant fixed costs to initiate trials, these are essentially sunk costs that do not influence optimal experimental design (although they may shift how Alice and Bob split the welfare surplus from experimentation).

The restriction to Gaussian priors is motivated by both technical and practical considerations. Formally, this assumption allows us to make use of the results of \cite{liang-mu-syrgkanis-ecta-2022} to characterize the optimal treatment assignment rule under very general conditions. Because the resulting assignment rule is history-independent and static, the complex dynamic optimization problem reduces to determining the optimal stopping rule. Empirically, meta-analyses of large-scale clinical trial databases suggest that the distribution of realized treatment effects is well-approximated by a Gaussian prior. Also, as suggested by the FDA draft guidance on Bayesian clinical trials (\citealp[Section III.A]{us2026guidance}), the prior could be estimated from Phase 2 trial results; by the central limit theorem, this is typically well-approximated by a Gaussian distribution. 

\subsection{Setup under incremental learning}\label{subsec:Incremental learning}
As discussed above, we consider a special case of the general setup from Section \ref{Sec:General_setup}, where the aim is to determine the best option out of two treatments $a \in \{0,1\}$. These treatments are associated with unknown mean rewards $\mu_0, \mu_1$ and known variances $\sigma^2_0, \sigma^2_1$. The state of the world is therefore $s = (\mu_1, \mu_0)$. We assume that both Bob and Alice agree on a Gaussian prior $(\mu_1, \mu_0) \sim \mathcal{N}(\boldsymbol\mu^0, \Sigma)$. 

We employ the formalism of \cite{adusumilli-continuous-arts} to describe adaptive experiments under incremental learning. The underlying information environment consists of two independent Gaussian signal processes corresponding to each treatment $a$, given by
\begin{align}
dz_a(\gamma) &= \mu_ad\gamma + \sigma_adW_a(\gamma),
\end{align}
where $W_0(\cdot), W_1(\cdot)$ are independent standard Brownian motions. Intuitively, $z_a(\gamma)$ corresponds to the cumulative sum of outcomes generated by treatment arm $a$ when it is sampled for a duration $\gamma$. The signal processes $\{z_a(\cdot)\}_a$ are supplemented by an exogenous random variable \(U \sim \text{Uniform}[0,1]\), which accounts for the randomness inherent in any potentially randomized experimentation strategy. Let $\mathbb{P}$ denote the joint probability distribution over the state space and the sample paths, composed of the prior $p_0$ over $s$, together with the conditional law of $\{z_1(\cdot), z_0(\cdot), U\}$ given $s$. Also, take $\mathcal{H}^{(a)}_\gamma \coloneq \sigma\{z_a(r) : r \le \gamma \}$ to be the natural filtration generated by the sample paths of $z_a(\cdot)$ between $0$ and $\gamma$, and set:
\begin{equation}
\mathcal{G}_{\gamma_1, \gamma_0} \coloneq \mathcal{H}^{(1)}_{\gamma_1} \vee \mathcal{H}^{(0)}_{\gamma_0} \vee \sigma(U)
\end{equation}

The experimental design consists of specifying both a sampling strategy and a stopping time.\footnote{Naturally, Bob must define this strategy at the outset and commit to it (Alice can monitor adherence).} Following \cite{adusumilli-continuous-arts}, we represent the sampling strategy by an allocation process $\bm{q} \equiv \{q_a(t)\}_a$, which records the cumulative amount of attention assigned to treatment $a$ up to time $t$. For example, if $q_a(t) = \gamma_a$, then at time $t$ we have observed the path of $z_a(t)$ over the interval $[0,\gamma_a]$. We work with allocation processes rather than standard sampling rules or policies because the latter are generally not measurable in continuous time. The requirements of an allocation process are that $\{q_a(\cdot)\}_a$ are monotone, satisfy $q_1(t) + q_0(t) = t$ for all $t$, and obey the informational constraint that the event $\{q_1(t) \le \gamma_1,\; q_0(t) \le \gamma_0\}$ is $\mathcal{G}_{\gamma_1,\gamma_0}$-measurable for every choice of $\gamma_1,\gamma_0,t$. The observed signal at time $t$ is given by
$$
x_a(t) \coloneq z_a(q_a(t)),
$$
and the information available at time $t$ under a given experiment is summarized by $\mathcal{F}_t^{\bm{q}} = \mathcal{G}_{q_1(t), q_0(t)}$. The filtration $\mathcal{F}_t^{\bm{q}}$ is indexed by $\bm{q}$ because it depends on the particular sampling strategy in use. A stopping time $\tau$ that is adapted to $\mathcal{F}_t^{\bm{q}}$ then specifies when the experiment terminates. Since $\mathcal{F}_t^{\bm{q}}$ incorporates the exogenous randomization $U$, the framework accommodates randomized stopping times. We denote the overall experimentation strategy, combining both the allocation process and the stopping time, by $\bm{d} = \{\{q_a(\cdot)\}_a, \tau \}$.

Following the experiment, Alice chooses an implementation rule $\delta \in \{0, 1\}$. Given an experimentation strategy $\bm{d}$, we represent Alice's (Bayes) optimal response by
$$
\delta^*_{\bm{d}} := \mathds{1}\left\{ \mathbb{E}\left[u_A(s,1) \vert \mathcal{F}_\tau^{\bm{q}} \right] \ge \mathbb{E}\left[u_A(s,0) \vert \mathcal{F}_\tau^{\bm{q}} \right] \right\}.
$$

Since the priors are Gaussian, standard results in stochastic filtering imply that the posterior distribution $p^{\bm{q}}(t)$ of $s$ induced by $\bm{q}$ is also Gaussian and is given by 
$$
p^{\bm{q}}(t) \propto \left[ \prod_{a\in \{0,1\}} \phi \left(x_a(t) \vert \mu_a q_a(t), \sigma_a^2q_a(t) \right) \right] \times \phi\left(\mu_1, \mu_0 \vert \bm{\mu}^0, \Sigma \right),
$$
where $\phi(\cdot|m, A)$ represents the normal pdf with mean $m$ and covariance matrix $A$. 

Denote the set of all distributions over posteriors inducible through incremental learning by
$
\mathcal{Q} \equiv \{\textrm{ex-ante law of }p^{\bm{q}}(\tau)\}_{\bm{d}}.
$
As noted earlier, we simplify the structure of experimentation costs by assuming it is given by $c\tau$ for any stopping time $\tau$, where $c > 0$ is some constant denoting the marginal cost of each observation. We can relate this to the ex-ante cost of information $C(\cdot)$ from Section \ref{Sec:General_setup} as 
$$
C(Q) = \begin{cases}
\inf_{\{(\bm{q},\tau): p^{\bm{q}}(\tau) \sim Q\}} c\tau &\text{if $Q\in \mathcal{Q}$}\\
\infty &\text{otherwise}.
\end{cases}
$$
Apart from this structure on $C(\cdot)$, our setup is otherwise the same as Section \ref{Sec:General_setup}; indeed, all the constraints on experimentation are subsumed into the form of $C(\cdot)$.

\subsection{Characterizing the optimal sampling strategy}\label{subsec:Optimal_sampling}
We now characterize the optimal sampling strategy. As we demonstrate below, this can be achieved under quite broad conditions.

In what follows, let $\mathbb{P}_{\bm{d}}$ represent the restriction of the measure $\mathbb{P}$ to the $\sigma$-algebra $\mathcal{F}_\tau^{\bm{q}}$, and $\mathbb{E}_{\bm{d}}[\cdot]$ its associated expectation (this is the probability induced by the sampling strategy). 

\begin{asm}\label{asm-1}
(i) The utility functions $u_A(\cdot), u_B(\cdot)$ depend on $s$ only through $\mu_1 - \mu_0$.\\
(ii) The functions $\bar{u}_A(s)\coloneq \vert\max_\delta u_A(s,\delta)\vert$ and $\bar{u}_B(s)\coloneq \vert\max_\delta u_B(s,\delta)\vert$ are integrable with respect to the prior $p_0$, and are also bounded whenever $\vert s \vert$ is bounded.\\
(iii) There exists some experimentation strategy $\bm{d}$ such that $\mathbb{E}_{\bm{d}}[u_A(s, \delta^*_{\bm{d}})] \geq V_0$ and 
$\mathbb{E}_{\bm{d}}[u_B(s, \delta^*_{\bm{d}}) -c\tau] \ge -L$ for some $L < \infty$.\\
(iv) The priors over $\mu_1, \mu_0$ are independent, i.e., $\mathcal{N}(\boldsymbol\mu^0, \Sigma) = \mathcal{N}(\mu_1^0, \Sigma_{11}) \times \mathcal{N}(\mu_0^0, \Sigma_{00})$. Furthermore, $\Sigma_{11}/\sigma_1^2 = \Sigma_{00}/\sigma_0^2$.
\end{asm}

Assumption \ref{asm-1}(i) posits that Alice and Bob care exclusively about the relative difference between $\mu_1$ and $\mu_0$ rather than their absolute levels. In effect, this serves as a standard normalization for welfare that simplifies the parameter space. Assumption \ref{asm-1}(ii) is a mild regularity condition on the growth rates of $u_A(\cdot), u_B(\cdot)$; given that a Gaussian prior possesses exponential tails, this condition is satisfied by most standard utility functions, which generally exhibit sub-exponential growth relative to $|s|$. Assumption \ref{asm-1}(iii) represents Bob's participation constraint by ensuring that his loss is bounded from below. Assumption \ref{asm-1}(iv) states that $\mu_1/\sigma_1$ and $\mu_0/\sigma_0$ have the same (independent) prior. It is introduced primarily to simplify the algebra because, as we demonstrate below, it implies that the Neyman allocation is always optimal. The case of more general covariance structures is analyzed in Appendix \ref{sec:general_cov_structures}. Even in this broader setting, the optimal sampling rule is qualitatively similar: it prescribes sampling from a single treatment for some initial period of time before switching to the Neyman allocation. 

Under Assumption \ref{asm-1}, we can rewrite Bob's experimental design problem (\ref{eq:regulator problem}) as 
\begin{equation}\label{eq:constrainedtwotreatments}
\begin{aligned} 
&\sup_{\bm{d}}\mathbb{E}_{\bm{d}}[u_B(\mu_1 - \mu_0, \delta^*_{\bm{d}}) -c\tau] \\ 
& \textrm{s.t. }\mathbb{E}_{\bm{d}}[u_A(\mu_1 - \mu_0, \delta^*_{\bm{d}})] \geq V_0.
\end{aligned}
\end{equation}
The (Lagrangian) dual of the problem is
\begin{align}\label{eq:twotreatments}
    \inf_{\lambda \geq 0}  \sup_{\bm{d}} \mathbb{E}_{\bm{d}}\left[u_B(\mu_1 - \mu_0, \delta^*_{\bm{d}}) + \lambda \left\{ u_A(\mu_1 - \mu_0, \delta^*_{\bm{d}}) - V_0 \right\} - c\tau \right].
\end{align}
The dual problem is very similar to the class of information acquisition problems analyzed by \cite{liang-mu-syrgkanis-ecta-2022}. It is, therefore, a lot easier to analyze than the primal. 

\begin{thm}\label{thm:Sampling_strategy}
Suppose Assumption \ref{asm-1} holds. Then:
\begin{enumerate}[(i)]
\item The primal and dual problems, (\ref{eq:constrainedtwotreatments}) and (\ref{eq:twotreatments}), attain the same bounded value. 

\item  The optimal sampling strategy is the Neyman allocation: 
$$
q_a^*(t) = \frac{\sigma_a}{\sigma_1 + \sigma_0}t \ \forall\ t. 
$$
\item Under the Neyman allocation, the posterior variance $\varrho_t^*$ of $\mu_1 - \mu_0$ is deterministic and given by 
$
\varrho_t^* = \frac{\sigma^2}{\sigma^2 \varrho_0^{-1} + t},
$
where $\sigma^2 \coloneq (\sigma_1 + \sigma_0)^2$ and $\varrho_0$ is the prior variance of $\mu_1 - \mu_0$. Furthermore, the posterior mean, $m_t$, of $\mu_1 - \mu_0$ evolves as 
\begin{equation}\label{eq:evolution of m(t)}
m_t = \mu_1^0 - \mu_0^0 + \int_0^t \frac{\sigma^{-1}}{\sigma^{-2}t + \varrho_0^{-1}} dW(t),
\end{equation}
where $W(\cdot)$ is standard Brownian motion.
\end{enumerate}
\end{thm}

The first part of Theorem \ref{thm:Sampling_strategy} is new. We demonstrate the equivalence between the primal and dual problems by establishing a minimax theorem. This is not as straightforward as it initially appears because the space of decisions $\bm{d}$ is very rich and infinite dimensional. 

If Alice’s choices $\delta_{\bm{d}^*}$ did not appear in the dual formulation, then the framework would be fully identical to \cite{liang-mu-syrgkanis-ecta-2022}, and the second part of the theorem would follow directly from their analysis. As it stands, however, we provide a multi-player generalization of those results.  Remarkably, the second part states that the optimal sampling strategy is independent of both sampling costs $c$ and the form of the utility functions $u_A(\cdot), u_B(\cdot)$. The strategy is also deterministic and extremely simple: one should always employ the Neyman allocation. The Neyman allocation is optimal because it minimizes the posterior variance of $\mu_1 - \mu_0$ at every possible time point. Intuitively, it is optimal to reach precise beliefs, even unfavorable ones, as quickly as possible to minimize experimentation costs. This is achieved by choosing the strategy that minimizes posterior variance.\footnote{Note that the optimality of the Neyman allocation in this setting is due to a different structural reason than in \cite{adusumilli2025}, where it arises from a minimax-regret analysis.}

To understand the third part of Theorem \ref{thm:Sampling_strategy}, observe that under the Neyman allocation and Assumption \ref{asm-1}(i), the process
\begin{equation}\label{eq:signal_process_under_Neyman_allocation}
x(t) \coloneq \sigma \left(\frac{x_1(t)}{\sigma_1} - \frac{x_0(t)}{\sigma_0} \right) \sim (\mu_1 - \mu_0) t + \sigma W(t),
\end{equation}
is a sufficient statistic for $\mu_1 - \mu_0$. Indeed, representing the allocation process under the Neyman allocation by $\bm{q}^*$, the posterior distribution of $\mu_1 - \mu_0$ is given by
\begin{equation}\label{eq:posterior_dist_normal}
\mu_1-\mu_0\vert\mathcal{F}_{t}^{\bm{q^*}}\sim\mathcal{N}\left(\frac{\sigma^{-2}x(t)+\varrho_0^{-1}(\mu^{0}_1-\mu^0_0)}{\sigma^{-2}t+\varrho_0^{-1}},\frac{1}{\sigma^{-2}t+\varrho_0^{-1}}\right).
\end{equation}
Equation (\ref{eq:evolution of m(t)}) then follows from standard results in stochastic filtering.

\subsection{Characterizing the optimal stopping time}

\subsubsection{Utility specifications}\label{subsec:Utility_spec}

In contrast to the optimal sampling strategy, the optimal stopping rule inherently depends on the specific utility functions of Alice and Bob. For tractability, we therefore adopt specific functional forms for $u_A(\cdot)$ and $u_B(\cdot)$.

\begin{asm}\label{Asm-2}
Alice's utility function is of the form:
\[u_A(s,\delta) \coloneq u(\mu_1 - \mu_0,\delta;\alpha) = \begin{cases}
    \alpha(\mu_1 - \mu_0),\ \text{if }\delta = 1,\\ \ (1-
    \alpha)(\mu_0-\mu_1),\ \text{if }\delta=0,
\end{cases}
\]
where $\alpha \in [0,1]$ is a known parameter.
\end{asm}

Assumption \ref{Asm-2} specifies that Alice's utility is linear in $|\mu_1 - \mu_0|$, conditional on the decision $\delta$. Within this framework, Alice faces two distinct errors: a Type I error occurs when approving a treatment ($\delta = 1$) despite $\mu_1 < \mu_0$, and a Type II error occurs when rejecting a treatment ($\delta = 0$) when $\mu_1 > \mu_0$. We allow for asymmetric utility across these error types, with the degree of asymmetry governed by the parameter $\alpha$. Two cases are particularly significant: $\alpha = 1$ recovers the standard utilitarian welfare, whereas $\alpha = 1/2$ is equivalent to the negative of welfare-regret.

Under Assumption \ref{Asm-2}, Alice's optimal action is to choose the treatment with the higher posterior mean, i.e. $\delta^* = \mathds{1}[m_\tau \geq 0]$. This results in an ex-post Bayes welfare of $\max\{\alpha m_\tau, -(1-\alpha)m_\tau\}$. 

Regarding Bob's preferences, we assume he derives some utility solely from the approval of the treatment.

\begin{asm}\label{Asm-3}
Bob's utility is of the form: 
$$
u_B(s,\delta) = B\mathds{1}[\delta = 1] + \gamma u(\mu_1 - \mu_0,\delta;\alpha^\prime),
$$
where $\gamma, B$ are positive quantities, and $u(s,\delta;\alpha^\prime)$ follows Alice's functional form with a potentially different asymmetry parameter $\alpha^\prime$.
\end{asm}

Beyond the fixed benefit $B$, Assumption \ref{Asm-3} allows Bob to derive additional utility (or disutility) from the magnitude of the treatment effect conditional on approval. For instance, setting $\alpha' = 1$ implies $u(\mu_1 - \mu_0, \delta; \alpha') = (\mu_1 - \mu_0)\delta$. However, as we demonstrate below, the fundamental incentive misalignment between Alice and Bob stems exclusively from the additive benefit $B$. While the players may weigh Type I and Type II errors differently (when $\alpha \neq \alpha'$), this preference heterogeneity does not alter the structural form of the optimal stopping rule. Indeed, setting $B=0$ reduces the model to the single-player experimentation setup in \cite{fudenberg-strack-strzalecki2018}. 

To formally understand how these properties come about, consider the Lagrangian (dual) objective, which represents a weighted average of the players' utilities:
$$
u_B(\cdot, \delta) + \lambda u_A(\cdot, \delta) = B \cdot \mathds{1}[\delta = 1] + \gamma u(\cdot, \delta; \alpha') + \lambda \left( u(\cdot, \delta; \alpha) - V_0 \right).
$$
It is easy to verify that there exist ``effective'' parameters $\tilde{V}_0 \in (0, \infty)$, $\tilde{\lambda} \in [0, \infty)$, and $\tilde{\alpha} \in [-1, 1]$ such that the combination of the players' error-weighted utilities maps to a single representative utility function:
$$
\gamma u(\cdot, \delta; \alpha') + \lambda u(\cdot, \delta; \alpha) = \tilde{\lambda} (u(\cdot, \delta; \tilde{\alpha}) - \tilde{V}_0).
$$
Consequently, the solution to the dual problem is observationally equivalent to a model where $\gamma = 0$ and $\alpha = \alpha' = \tilde{\alpha}$. Moreover, Lemma \ref{lem:alpha_equivalent_bounds} below establishes that the optimal solution is invariant to the specific choice of $\alpha$ for a given $\lambda$.

Variations in $\gamma, \alpha,$ and $\alpha'$ thus influence the experimental design only in an indirect manner, via their possible influence on the calibration of $V_0$, which in turn modifies the optimal choice of $\lambda$. Because our theoretical results already describe the optimal solution for the full range of $\lambda$ values, tracking these parameters separately does not yield any further explanatory insight. Consequently, for the rest of the analysis, we can, without loss of generality, impose $\gamma = 0$.

\subsubsection{Properties of the optimal stopping time}\label{subsubsec:properties_of_stopping}

Despite the functional form restrictions in Assumptions \ref{Asm-2} and \ref{Asm-3}, it is well known that a closed-form characterization of the optimal stopping time is intractable under a Gaussian prior (\citealp{Bather_Walker_1962}; \citealp{fudenberg-strack-strzalecki2018}). Instead, we follow the usual approach of identifying the optimal stopping rule with the solution to a boundary problem and use this to characterize its properties.

Consider the dual problem (\ref{eq:twotreatments}) with a fixed value of the multiplier $\lambda$. The optimal stopping time $\tau$ then solves:
\begin{equation}\label{eq:optimal_stopping_problem}
\sup_\tau \mathbb{E}_{\bm{d}}[B\mathds{1}[m_\tau \geq 0] + \lambda(\max\{\alpha m_\tau, -(1-\alpha)m_\tau\} - V_0)-c\tau].
\end{equation}
Because $\lambda$ is fixed, we omit the $V_0$ term in what follows; in Section \ref{sec:Calibration}, we show how $\lambda$ can be numerically calibrated given a level of $V_0$. By standard results in convex optimization, each strictly positive $\lambda$ corresponds to a unique $V_0$ value when the duality gap is 0 and the primal problem has a bounded value (both these conditions are verified in Theorem \ref{thm:Sampling_strategy}). 

Since the posterior distribution of $\mu_1 - \mu_0$ is Gaussian, it is uniquely characterized by its posterior mean and variance. Under the Neyman allocation, the evolution of the  posterior mean $m_t$ is described in Theorem \ref{thm:Sampling_strategy}, while the posterior variance is a deterministic function of time $t$. Consequently, the sufficient statistics for optimal stopping are $(t, m_t)$. Define 
$$
S_\alpha(x) = \max\{x,0\} - (1-\alpha)x,
$$
and let $V(t,m)$ denote the continuation value in state $(t,m)$ when $m$ is the current value of the posterior mean:
\begin{equation}\label{eq:contvalue}
V(t,m) = \sup_{\tau}\mathbb{E}_{\bm{d}}[B\mathds{1}[m_{t + \tau} \geq 0]+ \lambda S_\alpha(m_{t + \tau}) - c\tau |m_t = m].\end{equation}
As in \citet[Chapter 10]{oksendal2013stochastic}, the optimal stopping rule for Markov problems of this kind has the form 
\begin{equation}\label{eq:stopping-time-def}
\tau^* = \inf_{t\geq 0}\{B\mathds{1}[m_t\geq 0] + \lambda S_\alpha(m_t) \geq V(t,m_t)\}.
\end{equation}
Equivalently, we may define two time-dependent boundaries and take the optimal stopping rule, $\tau^* = \min\{\tau^+,\tau^-\}$, to be the first exit time from these boundaries: \begin{align}
    b^+(t) &= \inf \{m \ge 0: B + \lambda S_\alpha(m) \geq V(t,m)\}, \label{eq:def_of_b_t+} \\
    b^-(t) &= \sup \{m < 0: \lambda S_\alpha(m) \ge V(t,m)\}, \label{eq:def_of_b_t-}\\
    \tau^+ &= \inf \{t \ge 0: m_t \geq b^+(t)\}, \\
    \tau^-&= \inf \{t \ge 0: m_t \leq b^-(t)\}.
\end{align}

Recall that we chose to set $\gamma = 0$ without loss of generality. As referenced earlier, $b^+(t), b^-(t)$ are the same for any $\alpha \in [0,1]$. Thus, we need only express their properties under $\alpha = 1$ or $\alpha = \frac{1}{2}$, whichever is convenient.

\begin{lem}\label{lem:alpha_equivalent_bounds}
    Let $b^+(t;\alpha), b^-(t;\alpha)$ be the stopping boundaries given $\alpha \in [0,1]$. Then, \[b^+(t;\alpha) = b^+(t;1),\quad b^-(t;\alpha) = b^-(t; 1).\]
\end{lem} 

Lemma \ref{lem:alpha_equivalent_bounds} implies that allowing for asymmetric welfare under approval and rejection leaves the optimal experimental design unchanged for a given $\lambda$. It further allows us to extend the results of \cite{fudenberg-strack-strzalecki2018} beyond the welfare-regret objective. Intuitively, while $\alpha$ controls the `tilt' of the welfare function, this tilt does not modify the relative difference between the options' means, and the optimal boundaries thus remain unaffected.

Since the optimal $\lambda$ is uniquely determined by the value of $V_0$, it follows that, conditional on knowing $V_0$, one does not need to know $\alpha$. In the numerical analysis (Section \ref{sec:Calibration}), we calibrate $\lambda$ through $V_0$ so that it matches the expected welfare of a Randomized Controlled Trial (RCT) with a fixed sample size ($t = 1$). In fact, when $m_0 = 0$, we obtain the striking result that the welfare generated by the RCT---and therefore the welfare under the calibrated stopping rule---remains the same for every $\alpha \in [0,1]$.

We now characterize several additional properties of the stopping boundaries.

\begin{thm}\label{thm:optimal-stopping-rule}
    Under Assumptions 1-3, $(b^+(t),b^-(t))$ have the following properties:
     \begin{enumerate}[(i)]
        \item $b^+(t), b^-(t)$ are well defined, with $ b^+(t), \vert b^-(t)\vert < \infty$. \label{thm1wd}
        \item $b^+(t)$ and $|b^-(t)|$ are weakly decreasing in $t$, with $\lim_{t\to\infty}b^+(t) = \lim_{t\to\infty}b^-(t)=0$. \label{thm1lim}
        \item $b^+(t) \leq |b^-(t)|$. Also, whenever $B > 0$, $b^+(t) = 0\ \forall\ t\ge t^*$, where $t^*$ is such that $b^-(t^*) = -B/\lambda$.\label{thm1shift}
        \item $b^+(t), b^-(t)$ are continuous in $t$, with the latter being Lipschitz continuous. \label{thm1cont}
        \item $b^+(t)$ and $|b^-(t)|$ are weakly increasing in $\lambda$.\label{thm1lam}
        \item $b^+(t), b^-(t)$ are weakly decreasing in $B$ for a fixed $\lambda$.\label{thm1B}
        \item $b^-(t)$ is strictly less than 0. \label{thm1neg}
    \end{enumerate}
\end{thm}

Theorem \ref{thm:optimal-stopping-rule} is based on the idea that the stopping rule identifies the smallest treatment-effect magnitude at which stopping weakly dominates continuation. When $B = 0$, the theorem follows directly from \cite{fudenberg-strack-strzalecki2018}. But the case $B > 0$, which entails a discontinuous utility function, calls for substantial extensions and yields meaningfully different conclusions.

Parts \eqref{thm1wd} and \eqref{thm1cont} are technical results stating that the lower and upper stopping boundaries, $b^-(t)$ and $b^+(t)$,  are well defined and continuous (in fact, $b^-(t)$ is Lipschitz continuous). Part \eqref{thm1lim} further establishes that these boundaries converge monotonically to zero over time. The intuition is that the incremental value of new information shrinks as beliefs about the true mean become more precise. Since the sampling cost remains constant, there comes a point at which it is optimal to stop for any given belief. This pattern is driven by the Gaussian prior, which assigns little prior mass to large values of $s$. By contrast, \cite{adusumilli2025} shows that with a two-point prior, beliefs can remain as uninformative in the future as they are at present.

Part \eqref{thm1shift} shows that the stopping thresholds are asymmetric: in absolute value, the acceptance threshold is always lower than the rejection threshold. This asymmetry is driven by Bob's private benefit $B$, which gives him an incentive to continue the experiment even when early signals are weakly favorable, as he hopes for a possible reversion to a positive mean. Part \eqref{thm1shift} further implies that $b^+(t)$ eventually hits zero. Under standard Bayesian persuasion, which corresponds to a nonbinding welfare constraint ($\lambda = 0$), the approval threshold would be $b^+(t) = 0$ for all $t$, i.e., the experiment would stop whenever the posterior mean becomes nonnegative. By contrast, for sufficiently large $V_0$, Bob chooses $b^+(t) > 0$ for a finite time interval, though this boundary always lies strictly below the welfare-optimal one, $b^+(t) = -b^-(t)$, corresponding to $B = 0$. Ultimately, however, the Bayesian persuasion motive prevails from $t \ge t^*$ onward.

Since $\lambda$ is in one-to-one correspondence with the welfare constraint $V_0$, part \eqref{thm1lam} states that tighter welfare constraints necessarily push the boundaries outward.

Part \eqref{thm1B} is intuitive: for a fixed $\lambda$, a larger benefit $B$ makes waiting less appealing when beliefs are favorable $(m \ge 0)$ and more appealing when beliefs are unfavorable $(m < 0)$. It is important to note, however, that if we hold the welfare constraint $V_0$ constant, an increase in $B$ actually raises the corresponding value of $\lambda$. In combination with part \eqref{thm1lam}, this suggests that the overall effect of raising $B$ on $b^-(t)$ should be strictly negative: increasing the benefit always makes waiting more attractive. By contrast, the impact on $b^+(t)$ is ambiguous: a lower $\lambda$ tends to reduce $b^+(t)$, so in total $b^+(t)$ may either expand or contract, depending on which mechanism dominates.


\subsection{On tie-breaking}\label{subsec:tie-breaking}

Up to this point, our analysis has implicitly assumed that Alice breaks ties in favor of Bob's preferred action, $\delta = 1$, whenever she is indifferent, i.e., when $m_\tau = 0$. It turns out, however, that the specific tie-breaking rule Alice adopts is irrelevant. Because $m_t$ can be expressed as a time-changed Brownian motion, it satisfies the immediate-crossing property: whenever the process reaches $0$, it will cross it within an arbitrarily short time interval. Consequently, Bob can always ensure that Alice accepts the proposal as soon as $m_\tau = 0$ (since it implies that $m_t > 0$ almost immediately). 

In Section \ref{sec:local_asymptotics}, we develop finite-sample analogs of our optimal policies. In that context, to prevent any ambiguity, we impose that $m_\tau$ be strictly bounded away from 0 by a factor $\xi$. We then demonstrate that $\xi$ can be chosen sufficiently small so that Bob's overall welfare is only minimally affected.

\subsection{Lessons for experimental design}\label{subsec:Lessons_for_experiments}

In the previous sections, we have characterized how experimental design changes when there is incentive misalignment between the experimenter (Bob) and the regulator (Alice). We summarize some key takeaways below.

\subsubsection*{Neyman allocation is always optimal}
Incentive misalignment (or even alignment) does not alter the optimal sampling strategy. Both Alice and Bob share a common objective to learn the state of the world, $s$, as efficiently as possible relative to the cost of experimentation. Consequently, the allocation of subjects across treatment arms is determined by variance-minimization rather than divergence in preferences.

\subsubsection*{Incentive misalignment only arises from additive benefits}
The form of the optimal stopping rule is altered if and only if there exists a private benefit from approval, $B$, that Bob receives regardless of the size of the treatment effect ($\mu_1 - \mu_0$). An analysis of Section \ref{subsec:Utility_spec} and the proof of Lemma \ref{lem:alpha_equivalent_bounds} shows that this property is a structural feature of the decision problem, and not an artifact of Gaussian priors; it is a consequence of the players agreeing on a common prior. In the absence of such an additive benefit ($B=0$), the players' interests are aligned regarding the form of the optimal stopping time, regardless of their individual risk tolerances ($\alpha, \alpha'$) or benefits that are linear in $\vert \mu_1 - \mu_0\vert$. 

\subsubsection*{Asymmetry in decision thresholds}
A significant departure from standard sequential designs for clinical trials---such as those discussed in \cite{wassmer2016group}---is the emergence of asymmetric thresholds for acceptance and rejection. The rejection threshold is shifted downward in absolute terms relative to the approval threshold. This asymmetry reflects Bob's `pro-approval' bias: because he captures a fixed gain from approval, he is incentivized to persist with the experiment even when early evidence is underwhelming, as he hopes for a reversion toward a positive mean.

\subsubsection*{Terminal convergence to zero of approval thresholds}
In classical fixed-sample experiments, where the number of observations is characterized by time $t$, a treatment is approved if its mean effect exceeds a threshold of $1.96\sigma/\sqrt{t}$. In the commonly used group sequential designs of \cite{o1979multiple}, this threshold declines at a rate proportional to $1/t$. In our framework, the optimal threshold for approval declines even more aggressively, eventually reaching zero at some finite time $t^*$. This rapid convergence is driven by two factors. First, under Gaussian priors, a low posterior mean at a large $t$ suggests the true effect is likely near zero; thus, the potential `downside' from a Type I error diminishes. Second, lowering the threshold for late-stage stopping increases Bob's `option value' of approval, and incentivizes him to continue costly experimentation for a longer duration.

\subsection{Potential limitations}

Our analysis rests on a couple of important structural assumptions, which have economic content. 

\subsubsection*{Common prior} The first assumption is that Alice and Bob share a common prior. Our preferred view of this prior is that it represents an objective quantity that can be estimated from historical data. For instance, in its current guidance on Bayesian clinical trials, the FDA notes that one option is to estimate the prior for Phase 3 designs using data from Phase 2 studies of the same drug (\citealp[Section III.A]{us2026guidance}).

In the same document, the FDA also indicates that the prior may alternatively be derived from data obtained in earlier clinical trials. This aligns with the Empirical Bayes approach. The central assumption here is exchangeability: the treatment effect in the current trial is presumed to arise from the same distribution as the effects observed in earlier, comparable studies. In clinical research, the development of such an objective prior is supported by the existence of extensive meta-analytic repositories, most notably the Cochrane Database of Systematic Reviews (CDSR). The Empirical Bayes method is likely most appropriate for evaluating `follow-up' treatments or therapies within established drug classes, rather than for genuinely novel or `first-in-class' interventions.

More generally, pharmaceutical firms are legally obligated to disclose all pertinent information about a drug to the FDA before initiating a Phase 3 trial. Once this prior, together with the information used to construct it, becomes common knowledge, Alice and Bob can engage in a process of `communication' that will ultimately lead them to share a common prior, as implied by the classic Aumann agreement theorem (\citealp{aumann1976agreeing}). Indeed, FDA guidance (\citealp[Section V.D.1]{us2026guidance}) explicitly notes that companies must account for the time required for this prior alignment when planning and developing the trial.

\subsubsection*{Risk neutrality and linear utility} For the analysis of optimal stopping times, we employed a second fundamental assumption that the Bernoulli utility functions of Alice and Bob are linear in $|\mu_1 - \mu_0|$, conditional on $\delta$.

For Alice, this implies risk neutrality with respect to the magnitude of the treatment effect. This is a stance consistent with that of a utilitarian social planner, with the important caveat that we do not restrict Alice's risk preferences regarding the binary decision itself; she may still exhibit significant asymmetry in her weighting of Type I versus Type II errors.

For Bob, linear utility implies that the private returns from approval scale proportionally with the magnitude of the treatment effect. In practice, one might expect profits to be non-linear and convex in $\mu_1 - \mu_0$; indeed, a breakthrough treatment effect might yield exceptionally high returns due to market dominance or reduced competition. However, such non-linearities would exacerbate, rather than mitigate, the incentive misalignments identified in Section \ref{subsec:Lessons_for_experiments}. For example, if Bob derives increasing marginal utility from the treatment effect, his `pro-approval' bias would strengthen. This would lead to even more pronounced asymmetry between approval and rejection thresholds and an even faster convergence of the approval threshold toward zero. 

\section{Numerical Calibration for Clinical Trials}\label{sec:Calibration} 

Determining the optimal stopping time requires information on the marginal cost of experimentation $c$, the private fixed benefit from approval $B$, the treatment-effect dependent benefit $\gamma$, the prior variance $\varrho_0$, and the welfare lower bound $V_0$. In this section, we calibrate these parameters using evidence from the clinical trial literature and then use them to derive optimal stopping rules. This allows us to measure the percentage reduction in the expected sample size needed to reach the same expected welfare as that delivered by a standard Randomized Controlled Trial (RCT).

Up to this point, our analysis has been framed in continuous time, whereas in practical settings data are collected at discrete intervals. As outlined in Section \ref{sec:local_asymptotics}, the continuous-time framework serves as an approximation---under so-called `small-cost asymptotics' (\citealp{adusumilli2025})---to a discrete sampling scheme in which time $t$ corresponds to the number of observations divided by $n$, with $n \to \infty$. In this asymptotic setup, the continuous-time parameters $c, B, \gamma, \varrho_0$ emerge as renormalized versions of the underlying structural parameters:
\begin{equation}\label{eq:scaling}
c = \frac{C}{n^{3/2}}, \quad \frac{\gamma}{c} = \frac{\gamma_n}{n^{3/2} C},\quad \frac{B}{c} =  \frac{B_n}{nC},\quad \varrho_0 = \sqrt{n}\varrho_{0,n},
\end{equation}
where $C$ denotes the true per-observation cost, $B_n$ the private fixed benefit from regulatory approval, $\gamma_n$ the scaling factor on treatment-effect–adjusted benefits, and $\varrho_{0,n}$ the prior variance.

The scaling can be rationalized as follows: following \cite{adusumilli2025}, we interpret $n^{3/2}$ as the relevant population size for the drug. As argued in \cite{adusumilli2025}, welfare for such a population size is maximized when the design is tuned to detect treatment effects of order $1/\sqrt{n}$. The true sampling cost $C$ is evidently invariant in $n$. By contrast, treatment-adjusted benefits must be multiplied by the population size $n^{3/2}$ to obtain aggregate profits, implying $\gamma_n = O(n^{3/2})$. Simultaneously, to keep the private fixed benefit and the treatment-adjusted benefits on the same scale, we impose $B_n/\gamma_n = O(1/\sqrt{n})$. This is also intuitive: $B$ represents the pure profit from approval even if the drug has no therapeutic effect, and in realistic scenarios, this should constitute only a small share of the overall profits from approval. 

For ease of exposition, we assume throughout this section that $\sigma^2 := (\sigma_1^2 + \sigma_0^2) = 1$. Under this convention, $\varrho_{0}$ should be interpreted as the prior variance of the standardized treatment effect $\mu/\sigma$. While in practice $\sigma$ would be estimated from the data, all of our later results remain valid for arbitrary values of $\sigma$ if the $m_t$ terms appearing in this section are understood as $m_t/\sigma$.

\subsection{Calibration of $n$, $c$ and $B$}

\cite{chen2025investigating} use data from all published and unpublished Phase 3 clinical trials registered on ClinicalTrials.gov over the past 20 years to estimate a median sample size of $300$ participants for Phase 3 studies. We therefore take this value as our scaling factor $n$.

\cite{moore} report that the median cost per participant is approximately $C = \$41{,}000$.

\cite{tetenov2016} estimates that the value of an approved drug to pharmaceutical firms is approximately \$802 million. This figure, however, incorporates both the pure profit from regulatory approval and the benefits adjusted for treatment effects. As a result, the precise value of $B_n$ is not directly observable. Using the scaling argument outlined above, we regard it as reasonable to divide this total by $\sqrt{n}$, yielding an estimate of $B_n \approx \$46.3$ million. Under this approach, the pure benefit of approval amounts to roughly $5.5\%$ of total profits, a share we consider to be of a realistic order of magnitude.\footnote{This figure also aligns with the estimated value of approximately \$50 million for drugs in the preclinical phase, as reported by \cite{aryal2022valuing}.}

Combining these values yields a cost-to-benefit ratio of $c/B = 0.265$. The appendix outlines two alternative calibration approaches: one assuming $B_n = 0$, and another assuming $B_n =$ \$802 million, meaning that the entire profit is represented as a purely approval-based benefit (although we consider this second calibration to be implausible).

\subsection{Estimating the prior distribution}\label{subsec:Estimating_prior_parameters}

The Cochrane Database of Systematic Reviews (CDSR) is the principal repository for systematic reviews and meta-analyses in the clinical trials literature. \cite{vanzwet2021} use the Cochrane Database to estimate a prior distribution for scaled treatment effects ($\mu/\sigma$) in clinical trials via Empirical Bayes methods.\footnote{\cite{vanzwet2021} fit a four-component mixture of normal distributions, although one component receives only a negligible weight. We estimate $\varrho_0$ by taking $\varrho_0/n$ to be the variance of this distribution.} On the basis of their estimates, we set $\varrho_0 = 3.12^2 = 9.7344$.

\cite{vanzwet2021} contend that choosing a prior mean of $0$ for $\mu_1 - \mu_0$ is reasonable, as it treats the active treatment and control symmetrically. In fact, their empirical estimates also yield a prior mean that is close to zero.

\subsection{Calibration of the welfare lower bound}
A natural baseline choice for $V_0$ is the expected social welfare, $V_0^*$, that regulators could expect to obtain under a standard RCT. Recall that we normalize $t$ so that $t = 1$ matches the median sample size of typical clinical trials. With a $\mathcal{N}(0,\varrho_0)$ prior over $\mu_1 - \mu_0$, the posterior mean at $t = 1$ is distributed as $m_1 \sim \mathcal{N}(0, \varrho_0/(1+\varrho_0^{-1}))$. 

Observe that $V_0^* = E[S_\alpha(m_1)]$. Now, $E[S_\alpha(m_1)] = E[\max\{m_1, 0\}]$ does not depend on $\alpha$ since $E[m_1] = 0$. Substituting the estimated value of $\varrho_0$ thus yields
\[
   V_0^* = E[S_\alpha(m_1)] = \sqrt{\frac{\varrho_0}{2\pi(1+\varrho_0^{-1})}} \approx 1.1853.
\]
In what follows, we also vary $V_0$ by trying out different multiplicative factors of $V_0^*$.

\subsection{Results}\label{subsec:Results}

\begin{table}
    \centering
    \begin{tabular}{l l l}
        \toprule
        \textbf{Parameter} & \textbf{Value} & \textbf{Reference} \\
        \midrule
        Marginal cost ($C_n$)    & $\$41{,}000$        & \cite{moore} \\
        Private benefit ($B_n$)  & $\$46.3\,\text{million}$ & \cite{tetenov2016} \\
        Prior variance ($\varrho_0$) & $9.7344$              & \cite{vanzwet2021} \\
        Sample size ($n$)        & $300$               & \cite{chen2025investigating} \\
        \bottomrule
    \end{tabular}
    \caption{Key parameters used in the analysis}
    \label{tab:key-parameters}
\end{table}


Table \ref{tab:key-parameters} summarizes the values of the key parameters used in our analysis. A simple binary search pinpoints the value of $\lambda$ under which the optimal experimentation strategy achieves an ex-ante welfare of $V_0^*$. The resulting stopping rule is plotted in Figure \ref{fig:stopboundaries}. The stopping boundaries clearly satisfy the properties highlighted in Theorem \ref{thm:optimal-stopping-rule}. For comparison, we also plot the classical approval rule  (accept if the sample mean, scaled by $\sigma$, is above $1.96/\sqrt{t}$; see the discussion in Section \ref{subsec:Lessons_for_experiments}) as $b_{RCT}(t)$. As expected, $b^+(t)$ truncates to 0. In our calibration, the truncation occurs at around $t \approx 1.75$ (530 units); for the same sample size, $b_{RCT}(t)$ requires an effect size for approval at which the adaptive experiment would have ended much earlier.

\begin{figure}
    \centering
    \includegraphics[width=0.55\linewidth]{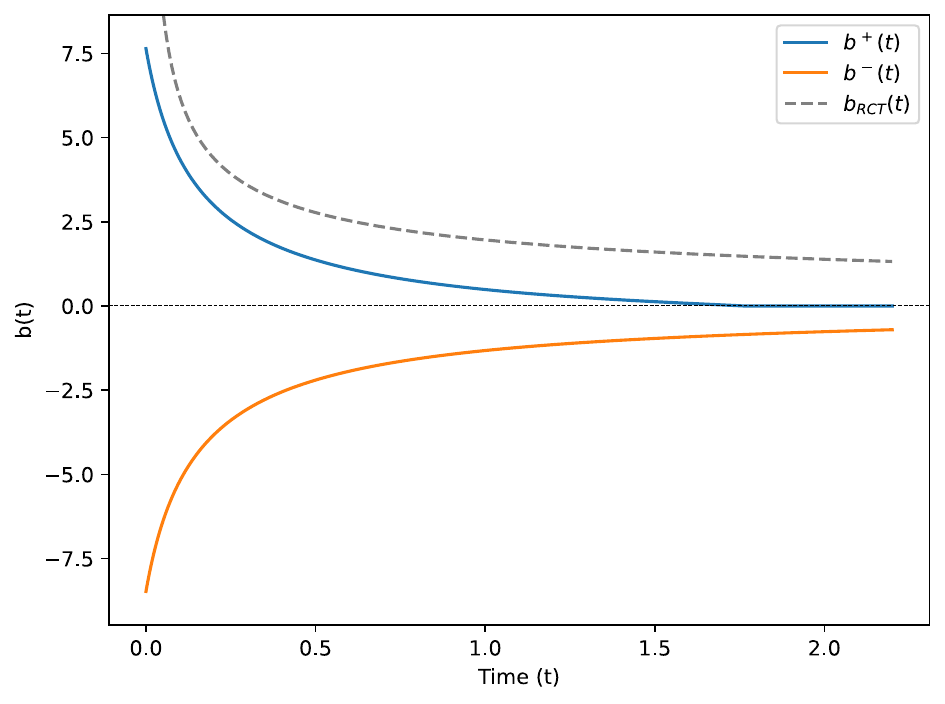}
    \caption{Stopping Boundaries $b^+(t), b^-(t)$.}
\begin{justify}
\scriptsize\textit{Notes: This figure displays the upper and lower stopping boundaries over time. The functions $b^+(t)$ and $b^-(t)$ characterize the optimal stopping region of the decision problem, where continuation is optimal for states between the two curves and stopping occurs once the process hits either boundary. For comparison, we also show the conventional approval cutoff at the 5\% significance level, assuming the sample size $t$ were exogenous.}
\end{justify}
    \label{fig:stopboundaries}
\end{figure}

In Figures \ref{fig:stoppingdist} and \ref{fig:postmeanstop}, we plot the distributions of $\tau$ and $m_\tau$, obtained by running 1 million simulations of the hypothetical adaptive experiment implied by the optimal experimentation strategy. The results show that the proposed strategy can reach the social welfare level $V_0^*$---the same level delivered by a conventional RCT---while using, on average, about $48\%$ fewer observations (i.e., $\mathbb{E}[\tau^*] \approx 0.52$). Under our parameter estimates $C_n = 41{,}000\$$ and $n = 300$, this implies expected cost savings of nearly 6 million dollars. In addition, approximately $86\%$ of the simulated paths stop before $t = 1$; recall that this corresponds to the median sample size in clinical trials. 

\begin{figure}
    \centering
    \begin{subfigure}[b]{0.45\linewidth}
        \centering
        \includegraphics[width=\linewidth]{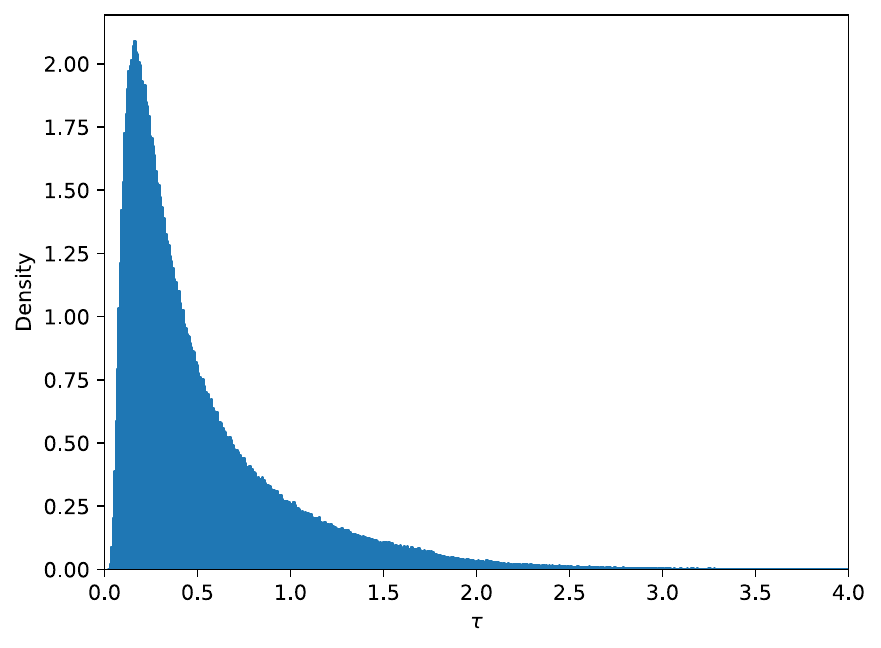}
        \caption{Distribution of $\tau$}
        \label{fig:stoppingdist}
    \end{subfigure}%
    \quad \quad
    \begin{subfigure}[b]{0.45\linewidth}
        \centering
        \includegraphics[width=\linewidth]{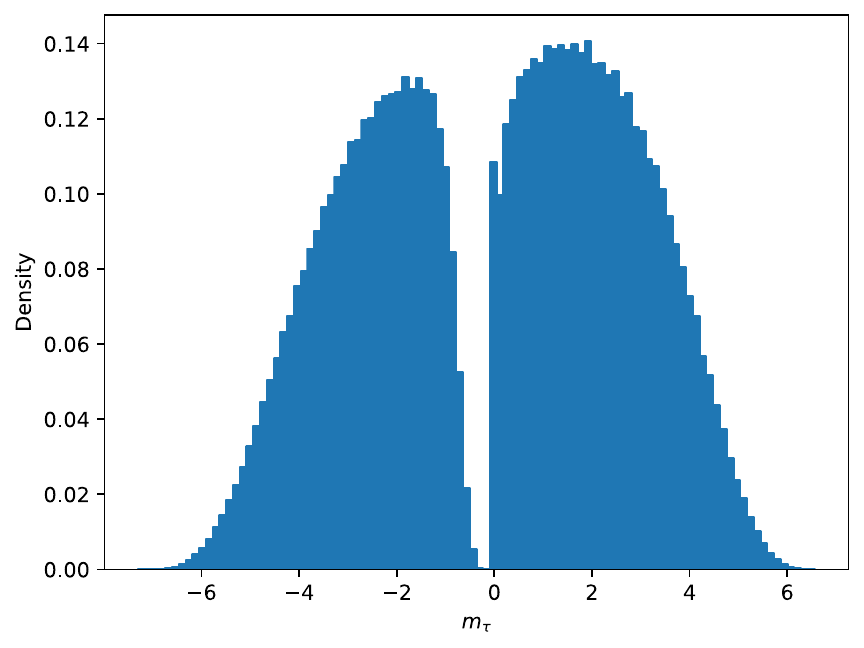}
        \caption{Distribution of $m_\tau$}
        \label{fig:postmeanstop}
    \end{subfigure}
    \caption{Distributions of $\tau$ and $m_\tau$.}
    \label{fig:stopping_and_postmean}
\end{figure}

The preceding analysis presumes that Bob appropriates the entire welfare surplus generated by more efficient experimentation. Alternatively, Alice might instead choose to increase social welfare beyond $V_0^*$. Figure \ref{fig:lambda_v_welfare} plots $V_0$ as a function of the multiplier $\lambda$, Figure \ref{fig:compare_bounds} shows how variations in $V_0$ affect the optimal boundaries, and Figures \ref{fig:welfare_v_meantime} and \ref{fig:welfare_v_medtime} present the expected and median stopping times over the same range of welfare values. The last two figures reveal that the mapping from $V_0$ to the expected sample size $\mathbb{E}[\tau]$ is highly non-linear. As discussed above, Alice can replicate the welfare level of a traditional RCT while using $48\%$ fewer observations. However, if she instead raises $V_0$ up to the point where Bob's average sample size matches that of a classical RCT (i.e., $\mathbb{E}[\tau^*] = 1$), this only yields a 3.2\% welfare gain relative to $V_0$. This welfare level is denoted by $V^*_1$ in Figure \ref{fig:welfare_v_meantime}. These findings appear to indicate that, under our calibration, pharmaceutical firms bear most of the inefficiency associated with relying on standard RCTs in clinical trial design: switching to more efficient approaches would significantly reduce firms' experimentation costs, while delivering only relatively small gains in overall social welfare.

Figure \ref{fig:calibrated_bounds_B} illustrates how the stopping boundaries vary with the benefit when $V_0$ is held constant. Reducing $B$ by half or increasing it two-fold has surprisingly little effect on the magnitudes of the stopping boundaries.

We describe additional comparative statics exercises in Appendix \ref{sec:comparativeplots}.

\begin{figure}
    \centering
     \begin{subfigure}[b]{0.45\linewidth}
        \centering
        \includegraphics[width=\linewidth]{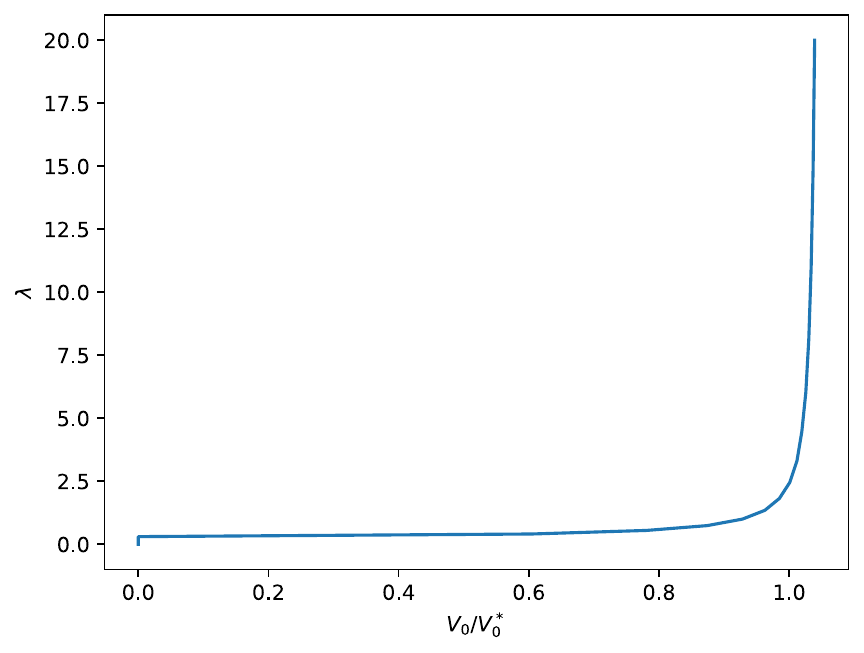}
        \caption{$\lambda$ vs. $V_0$}
        \label{fig:lambda_v_welfare}
    \end{subfigure}%
    \begin{subfigure}[b]{0.45\linewidth}
        \centering
        \includegraphics[width=\linewidth]{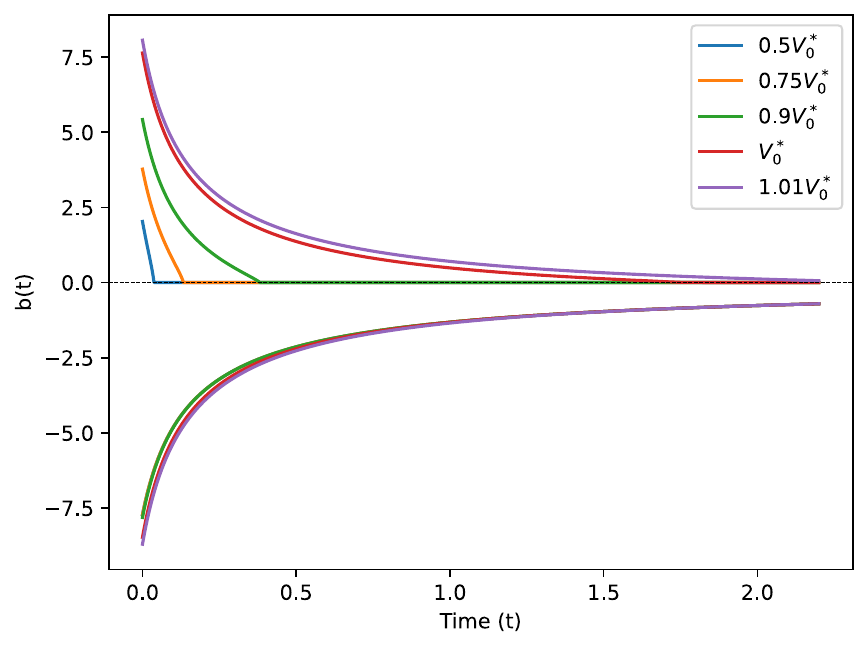}
        \caption{Stopping Boundaries for Various $V_0$}
        \label{fig:compare_bounds}
    \end{subfigure}   
    \begin{subfigure}[b]{0.45\linewidth}
        \centering
        \includegraphics[width=\linewidth]{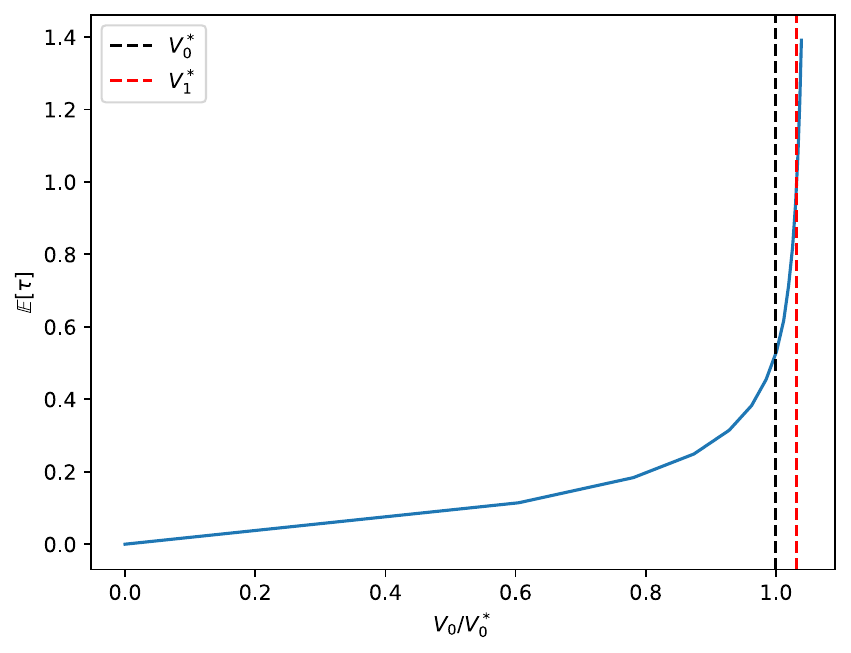}
        \caption{$\mathbb{E}[\tau^*]$ vs. $V_0/V^*_0$}
        \label{fig:welfare_v_meantime}
    \end{subfigure}%
    \begin{subfigure}[b]{0.45\linewidth}
        \centering
        \includegraphics[width=\linewidth]{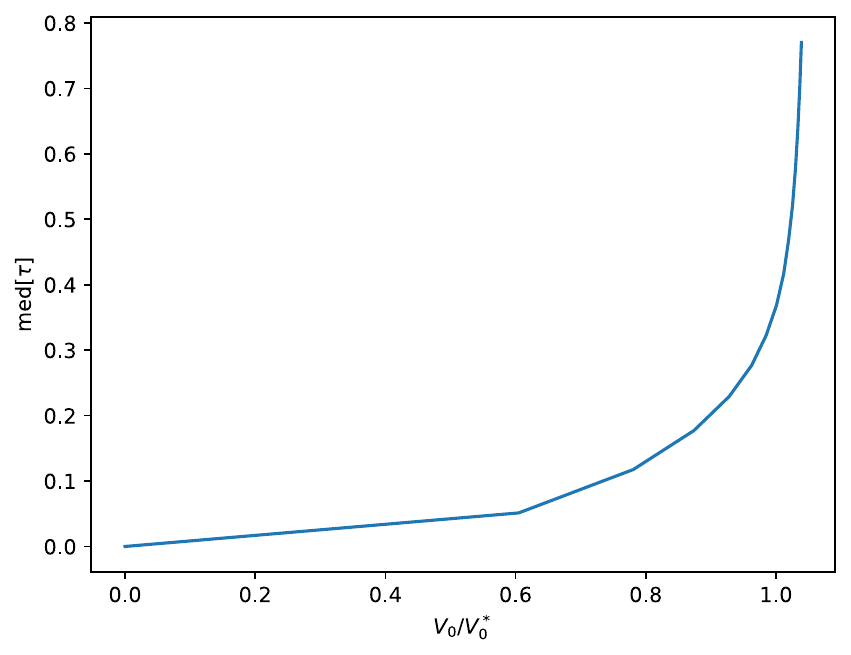}
        \caption{$\text{med}[\tau^*]$ vs. $V_0/V_0^*$}
        \label{fig:welfare_v_medtime}
    \end{subfigure}
    \caption{The effects of changing $V_0$.}
\end{figure}

\begin{figure}[h]
    \centering
    \includegraphics[width=0.55\linewidth]{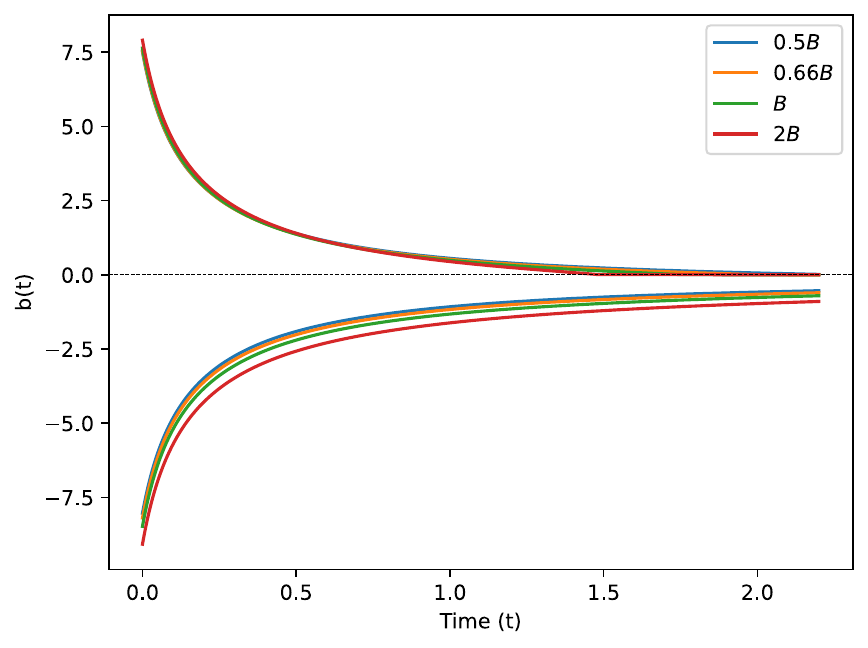}
    \caption{Stopping boundaries for various $B$ values.}
    \label{fig:calibrated_bounds_B}
\end{figure}

\section{Local Asymptotics and Parametric Models}\label{sec:local_asymptotics}

Up to this point, we have studied the optimal experimental design problem under incremental learning (in continuous time). In this section, we demonstrate that the continuous-time formulation can be viewed as an approximation to the corresponding experimental design problem in discrete time, under so-called `small-cost asymptotics'.

The central idea, following \cite{adusumilli2025}, is to consider a setting in which Bob and Alice care about detecting treatment effect differences on the order of $1/\sqrt{n}$ because the population size is of the order $O(n^{3/2})$. This motivates scaling time $t$ in units of $n$ observations, and employing the scaling for $B_n, \gamma_n$ described in (\ref{eq:scaling}). For simplicity, in what follows, we rescale units so that $B = 1$.  

\subsection{Experimental design in discrete time}
As before, Bob runs an experiment aimed at persuading Alice to take a particular action. The new feature is that the experiment now unfolds in discrete time.

Following \cite{lattimore}, we adopt the “stack-of-rewards” representation to formalize experimental design. Concretely, we imagine that there exists an infinite stack of observations \(\bm{y}_{a} := \{Y_{i}^{(a)}\}_{i=1}^{\infty}\) for each treatment $a$, generated at the outset as i.i.d draws from a parametric family \(\{P_{\theta^{(a)}}^{(a)}\}\), where \(\theta^{(a)} \in \mathbb{R}^d\) is unknown. Initially, neither Alice nor Bob observes any of these realizations. Each time a treatment is selected, we can imagine that Bob sees the current top element of the corresponding treatment stack; this element is then taken out of consideration.

Bob's sampling strategy is described by a policy \(\{\pi_{n,j}\}_{j} \equiv \{\pi_{n,\lfloor nt\rfloor}\}_{t}\), which specifies, at each period \(j\), the probability of assigning that observation to treatment 1 as a function of past data. The treatment assignment is then given by \(A_{j} \sim \text{Bernoulli}(\pi_{n,j})\). Define
\[
q_{n,a}(t) := \frac{1}{n}\sum_{j=1}^{\lfloor nt\rfloor} \mathbb{I}\{A_{j} = a\}
\]
as the number of observations that have been allocated to treatment \(a\) by time \(t\), normalized by $n$. We refer to \(\{q_{n,a}(\cdot)\}_{a}\) as the empirical allocation process. The policy rule \(\{\pi_{n,j}\}_{j}\) can be viewed as a function that takes the stacks \(\ensuremath{(\bm{y}_{1}, \bm{y}_{0})}\) and an exogenous random variable \(U \sim \text{Uniform}[0,1]\) as inputs and produces the realized trajectory of the empirical allocation process \(\{q_{n,a}(\cdot)\}_{a}\). The sampling strategy can thus be characterized by \(\{q_{n,a}(\cdot)\}_{a}\) instead of \(\{\pi_{n,j}\}_{j}\).

Let $\mathcal{F}_t^{\bm{q}_n}$ denote the $\sigma$-algebra generated by
\[
\xi_t \coloneq \left\{U, \{A_j\}_{j=1}^{\lfloor nt\rfloor},\{Y_i^{(1)}\}_{i=1}^{\lfloor nq_{n,1}(t)\rfloor}, \{Y_i^{(0)}\}_{i=1}^{\lfloor nq_{n,0}(t)\rfloor}\right\},
\]
that is, the sequence of actions and realized rewards up to time $\lfloor nt\rfloor$. In addition to the sampling strategy, Bob chooses a stopping time $\tau_n$ that is $\mathcal{F}_{t}^{\bm{q}_n}$-adapted. 

Overall, Bob's decision rule $\bm{d}_n$ consists of the combination of the empirical allocation process $\bm{q}_n(\cdot)$ and the stopping time $\tau_n$. After the experiment concludes, Alice makes a binary decision---either approval $(\delta_n = 1)$ or rejection $(\delta_n = 0)$---with the requirement that $\delta_n$ be $\mathcal{F}_{\tau_n}^{\bm{q}_n}$-measurable.

\subsection{Local asymptotics}

Following \cite{hirano2025asymptotic}, we assess decisions under local perturbations of the form $\{\theta_{0}^{(a)} + h_a/\sqrt{n} : h_a \in \mathbb{R}^{d}\}$, where $\theta_{0}^{(a)}$ is a reference parameter. Let the mean outcomes be defined by $\mu_a(\theta) \coloneq \mathbb{E}_{P_\theta^{(a)}}[Y_i^{(a)}]$. The reference parameters are selected so that $\mu_1\bigl(\theta_0^{(1)}\bigr) - \mu_0\bigl(\theta_0^{(0)}\bigr) = 0$. This setup is motivated by the consideration that Alice and Bob generally must be able to discriminate between treatment effects of order $1/\sqrt{n}$ in order to justify running experiments with sample size on the order of $n$.

To simplify notation, we normalize the means so that $\mu_1\bigl(\theta_0^{(1)}\bigr) = \mu_0\bigl(\theta_0^{(0)}\bigr) = 0$. Consequently, under local perturbations, we have
\begin{equation}\label{eq:Taylor_expansion_of_mean}
\mu_{n,a}(h) = \mu_a(\theta_0^{(a)} + h/\sqrt{n}) \approx \dot{\mu}_a^{\intercal}h/\sqrt{n},
\end{equation}
where $\dot\mu_a \coloneq \nabla_{\theta}\mu_a(\theta_0^{(a)})$.

Let $\nu$  denote a dominating measure for $\{P_{\theta}^{(a)}:\theta\in\mathbb{R}^{d},a\in\{0,1\}\}$, and set $p_{\theta}^{(a)}\coloneq dP_{\theta}^{(a)}/d\nu$ . We require $\{P_{\theta}^{(a)}\}_{\theta}$ to be quadratic mean differentiable (qmd): 

\begin{asm}\label{asm:qmd}
The class $\{P_\theta^{(a)}:\theta \in \mathbb{R}^d\}$ is differentiable in quadratic mean around $\theta^{(a)}_0$ for each $a \in \{0,1\}$, i.e., there exists a score function $\psi_a(\cdot)$ such that for each $h_a \in \mathbb{R}^d$, \[\int\left[\sqrt{p^{(a)}_{\theta^{(a)}_0 + h_a}} - \sqrt{p^{(a)}_{\theta_0^{(a)}}} - \frac{1}{2}h_a^{\intercal}\psi_a\sqrt{p^{(a)}_{\theta_0^{(a)}}}\right]^2d\nu = o\left(|h_a|^2\right).\]Furthermore, the information matrix $I_a \coloneq \mathbb{E}_0[\psi_a\psi_a^\intercal]$ is invertible for $a \in \{0,1\}$. 
\end{asm}

This assumption is fairly weak and holds for nearly all standard distributions, such as the Normal, Cauchy, Exponential, and Poisson distributions.

In what follows, define $P^{(a)}_h \coloneq P^{(a)}_{\theta_0^{(a)}+h/\sqrt{n}}$ and let $\mathbb{E}_{h_a}[\cdot]$ denote the corresponding expectation. Furthermore, let $P^{(a)}_{n,h}$ be the joint distribution of $\bm{y}_{n}^{(a)} = \left\{Y_1^{(a)},Y_2^{(a)}, \ldots\right\}$, where $Y_i^{(a)} \sim \text{i.i.d. } P^{(a)}_h$. We also set $\textbf{\textit{h}} \coloneq (h_1,h_0)$ and define 
$$
P_{n,\textbf{\textit{h}}} := P^{(1)}_{n,h_1} \times P^{(0)}_{n,h_0},
$$ 
with $\mathbb{E}_{n,\textbf{\textit{h}}}[\cdot]$ as its associated expectation.

For each treatment \(a\), the score process is defined as
\[
X_{n,a}(t):=\frac{I_{a}^{-1/2}}{\sqrt{n}}\sum_{i=1}^{\left\lfloor nq_{n,a}(t)\right\rfloor }\psi_{a}(Y_{i}^{(a)}).
\]
\cite{adusumilli-continuous-arts} demonstrates that the sample paths of \(\{X_{n,a}(\cdot),q_{n,a}(\cdot)\}\) up to time \(t\) form asymptotically sufficient statistics for the adaptive experiment. 

\subsection{Payoffs, priors and welfare}\label{subsec:Payoffs_and_welfare}

Let $\mu_n(\bm{h}) = \mu_{n,1}(h_1) - \mu_{n,0}(h_0)$. We assume that Alice and Bob's Bernoulli utility functions (exclusive of sampling costs) take the form
$$
\sqrt{n} u\left(\mu_n(\bm{h}),\delta_n; \alpha \right),\quad B_n \delta_n + \gamma_n  u\left(\mu_n(\bm{h}),\delta_n; \alpha^\prime \right),
$$
where the function $u(\cdot)$ is specified in Section \ref{subsec:Utility_spec}. To understand the scaling of Alice's welfare, recall from (\ref{eq:Taylor_expansion_of_mean}) that $\mu_{n,0}(h_0) = O(1/\sqrt{n})$, so we effectively scale the treatment effects by $\sqrt{n}$ to prevent them from vanishing asymptotically.

Observe that $\sqrt{n} u\left(\mu_n(\bm{h}),\delta_n; \alpha \right) = u\left(\sqrt{n}\mu_n(\bm{h}),\delta_n; \alpha \right)$, so for a given $\bm{h}$, Alice's expected utility under the strategy pair $(\bm{d}_n, \delta_n)$ is given by
\begin{align}
W_n^A((\bm{d}_n, \delta_n), \bm{h}) &\coloneq \mathbb{E}_{n,\bm{h}}\left [u(\sqrt{n}\mu_n(\bm{h}),\delta_n;\alpha) \right].
\end{align}
For Bob's expected utility under a given $\bm{h}$, we scale by the factor $1/B_n$ and define
\begin{equation}\label{eq:Bobs-payoff-finite-n}
\begin{aligned}
W_n^B((\bm{d}_n, \delta_n), \bm{h}) &:= \mathbb{E}_{n,\bm{h}}[\delta_n] + \frac{\gamma_n}{B_n}  \mathbb{E}_{n,\bm{h}}\left [u(\mu_n(\bm{h}),\delta_n;\alpha^\prime) \right]- \frac{C_n}{B_n} n \mathbb{E}_{n,\bm{h}}[\tau_n] \\
&= \mathbb{E}_{n,\bm{h}}[\delta_n] + \gamma \mathbb{E}_{n,\bm{h}}\left [u(\sqrt{n}\mu_n(\bm{h}),\delta_n;\alpha^\prime) \right] - c \mathbb{E}_{n,\bm{h}}[\tau_n],
\end{aligned}
\end{equation}
where the second line uses the asymptotic scaling regime from (\ref{eq:scaling}).

\subsubsection{Prior choice}\label{subsubsec:Prior_choice}
We assume that Alice and Bob share a common Gaussian prior $\Gamma_0(\bm{h})$ on the local parameter $\bm{h}$. This prior induces a corresponding prior $p_0$ on the scaled treatment effects $(\dot{\mu}_1^\intercal h_1, \dot{\mu}_0^\intercal h_0)$. We further assume that $\Gamma_0(\bm{h})$ factorizes into the prior $p_0$ on $(\dot{\mu}_1^\intercal h_1, \dot{\mu}_0^\intercal h_0)$ and an independent prior $\tilde{\Gamma}_0$ on the remaining components of $\bm{h}$. The multiplicative separability of $p_0$ and $\tilde{\Gamma}_0$ can be justified by an invariance requirement: inference about $(\dot{\mu}_1^\intercal h_1, \dot{\mu}_0^\intercal h_0)$ should not depend on the values of `nuisance parameters' associated with the other components of $\bm{h}$.

For our asymptotic framework, we additionally assume that $\Gamma_0$ remains fixed as $n$ increases. As highlighted in \cite{adusumilli2025bandits}, such local priors provide a more accurate description of the asymptotic behavior of decisions because their influence does not vanish with growing sample size.

\subsubsection{Finite sample welfare}
Let $\delta_n^*$ denote Alice's optimal strategy under the given prior, given Bob's choice of $\bm{d}_n$. Alice requires that social welfare exceed a predetermined level $V_0$, up to a slackness term $\epsilon_n \to 0$. Formally, her requirement is
$$
\int W_n^A((\bm{d}_n, \delta_n^*), \bm{h}) \, d\Gamma_0(\bm{h}) \ge V_0 - \epsilon_n.
$$
Introducing an arbitrarily small relaxation of the welfare constraint makes it less stringent, allowing it to be met in the limit rather than exactly. Given Alice's choice, Bob's experimental design problem is to select $\bm{d}_n$ to maximize his own expected welfare, i.e., Bob solves
\begin{equation}\label{eq:Bobs_primal_fixed_n}
\begin{aligned}
\bar{W}_n^* := \ & \sup_{\bm{d}_n} \int W_n^B((\bm{d}_n, \delta_n^*), \bm{h}) \, d\Gamma_0(\bm{h}) \\
& \text{s.t. } \int W_n^A((\bm{d}_n, \delta_n^*), \bm{h}) \, d\Gamma_0(\bm{h}) \ge V_0 - \epsilon_n.
\end{aligned}
\end{equation}

\subsection{Limit approximations and upper bounds on welfare}\label{subsec:limit_approximations}

Consider a limit experiment in which the underlying informational environment consists of signal processes of the form
\begin{equation}\label{eq:limit_experiment_process}
Z_a(\gamma) = I_a^{1/2}h_a \gamma + \bar{W}_a(\gamma),  
\end{equation}
where $\bar{W}_1(\cdot), \bar{W}_0(\cdot)$ are independent $d$-dimensional Brownian motions. As before, in this limit experiment, Bob chooses an experimental strategy $\bm{d}$---consisting of an allocation strategy $\{q_a(\cdot)\}_a$ and a stopping time $\tau$---while Alice makes a binary decision $\delta$. 

Suppose further that Alice and Bob's expected utilities in the limit experiment, under a strategy pair $(\bm{d}, \delta)$, take the form
\begin{equation}
\begin{aligned}
   W^A((\bm{d}, \delta), \bm{h}) &\coloneq \mathbb{E}_{\bm{h}}\left [u(\mu(\bm{h}),\delta;\alpha) \right]; \textrm{ and}\\
   W^B((\bm{d}, \delta), \bm{h}) &\coloneq \mathbb{E}_{\bm{h}}[\delta] + \gamma \mathbb{E}_{\bm{h}}\left [u(\mu(\bm{h}),\delta;\alpha) \right] - c \mathbb{E}_{\bm{h}}[\tau],
\end{aligned}
\end{equation}
where $\mu(\bm{h}) := \dot{\mu}_1^\intercal h_1 - \dot{\mu}_0^\intercal h_0$. Based on the above, we can write Bob's experimental design problem in this limit-experiment as:
\begin{equation} \label{eq:Bobs_primal_limit}
\begin{aligned}
\bar{W}^* := \ & \sup_{\bm{d}} \int W^B((\bm{d}, \delta^*_{\bm{d}}), \bm{h}) \, d\Gamma_0(\bm{h}) \\
& \text{s.t. } \int W^A((\bm{d}, \delta^*_{\bm{d}}), \bm{h}) \, d\Gamma_0(\bm{h}) \ge V_0.
\end{aligned}
\end{equation}

Observe that $W^A((\bm{d}, \delta), \bm{h})$ and $W^B((\bm{d}, \delta), \bm{h})$ depend on $\bm{h}$ solely through the terms $\dot{\mu}_1^\intercal h_1$ and $\dot{\mu}_0^\intercal h_0$. Together with the assumption that we restrict attention to multiplicatively separable priors $\Gamma_0$ (as specified in Section \ref{subsubsec:Prior_choice}), this implies that the signal processes $\dot{\mu}_1^\intercal I_1^{-1/2}Z_1(\cdot)$ and $\dot{\mu}_0^\intercal I_0^{-1/2} Z_0(\cdot)$ constitute sufficient statistics for the limit experiment. Consequently, we can directly relate this experiment to the one in Section \ref{Sec:Incremental_learning} by equating $\mu_a$ with $\dot{\mu}_a^\intercal h_a$, $z_a(\cdot)$ with $\dot{\mu}_a^\intercal I_a^{-1/2} Z_a(\cdot)$, and $\sigma_a^2$ with $\dot{\mu}_a^\intercal I_a^{-1} \dot{\mu}_a$. See Lemma \ref{lem:equivalence_of_limit_experiments} for a formal proof of the equivalence between these two experiments.

The asymptotic representation theorem of \cite{adusumilli-continuous-arts} enables us to match the expected welfare of Alice and Bob under any sequence of strategy pairs, $\{(\bm{d}_n, \delta_n)\}_n$, with that from a strategy pair, $(\bm{d},\delta)$, in the limit experiment described above. Here, we do not necessarily require $\delta_n, \delta$ to be Bayes-optimal relative to $\bm{d}_n, \bm{d}$; they just represent permissible strategies for Alice. The result on the matching of expected welfare relies on the following assumptions:

\begin{asm}\label{asm:finite_stopping_time}
    There exists $T < \infty$ such that $\tau_n \le T$ for all $n$. 
\end{asm}

\begin{asm}\label{asm:asm-on-mu}
    There exists $\dot\mu_a \in \mathbb{R}^d$ and $\varepsilon_n \to 0$ independent of $(a, \bm{h})$ such that $\sqrt{n}\mu_{n,a}(\bm{h}) = \dot\mu_a^\intercal h_a + \varepsilon_n|h_a|^2$.
\end{asm}

Assumption \ref{asm:finite_stopping_time} requires that the stopping times are bounded. While our analysis of the incremental learning regime allows for unbounded stopping times, it poses technical challenges for asymptotic approximations; therefore, we restrict our attention to stopping times bounded by some arbitrarily large, but finite, $T$. Assumption \ref{asm:asm-on-mu} is a mild requirement on the smoothness properties of $\mu_{n,a}(\textbf{\textit{h}})$.
 
\begin{thm}\label{thm:limit_experiment}
    Suppose Assumptions \ref{asm:qmd}, \ref{asm:finite_stopping_time}, and \ref{asm:asm-on-mu} hold. Then, for any sequence of strategy pairs $(\bm{d}_n, \delta_n)$ and corresponding welfare functions $W^A_n(\cdot, \bm{h}), W_n^B(\cdot, \bm{h})$, there exists a subsequence $(\bm{d}_{n_k}, \delta_{n_k})$ and a strategy pair $(\bm{d},\delta)$ in the limit experiment, with welfare functions $W^A(\cdot, \textbf{\textit{h}}), W^B(\cdot, \bm{h})$, such that $W^A_{n_k}(\cdot, \textbf{\textit{h}}) \to W^A(\cdot, \textbf{\textit{h}})$ and $W^B_{n_k}(\cdot, \textbf{\textit{h}}) \to W^B(\cdot, \textbf{\textit{h}})$ for each $\textbf{\textit{h}}$.
\end{thm}

The quantities $\bar{W}_n^*$ and $\bar{W}^*$, defined in (\ref{eq:Bobs_primal_fixed_n}) and (\ref{eq:Bobs_primal_limit}), denote Bob's optimal welfare in the finite-$n$ and limit settings, respectively. We next establish $\limsup_{n \to \infty} \bar{W}_n^* \le \bar{W}^*$, which shows that Bob's welfare in the limit experiment serves as an asymptotic upper bound for his welfare in the finite-$n$ case.

\begin{thm}\label{thm:asymptotic_upper_bound}
    Suppose Assumptions \ref{asm:qmd}, \ref{asm:finite_stopping_time}, and \ref{asm:asm-on-mu} hold, and the prior $\Gamma_0$ is Gaussian.\footnote{An inspection of the proof of this theorem reveals that the extra conditions imposed on the Gaussian prior $\Gamma_0$ in Section \ref{subsubsec:Prior_choice} are in fact not needed here.} Then, $\lim_{T \to \infty} \limsup_{n \to \infty} \bar{W}_n^* \le \bar{W}^*$. 
\end{thm}
 
\subsection{Attaining the upper bound} 

We now demonstrate that the upper bound stated in Theorem \ref{thm:asymptotic_upper_bound} can be attained using finite-sample counterparts of the optimal strategies derived in Section \ref{Sec:Incremental_learning}. 

In what follows, we assume that the normal prior $\Gamma_0(\bm{h})$ induces a distribution $p_0$ over $(\dot{\mu}_1^\intercal h_1,  \dot{\mu}_0^\intercal h_0)$ for which Assumption \ref{asm-1}(iv) is satisfied. Recall that under this condition, the optimal sampling strategy in the limit experiment reduces to the Neyman allocation in every period.

Let $\Sigma_{aa}$ denote the prior variance of $\dot{\mu}_a^\intercal h_a$ and define
$$
\begin{aligned}
\mu_{n,a}(t) &:= \frac{\dot{\mu}_a^\intercal I_a^{-1/2}  X_{n,a}(t) + \Sigma_{aa}^{-2}\mu_0^{(a)}}{\sigma_a^{-2} q_{n,a}(t) + \Sigma_{aa}^{-2}}, \\
m_{n,t} &:= \mu_{n,1}(t) - \mu_{n,0}(t).
\end{aligned}
$$
In this notation, $m_{n,t}$ is the sample analogue of $m_t$, the posterior average treatment effect in the limit experiment. Furthermore, let $b^+(t; \lambda)$ and $b^-(t; \lambda)$ be the optimal stopping boundaries defined in Theorem \ref{thm:optimal-stopping-rule}, where we now make their dependence on $\lambda$ explicit.

Take $\lambda^*$ to be the Lagrange multiplier corresponding to $V_0$ in the limit experiment. For some $T < \infty$ and $\xi > 0$, we construct finite sample analogs of $\bm{d}^*$ as $\bm{d}^*_{n,T,\xi} = (\pi^*_n, \tau^*_{n,T,\xi})$, where 
\begin{equation} \label{eq:asymptotically_optimal_strategy}
\begin{aligned}
    \pi^*_{n,a}(t) &= \mathds{1}\left[q_{n,a}(t) \leq \frac{\sigma_a}{\sigma_0 + \sigma_1}t\right], \textrm{ and}\\
    \tau^*_{n,T,\xi} &= \inf_{t \geq 0} \left\{m_{n,t} \not\in (b^-(t; \lambda^*), b^+(t;\lambda^*) + \xi)\right\} \wedge T,
\end{aligned}    
\end{equation}
and it may be recalled from Section \ref{subsec:limit_approximations} that $\sigma_a^2 := \dot{\mu}_a^\intercal I_a^{-1} \dot{\mu}_a$. The policy rule $\pi^*_{n,a}(t)$ realizes the Neyman allocation $q_{n,a}^*(t) = \sigma_a t /(\sigma_1 + \sigma_0)$, up to an approximation error of order $O(1/(nt))$. The stopping time $\tau^*_{n,T,\xi}$ is a finite-sample counterpart of $\tau^*$ in which $m_{n,t}$ substitutes for $m_t$, but with two further modifications: we cap the stopping time at an arbitrarily large horizon $T$, and we enlarge the upper boundary by an arbitrarily small $\xi > 0$. The latter adjustment ensures that Alice never enters her indifference region, which corresponds to $b^+(t) = 0$. From a formal standpoint, this is required to eliminate discontinuity issues when proving convergence of welfare. Practically, it also means that the exact tie-breaking rule Alice uses within her indifference region is no longer relevant.

The theorem below shows that, for sufficiently large $T$ and sufficiently small $\xi$, the policy $\bm{d}_{n,T,\xi}^*$ asymptotically satisfies Alice's welfare constraint while also delivering Bob's asymptotically optimal welfare level $\bar{W}^*$. Therefore, this experimentation strategy is asymptotically optimal.

\begin{thm}\label{thm:optimal-strategy-discrete-time}
    Suppose Assumptions \ref{asm:qmd}, \ref{asm:finite_stopping_time}, and \ref{asm:asm-on-mu} hold. Further assume that the prior $\Gamma_0$ is Gaussian and can be factored into a prior $p_0$ over $(\dot{\mu}_1^\intercal h_1, \dot{\mu}_0^\intercal h_0)$, which fulfills Assumption \ref{asm-1}(iv), and an independent prior $\tilde{\Gamma}_0$ over the remaining components of $\bm{h}$. Then
    $$
       \lim_{T \to \infty} \lim_{\xi \to 0}\lim_{n \to \infty} \int W_n^A\left((\bm{d}_{n,T,\xi}^*, \delta_{n,T,\xi}^*), \bm{h} \right) d\Gamma_0(\bm{h}) \ge V_0,
    $$
    where $\delta_{n,T,\xi}^*$ denotes Alice's Bayes-optimal strategy given $\bm{d}_{n,T,\xi}^*$, and
    $$
        \lim_{T \to \infty} \lim_{\xi \to 0}\lim_{n\to\infty} \int W^B_n\left((\bm{d}^*_{n,T,\xi}, \delta^*_{n,T,\xi}), \bm{h} \right)d\Gamma_0(\bm{h}) = \bar W^*.
    $$
\end{thm}

In applications, we recommend choosing $\xi$ to be a small multiple of $\sigma$, e.g., $0.05\sigma$. Regarding $T$, we advise taking $T = \infty$ in practice, as the restriction to a finite $T$ in Theorems \ref{thm:asymptotic_upper_bound} and \ref{thm:optimal-strategy-discrete-time} is primarily to simplify the theory. 

\subsubsection{Unknown variances}

Up to this point, our analysis has taken $\sigma_1, \sigma_0$ as known. These objects are the information matrices evaluated at the reference parameters $\theta_{0}^{(1)}, \theta_{0}^{(0)}$. Conceptually, within the local asymptotic framework, these reference parameters are treated as if they are known beforehand. As emphasized in \cite{adusumilli2025}, although in applications one would ideally design procedures that either adapt to or are invariant with respect to these quantities, the local asymptotic framework itself cannot capture the effect of estimating them.

In principle, one could use a ‘forced exploration' phase (see, for example, \citealp[Chapter 33, Note 7]{lattimore}): for the first $\bar{n} = n^{c}$ observations, with $c \in (0,1)$, we set $\pi_{n,1}^{*}(t) = 1/2$. This is equivalent to applying the equal allocation strategy until time $\bar{t} = n^{c-1}$. The data collected in this preliminary stage are then used to obtain consistent estimators $\hat{\sigma}_{1}^{2}, \hat{\sigma}_{0}^{2}$ of $\sigma_{1}^{2}, \sigma_{0}^{2}$. From time $\bar{t}$ onward, we apply the asymptotically optimal experimentation rule $\bm{d}_{n,T,\xi}^*$, replacing $\sigma_{1}, \sigma_{0}$ with their estimates $\hat{\sigma}_{1}, \hat{\sigma}_{0}$. In practical applications, since clinical trials are typically conducted in groups, we suggest using $\pi_{n,1}^{*}(t) = 1/2$ for the first group of participants.

In the special case of Bernoulli outcomes, the local asymptotic structure forces $\sigma_1 = \sigma_0$, which in turn implies that $\pi_{n,1}^{*}(t) = 1/2$ throughout.

\subsection{Simulation with Bernoulli outcomes}

To study the finite-sample behavior of our proposed procedures, we simulate the stopping times under Bernoulli outcomes when using the asymptotically optimal strategy from (\ref{eq:asymptotically_optimal_strategy}).

Let $Y^{(a)} \sim \text{Bernoulli}(\theta^{(a)})$ represent the outcome under treatment $a$. As noted earlier, conducting a local asymptotic analysis with Bernoulli outcomes requires specifying reference parameters $\theta^{(1)}_0$, $\theta^{(0)}_0$ such that $\theta^{(1)}_0 = \theta^{(0)}_0 := \theta_0$. We then posit that each $\theta^{(a)}$ is drawn independently from a Gaussian prior, $\theta^{(a)} \sim \mathcal{N}(\theta_0, \sigma^2 \nu^2/n)$, where $\sigma^2 = 4\theta_0(1-\theta_0)$.\footnote{Recall that $\sigma^2 := (\sigma_1 + \sigma_0)^2$, and in the Bernoulli case $\sigma_1^2 = \sigma_0^2 = \theta_0(1-\theta_0)$.} In our simulations, we set $\theta_0 = 0.5$ and $\nu^2 = \varrho_0/2$, where $\varrho_0 \approx 9.734$ is the value estimated in Section \ref{subsec:Estimating_prior_parameters}. With this choice, we have $\sqrt{n}(\theta^{(1)} - \theta^{(0)})/\sigma \sim \mathcal{N}(0,\varrho_0)$, in agreement with the prior specification obtained in Section \ref{subsec:Estimating_prior_parameters}.

For Bernoulli outcomes, the score function takes the form $\psi_a(Y) = (\theta_0(1-\theta_0))^{-1}(Y^{(a)} - \theta_0)$, so the (normalized) score process can be written as 
$$
X_{n,a}(t):=\frac{1}{\sqrt{n\theta_0(1-\theta_0)}}\sum_{i=1}^{\left\lfloor nq_{n,a}(t)\right\rfloor } (Y_i^{(a)} - \theta_0).
$$
We then use this specification of $X_{n,a}(\cdot)$ in (\ref{eq:asymptotically_optimal_strategy}). Figure \ref{fig:bernoulli_n_v_welfare} illustrates how Alice and Bob's finite-sample welfare under the asymptotically optimal strategy converges as $n$ grows.\footnote{Due to numerical approximations used to compute the asymptotic welfare, Bob's finite-sample welfare appears slightly higher than the approximate asymptotic benchmark.}

\begin{figure}
    \centering
     \begin{subfigure}[b]{0.45\linewidth}
        \centering
        \includegraphics[width=\linewidth]{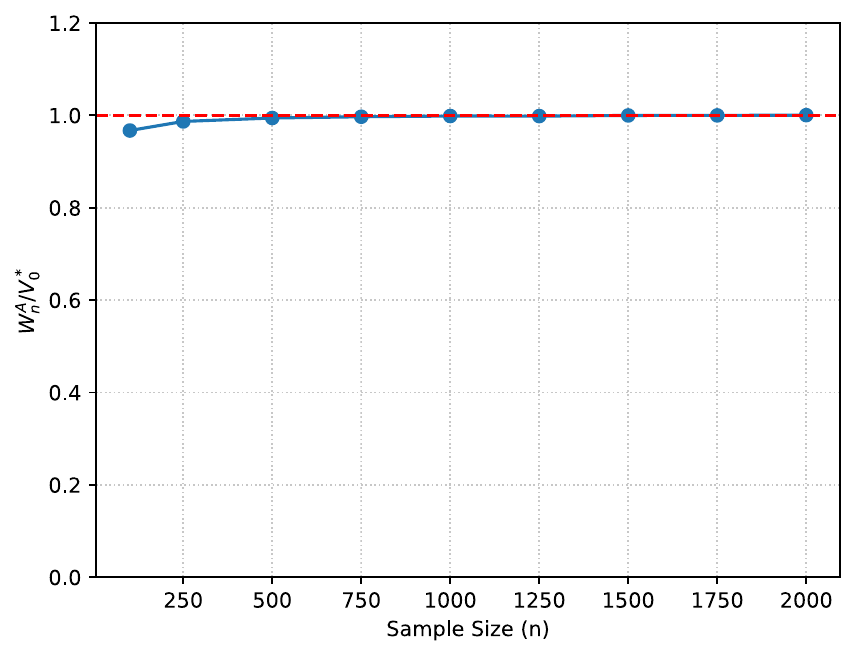}
        \caption{Alice}
    \end{subfigure}%
    \begin{subfigure}[b]{0.45\linewidth}
        \centering
        \includegraphics[width=\linewidth]{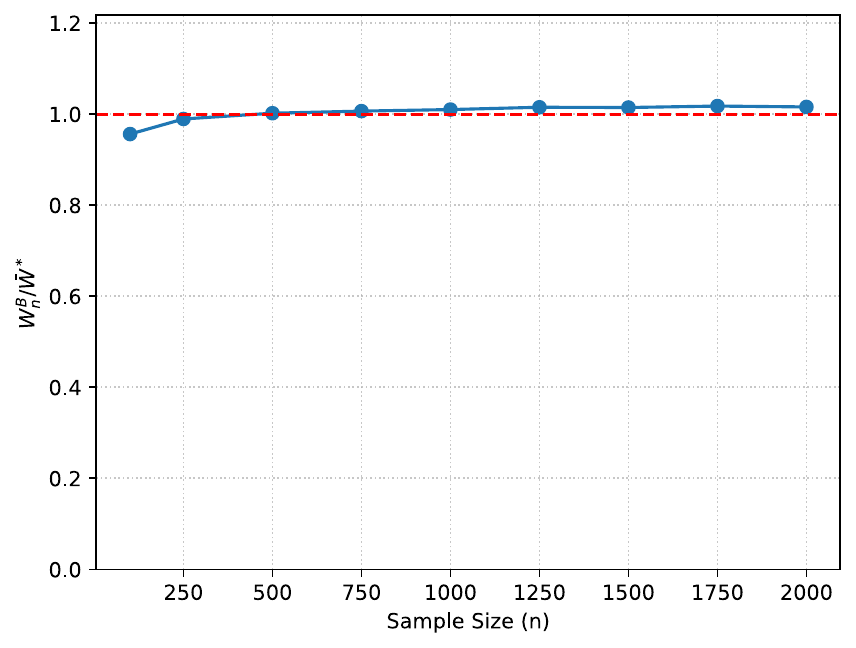}
        \caption{Bob}
    \end{subfigure}   
    \caption{Finite-sample welfare under Bernoulli outcomes.}
    \label{fig:bernoulli_n_v_welfare}
    \begin{justify}
    \scriptsize\textit{Notes: This figure depicts the finite-sample welfare of Alice and Bob under different values of $n$. Alice’s welfare is plotted relative to $V_0$, the asymptotic lower bound on welfare. Bob’s welfare is plotted relative to $\bar{W}_n^*$, the maximal welfare he can achieve in the limit experiment given $V_0$.}
    \end{justify}
\end{figure}

\section{Conclusion}

In this article, we propose that regulators directly target social welfare when regulating experimental designs. We characterize the optimal design in a continuous-time setting with two treatments and Gaussian priors over the mean effects of these treatments. Finally, we show that the optimal design in continuous time is asymptotically optimal under parametric outcome distributions and construct a finite sample analog of the optimal design that can be used in practical settings. 

\bibliography{bib}
\bibliographystyle{chicago}
\appendix
\section{Proofs of Theorem \ref{thm:Sampling_strategy} and Lemma  \ref{lem:alpha_equivalent_bounds}}

\subsection{Proof of Theorem \ref{thm:Sampling_strategy}}

We begin by reviewing some properties of randomized stopping times. An $\mathcal{F}^{\bm{q}}_t$-adapted randomized stopping time is a Markov kernel $\tau(\omega, dt)$ from the sample-path space $\Omega$ of $(\{z_a(\cdot)\}_a,U)$ to $[0,\infty)$, such that $\tau(\omega, [0,t])$ is $\mathcal{F}^{\bm{q}}_t$-measurable for all $t$. For any two randomized stopping times $\tau_0, \tau_1$ and a weight $\alpha \in [0,1]$, the mixture $\alpha\tau_1 + (1-\alpha)\tau_0$ is defined as the randomized stopping time that selects $\tau_1$ with probability $\alpha$ and $\tau_0$ with probability $1-\alpha$; equivalently, $\alpha\tau_1 + (1-\alpha)\tau_0$ can be considered a mixture of Markov kernels. A sequence of randomized stopping times, $\tau_n$, is said to converge to $\tau$ in the Baxter-Chacon topology if and only if
$$
\lim_{n \to \infty} \mathbb{E}[Y_{\tau_n}] = \mathbb{E}[Y_\tau]
$$
for all bounded, continuous, $\mathcal{F}^{\bm{q}}_t$-adapted processes $Y$. Let $\mathcal{T}_M \equiv \{\tau: \mathbb{E}[\tau] \le M\}$ denote the set of all randomized stopping times whose expectation is bounded by a constant $M$. \cite{baxter1977compactness} establish that $\mathcal{T}_M$ is compact under their namesake topology.

\subsubsection*{Step 1}
We start by proving parts (ii) and (iii) of Theorem \ref{thm:Sampling_strategy} for the dual problem. 

For part (ii), our proof strategy is based on Lemma 5 from \cite{liang-mu-syrgkanis-ecta-2022}. Let $\mu \coloneq \mu_1-\mu_0$. It turns out to be convenient to rewrite the dual problem (\ref{eq:twotreatments}) in terms of the policy invariant measure $\mathbb{P}$ from Section \ref{subsec:Incremental learning}, with $\mathbb{E}[\cdot]$ denoting its corresponding expectation. Define
$$
\begin{aligned}
\delta^{\bm{q}}_t &\coloneq \mathds{1}\left\{ \mathbb{E}\left[u_A(s,1) \vert \mathcal{F}_t^{\bm{q}} \right] \ge \mathbb{E}\left[u_A(s,0) \vert \mathcal{F}_t^{\bm{q}} \right] \right\}.
\end{aligned}
$$ We can then write the dual problem as
\[
    \inf_{\lambda \geq 0}  \sup_{q}\sup_{\tau}\mathbb{E}\left[\mathbb{E} \left[u_B(\mu, \delta^{\bm{q}}_\tau) + \lambda\{u_A(\mu,\delta^{\bm{q}}_\tau) - V_0\} - c\tau \vert \mathcal{F}_\tau^{\bm{q}} \right] \right].
\]
where the decision rule has been separated as $\bm{d} = (q, \tau)$. 

Let $\bm{q}^*$ represent the Neyman allocation, with the associated filtration denoted by $\mathcal{F}_t^* \coloneq \mathcal{F}_t^{\bm{q}^*}$. By construction (see, \citealt{liang-mu-syrgkanis-ecta-2022}), this strategy achieves the smallest possible posterior variance across all strategies at every time $t$. This minimized posterior variance is denoted as $\varrho_t^*$. Let $\varrho_0$ denote the prior variance and $\varrho_0 - \varrho_t^*$, the reduction in variance due to sampling. The latter is the same as the quadratic variation of $\mathbb{E}[\mu\vert \mathcal{F}_t^*]$. Hence, it follows from the Dambis-Dubins-Schwartz theorem that there exists a Brownian motion $(B^*_v)_{v\in[0,\varrho_0)}$ such that 
\[
    B^*_{\varrho_0-\varrho_t^*} = \mathbb{E}[\mu|\mathcal{F}_t^*].
\]
Let $T^*(\varrho)$ be the time $t$ such that $\varrho_t^* = \varrho_0 - \varrho$. Then, under $\bm{q}^*$, the value of the dual problem becomes
\begin{equation}\label{eq:dual_problem_in_variance_form}
\inf_{\lambda \geq 0} \sup_{\tau}\mathbb{E}\left[\mathbb{E} \left[u_B(\mu, \delta^{\bm{q}^*}_\tau) + \lambda\{u_A(\mu,\delta^{\bm{q}^*}_\tau) - V_0\} - cT^*(\varrho_0 - \varrho_\tau^*) \vert \mathcal{F}_\tau^{*} \right] \right].
\end{equation}
Now, as shown in Section \ref{subsec:Incremental learning}, the posterior distribution of $\mu$ given $\mathcal{F}_t^*$ is Gaussian and is therefore characterized solely by its posterior variance, $\varrho_t^*$, and posterior mean, $B^*_{\varrho_0 - \varrho_t^*}$. Hence, $\delta_\tau^{\bm{q}^*}$ is a function solely of $(B^*_{\varrho_0 - \varrho_t^*}, \varrho_t^*)$, and we can similarly write
$$
\mathbb{E}\left[u_B(\mu, \delta^{\bm{q}^*}_\tau) + \lambda\{u_A(\mu,\delta^{\bm{q}^*}_\tau) - V_0\} - cT^*(\varrho_0 -\varrho_\tau^*) \vert \mathcal{F}_\tau^{*} \right] = g(B^*_{\varrho_0 - \varrho_\tau^*}, \varrho_\tau^*; \lambda),
$$
for some measurable function $g(\cdot, \cdot;\lambda)$. Consequently, by standard time-change arguments (see, e.g., \citealt[Theorem I.2.2]{peskir2006optimal}), the optimal stopping problem in (\ref{eq:dual_problem_in_variance_form}) can be written as 
\begin{equation}\label{eq:time_change_equivalence}
\inf_{\lambda \ge 0} \sup_\tau \mathbb{E}[g(B^*_{\varrho_0 - \varrho_\tau^*}, \varrho_\tau^*;\lambda)] = \inf_{\lambda \ge 0} \sup_\rho \mathbb{E}[g(B_\rho^*, \rho;\lambda)],
\end{equation}
where $\rho$ is a stopping value adapted to the filtration generated by the sample paths of $(B^*_v)_{v\in[0,\varrho_0)}$. 

Let $\mathcal{F}_t^{\bm{q}}$ be the induced filtration under some other allocation rule $\bm{q} = \{q_a(\cdot)\}_a$. Applying the Dambis-Dubins-Schwartz theorem again, there exists a Brownian motion $(B_v)_{v \in [0,\varrho_0]}$ such that 
\[
    B_{\varrho_0-\varrho_t} = \mathbb{E}[\mu|\mathcal{F}_t^{\bm{q}}],
\]
where $\varrho_t$ is the posterior variance of $\mu$ under $\bm{q}$. We know that $\varrho_t \geq \varrho^*_t$. For any $t \geq 0$, we then have $t = T^*(\varrho_0 - \varrho_t^*) \geq T^*(\varrho_0 - \varrho_t)$. Hence, we can bound the dual value of allocation rule $\bm{q}$ from above by
\begin{align}\label{eq:upper_bound_payoff}
    \notag &\inf_{\lambda \geq 0} \sup_{\tau} \mathbb{E}\left[\mathbb{E} \left[u_B(\mu, \delta^{\bm{q}}_\tau) + \lambda\{u_A(\mu,\delta^{\bm{q}}_\tau) - V_0\} - c\tau \vert \mathcal{F}_\tau^{\bm{q}} \right] \right] \\
     & \le \inf_{\lambda \geq 0} \sup_{\tau} \mathbb{E}\left[\mathbb{E} \left[u_B(\mu, \delta^{\bm{q}}_\tau) + \lambda\{u_A(\mu,\delta^{\bm{q}}_\tau) - V_0\} - cT^*(\varrho_0 - \varrho_\tau) \vert \mathcal{F}_\tau^{\bm{q}} \right] \right].
\end{align}
The posterior distribution of $\mu$ given $\mathcal{F}_t^{\bm{q}}$ is the same as the posterior distribution of $\mu$ given $\mathcal{F}_t^*$ as long as the posterior mean and variances are the same. Hence, we can make another time change $\rho = \varrho_0 - \varrho_\tau$, and rewrite the upper bound in (\ref{eq:upper_bound_payoff}) as 
$$
\inf_{\lambda \ge 0} \sup_\rho \mathbb{E}[g(B_\rho, \rho;\lambda)],
$$
where $\rho$ is adapted to the filtration generated by the sample paths of $(B_v)_{v\in[0,\varrho_0)}$. This is the same as the right hand side of (\ref{eq:time_change_equivalence}) since $B_v^*, B_v$ are both standard Brownian motions. Hence, the payoff from using any sampling strategy is upper bounded by the payoff from using the Neyman allocation.

Part (iii) of the theorem follows from part (ii) and Lemma 1 of \cite{fudenberg-strack-strzalecki2018}.  

\subsubsection*{Step 2} Next, we demonstrate that the primal and dual problems lead to the same value under two restrictions: (1) the sampling strategy is fixed at the Neyman allocation, $q^*(t)$; and (2) the set of stopping times is restricted to $\mathcal{T}_M$ for some $M < \infty$. These restrictions will be relaxed in the next step. 

Under the above restrictions, the primal problem (\ref{eq:constrainedtwotreatments}) can be expressed as
\begin{equation}\label{eq:pf:Thm1}
\sup_{\tau \in \mathcal{T}_M} \inf_{\lambda \geq 0} \mathbb{E}\left[u_B(\mu, \delta^*_\tau) + \lambda \left\{ u_A(\mu, \delta^*_\tau) - V_0 \right\} - c\tau \right],
\end{equation}
where $\tau$ is $\mathcal{F}_t^*$-adapted, and
$$
\delta^*_t := \argmax_{\delta \in \{0,1\}} \mathbb{E}[u_A(\mu, \delta) \mid \mathcal{F}_t^*].
$$
The primal and dual problems would therefore attain the sample value if it is possible to interchange the sup and inf operations in (\ref{eq:pf:Thm1}). 

Define
$$
U_A(t) \coloneq \mathbb{E}[u_A(\mu, \delta^*_t) \mid \mathcal{F}_t^*], \ \textrm{and} \ \ 
U_B(t) \coloneq \mathbb{E}[u_B(\mu, \delta^*_t) \mid \mathcal{F}_t^*].
$$
Then, using the law of iterated expectations, (\ref{eq:pf:Thm1}) can be rewritten as
$
\sup_{\tau \in \mathcal{T}_M} \inf_{\lambda \ge 0} W(\tau, \lambda),
$
where
$$
W(\tau, \lambda) \coloneq \mathbb{E}[U_B(\tau) + \lambda(U_A(\tau) - V_0) - c\tau].
$$
The desired interchange between the sup and inf operations can be justified using Sion's minimax theorem \citep{sion1958general}, which implies
$$
\sup_{\tau \in \mathcal{T}_M} \inf_{\lambda \ge 0} W(\tau, \lambda) = \inf_{\lambda \ge 0} \sup_{\tau \in \mathcal{T}_M}  W(\tau, \lambda),
$$
provided $\mathcal{T}_M$ is a compact set, and $W(\tau, \lambda)$ is concave-convex and continuous in both arguments.

To verify these conditions, we first recall that $\mathcal{T}_M$ is compact under the Baxter-Chacon topology. Regarding convexity, since $W(\tau, \cdot)$ is linear in $\lambda$, it is evidently convex and continuous. Regarding concavity, the definition of the mixture stopping time $\alpha\tau_1 + (1-\alpha)\tau_0$ implies that for any $\tau_1, \tau_0 \in \mathcal{T}_M$ and $\alpha \in [0,1]$,
$$
W(\alpha\tau_1 + (1-\alpha)\tau_0, \lambda) = \alpha W(\tau_1, \lambda) + (1-\alpha)W(\tau_0, \lambda).
$$
Consequently, $W(\cdot, \lambda)$ is concave on $\mathcal{T}_M$. It remains to show that $W(\cdot, \lambda)$ is continuous. Let $\tau_n \in \mathcal{T}_M$ be a sequence of randomized stopping times converging to $\tau \in \mathcal{T}_M$ under the Baxter-Chacon topology. As noted earlier, this implies $\mathbb{E}[\tau_n] \to \mathbb{E}[\tau]$. We now show that $\mathbb{E}[U_A(\tau_n)]$ converges to $\mathbb{E}[U_A(\tau)]$. For any constant $C < \infty$, define
$$
U_A^c(t) \coloneq \mathbb{E}\left[\mathds{1}_{\{\vert \mu \vert \le C\}} u_A(\mu, \delta^*_t) \mid \mathcal{F}_t^* \right],\  \textrm{and }\
U_A^d(t) \coloneq \mathbb{E}\left[\mathds{1}_{\{\vert \mu\vert > C\}} u_A(\mu, \delta^*_t) \mid \mathcal{F}_t^* \right].
$$
We can then decompose 
\begin{equation} \label{pf:Thm1:2}
\left| \mathbb{E}[U_A(\tau_n)] - \mathbb{E}[U_A(\tau)] \right| 
\le \left| \mathbb{E}[U_A^c(\tau_n)] - \mathbb{E}[U_A^c(\tau)] \right| 
+ \mathbb{E}[\vert U_A^d(\tau_n)\vert] + \mathbb{E}[\vert U_A^d(\tau) \vert].
\end{equation}
By the law of iterated expectations, Assumption \ref{asm-1}(ii), and the dominated convergence theorem, for any $\epsilon >0$, we can choose $C$ sufficiently large such that
$$
\sup_{\tau \in \mathcal{T}_M} \mathbb{E}[\vert U_A^d(\tau)\vert] \le \mathbb{E}_{p_0}\left[ \mathds{1}_{\{\vert \mu\vert > C\}} \bar{u}_A(\mu) \right] \le \epsilon.
$$
Thus, the last two terms in (\ref{pf:Thm1:2}) are bounded by $\epsilon$. Regarding the first term, note that $U_A^c(\cdot)$ is a martingale adapted to the filtration $\mathcal{F}_t^*$ generated by the Neyman allocation. By standard stochastic filtering results, $\mathcal{F}_t^*$ is equivalent to an initial enlargement (by the exogenous noise $U$) of a Wiener process filtration generated by an innovation process; hence, $U_A^c(\cdot)$ possesses continuous sample paths.\footnote{By the Ito representation theorem, local-martingales adapted to filtrations generated by a Wiener process have continuous sample paths} Furthermore, it is bounded since $\sup_t U_A^c(t) \le \sup_{\vert \mu \vert \le C} \bar{u}_A(\mu) < \infty$, where the last inequality is due to Assumption \ref{asm-1}(ii). We can therefore apply the definition of Baxter-Chacon convergence to obtain
$$
\mathbb{E}[U_A^c(\tau_n)] - \mathbb{E}_{\pi^*}[U_A^c(\tau)] \to 0.
$$
Since $\epsilon$ is arbitrary, we conclude that $\mathbb{E}[U_A(\tau_n)] \to \mathbb{E}[U_A(\tau)]$. An analogous argument demonstrates that $\mathbb{E}[U_B(\tau_n)] \to \mathbb{E}[U_B(\tau)]$. Taken together, these results confirm $W(\cdot, \lambda)$ is continuous under the Baxter-Chacon topology. This completes the verification of the conditions for Sion's minimax theorem.

\subsubsection*{Step 3} 

To complete the proof, we relax the restrictions placed in Step 2 and show that the primal and dual problems still attain the same value. 

We start by arguing that it is without loss of generality to restrict the stopping times to the set $\mathcal{T}_M$ for some $M < \infty$. Indeed, by Assumption \ref{asm-1}(ii), $\sup_{\bm{d}} \mathbb{E}_{\bm{d}}[u_B(\mu, \delta^*_{\bm{d}})] \le \mathbb{E}_{p_0}[\bar{u}_B(s)] < \infty$. Combined with Assumption \ref{asm-1}(iii), this implies that any optimal stopping time must satisfy $\mathbb{E}[\tau] \le c^{-1}(L+\mathbb{E}_{p_0}[\bar{u}_B(s)]) \coloneq M$. 

It remains to show that the primal and dual problems attain the same value even if we do not restrict the sampling strategy to the Neyman allocation. Let $W^P$ and $W^D$ denote the optimal values under the primal and dual problems. Step 1 of this proof shows that the optimal value, $W^D$, under the dual problem is attained when using the Neyman allocation. Step 2 states that the duality gap is 0 when we restrict the sampling strategy to the Neyman allocation. This implies there exists a (possibly randomized) stopping time $\tau^*$ such that choosing $\bm{d^*}$ to be the combination of Neyman-allocation and $\tau^*$ leads to a value, say $W^*$, in the primal problem that is arbitrarily close to $W^D$. But $W^* \le W^P$, and in turn, $W^P \le W^D$ by the max-min inequality, so we conclude that $W^P = W^D$. That $W^P$ is bounded follows from Assumption \ref{asm-1}(ii) since $W^P \le \mathbb{E}_{p_0}[\bar{u}_B(s)]$.

\subsection{Proof of Lemma \ref{lem:alpha_equivalent_bounds}}
Note that $\max\{x, 0\} - S_\alpha(x) = (1-\alpha)x$. Let $V(t,m;\alpha)$ and $(b^+(t;\alpha), b^-(t;\alpha))$ represent the value function and stopping time boundaries under a given $\alpha$. Then, 
\[
V(t,m;\alpha) := \sup_\tau \mathbb{E}_{\bm{d}}\left[B\mathds{1}\left[m_{t + \tau} \geq 0\right] + \lambda\left(\max\{m_{t + \tau}, 0\} - (1-\alpha)m_{t + \tau}\right) - c\tau \vert m_t = m \right].
\]
Assumption \ref{asm-1} implies that we may restrict attention to stopping times $\tau$ satisfying $\mathbb{E}[\tau] < \infty$ (see Step 3 of the proof of Theorem \ref{thm:Sampling_strategy}). Hence, by Doob's optional stopping theorem, 
$$
\mathbb{E}_{\bm{d}}[m_\tau \vert m_t = m] = \mathbb{E}_{\bm{d}}\left[m + \int_t^{t+\tau} \frac{\sigma^{-1}}{\sigma^{-2}\tilde{t} + \varrho_0^{-1}}dW(\tilde{t}) \right] = m.
$$
We thus conclude 
$
V(t,m;\alpha) = V(t,m; 1) - \lambda (1-\alpha)m,
$
which in turn implies
$$
b^+(t; \alpha) = \inf_{m \geq 0} \{B + \lambda \alpha m \geq V(t,m; 1) - \lambda (1-\alpha) m\} = b^+(t; 1).
$$
A similar argument applies to $b^-(t; \alpha)$. \qed

\section{Optimal Strategies under General Covariance Structures}\label{sec:general_cov_structures}

Although Assumption \ref{asm-1}(iv) simplifies the sampling strategy, many of our findings remain valid even without this condition. In this section, we examine the optimal sampling strategy and the optimal stopping rule under more general covariance structures.

Let $\mathcal{N}(\bm{\tilde\mu}_0, \tilde\Sigma)$ be the prior over the transformed treatment means $\bm{\tilde{\mu}} \coloneq (\mu_1/\sigma_1, -\mu_0/\sigma_0)$, and define 
\[
    \tilde{\Sigma}_{11} := \Sigma_{11}/\sigma_1^2,\ \tilde{\Sigma}_{00} :=  \Sigma_{00}/\sigma_0^2,\ \tilde{\Sigma}_{01} := \Sigma_{01}/(\sigma_0\sigma_1).
\]

As in \cite{liang-mu-syrgkanis-ecta-2022}, the prior variance $\tilde\Sigma$ is required to satisfy $\sigma_1(\tilde\Sigma_{11} + \tilde\Sigma_{10}) + \sigma_0(\tilde\Sigma_{01} + \tilde\Sigma_{00}) \geq 0$. Let $\text{cov}_1 = \sigma_1(\tilde\Sigma_{11} + \tilde\Sigma_{01})$ and $\text{cov}_0 = \sigma_0(\tilde\Sigma_{00} + \tilde\Sigma_{10})$. Assume without loss of generality that $\text{cov}_1 \ge \text{cov}_0$, and set
\[
    t^* \coloneq \frac{\text{cov}_1 - \text{cov}_0}{\sigma_0\text{det}(\tilde\Sigma)}.
\]
We can characterize the optimal sampling strategy and the resulting posterior variance by employing similar arguments as in \citet[Lemma 11]{liang-mu-syrgkanis-ecta-2022}.

\begin{thm}\label{thm:Sampling-strategy-general-covariance}
    Assume that Assumptions 1(i)-(iii) hold, and that the Gaussian prior satisfies $\textup{cov}_1 + \textup{cov}_0 \geq 0$, along with $\textup{cov}_1 \geq \textup{cov}_0$. Then:
    \begin{enumerate}[(i)]
        \item Theorem \ref{thm:Sampling_strategy} continues to hold, with the optimal sampling strategy now being
            \[
            q^*_1(t) = \begin{dcases}
                1,\ &\text{if } t \leq t^*, \\
                \frac{\sigma_1}{\sigma_0 + \sigma_1},\ &\text{if }t > t^*.
            \end{dcases}
             \]
        Under $\bm{q}^*$, the posterior variance, $\varrho_t^*$, of $\mu_1 - \mu_0$ can be expressed as
            \begin{align*}
                \varrho_t^* \coloneq 
                \begin{dcases}
                \frac{\sigma_1^2\tilde\Sigma_{11} + \sigma_0^2\tilde\Sigma_{00} + 2\sigma_0\sigma_1\tilde\Sigma_{01} + \sigma_0^2\textup{det}(\tilde\Sigma)t}{1 + \tilde\Sigma_{11}t}\ &\text{if }t\leq t^*; \\
                \frac{(\sigma_0 + \sigma_1)^2\textup{det}(\tilde\Sigma)}{\tilde\Sigma_{11} + \tilde\Sigma_{00} - 2\tilde\Sigma_{01} + \textup{det}(\tilde\Sigma)t}\ &\text{if }t > t^*,
                \end{dcases}
            \end{align*}
        and the posterior mean $m_t$ of $\mu_1 - \mu_0$ evolves as:
        \[
             m_t := \mathbb{E}[\mu_1-\mu_0|\mathcal{F}_t] = \mu_1^0 - \mu_0^0 + \int_0^t \sqrt{-\frac{d\varrho_s^*}{ds}}\ dW(s),
        \] 
        where $W(\cdot)$ represents standard Brownian motion. 
    \item If, furthermore, Assumptions 2 and 3 apply, then Theorem \ref{thm:optimal-stopping-rule} also continues to hold.
\end{enumerate}    
\end{thm}
\begin{proof}
By construction (see \citealt{liang-mu-syrgkanis-ecta-2022}, Theorem \ref{thm:Sampling_strategy}), the proposed sampling strategy induces the lowest posterior variance of any strategy at every time $t$. Part (i) of this result then follows by applying the same arguments as in the proof of our Theorem \ref{thm:Sampling_strategy}.

The modification of the optimal sampling strategy influences the optimal stopping rule solely via its impact on the quadratic variation of $m_t$. In fact, the majority of the conclusions in Theorem \ref{thm:optimal-stopping-rule} do not depend on the specific form of $m_t$. The only step that requires an extra argument is proving that $b^+(\cdot)$ and $|b^-(\cdot)|$ are decreasing in $t$. As in the proof of Theorem \ref{thm:optimal-stopping-rule}(ii), this would in turn follow once we establish that the value function $V(t,m)$ is decreasing in $t$.

To this end, let $\varsigma(\rho)$ be such that
$
    \rho = \varrho_0 - \varrho_{\varsigma(\rho)}^*.
$
By part (i) of this theorem,
\[
    \varsigma(\rho) = \begin{dcases}
        \frac{\rho}{\text{cov}_1^2 - \tilde\Sigma_{11}\rho}, &\text{if }t \leq t^*,\\
        \frac{(\sigma_0 + \sigma_1)^2}{\sigma_0^2\tilde\Sigma_{00} + \sigma_1^2\tilde\Sigma_{11} + 2\sigma_0\sigma_1\tilde\Sigma_{01} - \rho} - \frac{\tilde\Sigma_{00} + \tilde\Sigma_{11} - 2\tilde\Sigma_{01}}{\tilde\Sigma_{00}\tilde\Sigma_{11} - \tilde\Sigma^2_{01}},\ &\text{if }t > t^{*}.
    \end{dcases}
\]
It is easily verified that $\varsigma(\rho)$ is strictly increasing and continuous, so we can apply the same argument as in the proof of  \citet[Lemma 2(vi)]{fudenberg-strack-strzalecki2018} to show that $V(t,m)$ is decreasing in $t$.
\end{proof}

\newpage

\section*{\textbf{SUPPLEMENTAL APPENDIX}}

\section{Proofs of the Remaining Results}

\subsection{Proof of Theorem \ref{thm:optimal-stopping-rule}}

In what follows, we occasionally index $V(t,m), b^+(t), b^-(t)$ by the parameters $\lambda, B$ or $c$ to make their dependence on these quantities explicit. 

\begin{enumerate}[(i)]
        \item From the definition of $\tau^*$ in (\ref{eq:stopping-time-def}), the stopping set at time $t$, $\mathcal{S}_t$, can be written as
        \begin{align*}
        \mathcal{S}_t & \equiv \{m \geq 0: B + \lambda S_\alpha(m) \geq V(t, m)\} \cup \{m < 0: \lambda  S_\alpha(m) \geq V(t,m)\} \\
        &:= \mathcal{S}_t^+ \cup \mathcal{S}_t^-.
        \end{align*}

        By parts \eqref{lem1m+} and \eqref{lem1m-} of Lemma \ref{lem:lemsupp}, $V(t,m) - \lambda S_\alpha(m)$ is increasing on the domain $m < 0$ and decreasing on the domain $m \geq 0$. Consequently, we can write
        $$
        \mathcal{S}_t^- = \{m \leq b^-(t)\}, \quad
        \mathcal{S}_t^+ = \{m \geq b^+(t)\}
        $$
        for some $b^-(\cdot) \le 0$ and $b^+(\cdot) \ge 0$, with the understanding that $b^-(t),b^+(t)$ may be taken to be $-\infty, \infty$ if the sets are empty. 

        When $B = 0$, we have $b^+(t; B = 0) < \infty$ by \citet[Theorem 4.1]{fudenberg-strack-strzalecki2018}. But by part (vi) of this theorem, $b^+(t;B) \le b^+(t;0)$ for all $B$. Hence, $b^+(t;B) < \infty$.

        To show $|b^-(t)| < \infty$ it suffices to show  $\mathcal{S}_t^-$ is non-empty for all $t$. By Lemma \ref{lem:alpha_equivalent_bounds}, we may take $\alpha = 1$ without loss of generality. Suppose, by way of contradiction, that there exists $t > 0$ such that $V(t,m) > \lambda \max\{m,0\} = 0$ for all $m < 0$ (if $t = 0$, we may use Lemma \ref{lem:lemsupp}(\ref{lem1time-change}) to shift time). Since $V(\cdot)$ is non-increasing in $t$ by Lemma \ref{lem:lemsupp}\eqref{lem1t}, we also have  $V(\tilde{t},m) > 0$ for any $m < 0$ and $\tilde{t} < t$. Fix some state $(t^\prime, m < 0)$ such that $t^\prime < t$ and define the stopping time 
        $$
        \Delta\tau := \left(\inf \{\tilde{t} \ge t^\prime: m_{\tilde{t}} \ge 0\} \wedge t \right) - t^\prime.
        $$
        As $V(\tilde{t},m) > 0$ for any $m < 0$ and $\tilde{t}\in [t^\prime, t]$, under $\alpha = 1$, stopping is never optimal between $[t^\prime, t]$, as long as $m_{\tilde{t}} < 0$. Then, an upper bound on the continuation-value $V(t^\prime,m)$ can be obtained by supposing that the true state $\mu$ is revealed at $\Delta\tau$; this implies
        $$
        V(t^\prime,m) \le B\mathbb{P}(\mu \ge 0 |t^\prime,m) + \lambda \mathbb{E}[\max\{\mu, 0\}|t^\prime,m] - c\mathbb{E}[\Delta\tau |t^\prime, m].
        $$
        
        Now, as $m \to -\infty$, standard properties of the Gaussian distribution imply $\mathbb{P}(\mu \ge 0 |t^\prime,m) \to 0$ and $\mathbb{E}[\max\{\mu, 0\}|t^\prime,m] \to 0$. Turning to the term $\mathbb{E}[\Delta\tau |t^\prime, m]$, observe that the quadratic variation of $m_t$ is $\varrho_0t/(\sigma^2 \varrho^{-1} +t)$ and is therefore bounded by $\varrho_0$. Consequently, it follows from standard results in stochastic analysis that $\mathbb{P}(\Delta\tau \neq t -t^\prime) \le \exp\{-\vert m \vert^2/2\varrho_0(t-t^\prime)\}$, which implies $\lim_{m \to -\infty}\mathbb{E}[\Delta\tau |t^\prime, m] = (t - t^\prime) > 0$. Taken together, we conclude $\lim_{m \to -\infty} V(t^\prime, m) < 0$ --- a contradiction of our earlier assertion that $V(t,m) > 0$ for all $m < 0$! 
        
        \item By Lemma \ref{lem:lemsupp}\eqref{lem1t}, $V(\cdot, m)$ is decreasing in $t$. This implies, by the definitions of $b^+(t)$, $b^-(t)$ in (\ref{eq:def_of_b_t+}) and (\ref{eq:def_of_b_t-}), that $b^+(t)$ is decreasing in $t$ and $b^-(t)$ is increasing in $t$. 

        We now show that $\lim_{t \to \infty}b^+(t) = 0$. By Lemma \ref{lem:alpha_equivalent_bounds}, we may take $\alpha = 1$ without loss of generality. 
        Define $\nu_t = \frac{\sigma^{-1}}{\sigma^{-2}t + \varrho_0^{-1}}$, let
        $$
        \langle m \rangle (t) \coloneq\int_0^t \upsilon_{t^\prime}d{t^\prime} = \frac{\varrho_0t}{\sigma^2\varrho_0^{-1}+t}
        $$
        denote the quadratic variation of $m_t$, and set $\varsigma(\gamma) = \langle m \rangle^{-1}(\gamma) = \frac{\gamma\sigma^2\varrho_0^{-1}}{\varrho_0 - \gamma}$. Then, by the Dambis-Dubins-Schwartz theorem, $W(\gamma) \coloneq m_{\varsigma(\gamma)} - m_0$ is a standard Brownian motion for $\gamma \in [0, \varrho_0)$. Suppose, by way of contradiction, that there exists $K^* > 0$ such that $b^+(t) \geq K^*$ for all $t$. Consider the state $(t = 0,m = K^*/2)$. The probability of $m_t$ staying within the continuation interval $(b^-(\cdot), b^+(\cdot))$ over all time is then bounded from below by
        \begin{align*}
            \mathbb{P}\left(   m_t \in (0, b^+(t))\ \forall\ t \big\vert m_0 = \frac{K^*}{2} \right) 
            &  \geq \mathbb{P}\left(  m_t \in (0, K^*)\ \forall\ t \big\vert m_0 = \frac{K^*}{2} \right)  \\
            & = \mathbb{P}\left(\sup_{\varsigma \in [0, \varrho_0)} \vert W_\varsigma \vert \in (-K^*/2, K^*/2)\right) > 0.
        \end{align*}
        Because there is a positive probability of waiting forever and incurring infinite cost, we conclude $V(0, K^*/2)$ is $-\infty$. This is a contradiction because the optimal continuation value is always bounded below by Assumption \ref{asm-1}(iii). 

        The argument that $\lim_{t \to \infty}|b^-(t)| = 0$ is analogous. 

        \item Let $b^+(t;B), b^-(t;B)$ denote the optimal boundaries with their dependence on $B$ made explicit. Note that $b^+(t; 0) = -b^-(t;0)$; this follows from Theorem 2 of \cite{fudenberg-strack-strzalecki2018}. As a result,
        \[
            b^+(t;B) \leq b^+(t; 0) = |b^-(t; 0)| \leq |b^-(t; B)|,
        \]
        where the inequalities follow from part (vi) of this theorem.
        
        For the second statement, we start by applying Lemma \ref{lem:alpha_equivalent_bounds} to take $\alpha = 1/2$ without loss of generality. Set $t^*$ such that $b^-(t^*) = B/\lambda$. Such a $t^*$ exists because $b^-(t)$ is a monotonically increasing function that converges to 0 by part (ii) of this theorem, and it is also Lipschitz continuous as established in part (iv) of this theorem.
        
        Consider a pair of stopping strategies $(\tau^+, \tau^-)$ employed starting from the state $(t^*,0)$. By construction, they are adapted to the innovation process 
        \[
            \tilde{m}_{t} \coloneq \int_{t^*}^{t} \upsilon_{t^\prime} dW_{t^\prime},\quad\textrm{where }\ \nu_t := \frac{\sigma^{-1}}{\sigma^{-2}t+\varrho_0^{-1}}.
        \]
        By definition, $\tilde{m}_{\tau^+} \geq 0$ and $\tilde{m}_{\tau^-} < 0$. We can further impose that $\tilde{m}_{\tau^-} > -B/\lambda$, since $\vert b^-(t) \vert$ is decreasing by part (ii), so it would be suboptimal to stop at $\tilde{m}_{\tau^-} \leq -B/\lambda$ after time $t^*$.

        We show that the terminal payoffs when the posterior means take on the values $-2B/\lambda - \tilde{m}_{\tau^+}$ and $-2B/\lambda - \tilde{m}_{\tau^-}$ weakly dominate those when the posterior means take on the values $\tilde{m}_{\tau^+}, \tilde{m}_{\tau^-}$. Indeed, the payoffs at $\tilde{m}_{\tau^+}$ and $-2B/\lambda - \tilde{m}_{\tau^+}$ are, respectively,
        \[
            B + \frac{\lambda}{2}|\tilde{m}_{\tau^+}|;\quad \frac{\lambda}{2}\left|-\frac{2B}{\lambda} - \tilde{m}_{\tau^+}\right| = B + \frac{\lambda}{2}|\tilde{m}_{\tau^+}|.
        \]
        These are identical at every realization of $\tilde{m}_{\tau^+}$. On the other hand, the payoffs at $\tilde{m}_{\tau^-}$ and $-2B/\lambda - \tilde{m}_{\tau^-}$ are, respectively,
        \[
            -\frac{\lambda}{2}\tilde{m}_{\tau^-};\quad \frac{\lambda}{2}\left|-\frac{2B}{\lambda} - \tilde{m}_{\tau^-}\right| = B + \frac{\lambda}{2}\tilde{m}_{\tau^-}.
        \]
        The second payoff is greater than the first for any realization of $\tilde{m}_{\tau^-} \in (-B/\lambda, 0)$.

        By the reflection principle, the terminal payoffs $-2B/\lambda - \tilde{m}_{\tau^+}$ and $-2B/\lambda - \tilde{m}_{\tau^-}$ can be attained starting from the state $(t^*,-2B/\lambda)$ by employing the same stopping strategy $(\tau^+, \tau^-)$, but now adapted to $-\tilde{m}_t$.
        This implies that for any stopping strategy $(\tau^+, \tau^-)$ starting at $(t^*,0)$, and such that $\tilde{m}_{\tau^-} > -B/\lambda$, there exists a corresponding stopping strategy with the same distribution of stopping times (due to the reflection principle), and which delivers a greater stopping payoff when starting from the state $(t^*, -2B/\lambda)$. Taking the supremum over all such stopping times, we obtain
        \begin{equation}\label{pf:Thm2:1}
            B = V(t^*, -2B/\lambda) \geq V(t^*, 0) \geq B.
        \end{equation}
        The first equality in (\ref{pf:Thm2:1}) arises because $b^-(t^*) = -B/\lambda$, so it is optimal to stop at the state $(t^*, -2B/\lambda)$ delivering a payoff $B$ (recall that we set $\alpha = 0.5$). The last inequality in (\ref{pf:Thm2:1}) arises because stopping immediately at the state $(t^*, 0)$ delivers a payoff of $B$. Therefore, (\ref{pf:Thm2:1}) shows that it is optimal to stop immediately at $(t^*, 0)$. 
        
        \item Continuity of $b^+(t)$ follows by similar probabilistic arguments as in \citet[Chapter VI, Section 25.1]{peskir2006optimal}. Lipschitz continuity of $b^-(t)$ is shown in Lemma \ref{lem:Lip_continuity}.
        
        \item We show that $b^+(t)$ is non-decreasing in $\lambda$. The argument that $b^-(t)$ is non-increasing in $\lambda$ is analogous. 
        
        Index $V_1(t,m; \lambda), b^+_1(t; \lambda)$ by $\lambda$ to make explicit the dependence of these quantities on the latter. Consider some $\lambda_1 > \lambda$. By the definition of $b^+(\cdot; \lambda_1)$ and Lemma \ref{lem:lemsupp}\eqref{lem1lam}, 
        \begin{align*}
        & V(t, b^+(t;\lambda_1);\lambda) - B -\lambda S_\alpha(b^+(t;\lambda_1)) \\
        & \le V(t, b^+(t;\lambda_1);\lambda_1) - B - \lambda_1 S_\alpha(b^+(t;\lambda_1)) \le 0.
        \end{align*}
        Hence, it follows by the definition of $b^+(t;\lambda)$ that $b^+(t;\lambda) \le b^+(t;\lambda_1)$.

        \item We start with the case of fixed $\lambda$ and index $V(t,m; B), b^+(t; B)$ by $B$ to make explicit the dependence of these quantities on the latter. Consider some $B_1 > B$. By Lemma \ref{lem:lemsupp}\eqref{lem1r}, 
        \begin{align*}
        & V(t,b^+(t;B); B_1) - B_1 - \lambda S_\alpha(b^+(t;B)) \\
        & \leq V(t,b^+(t;B);B) - B - \lambda S_\alpha(b^+(t;B)) \leq 0.
        \end{align*}
        
        Hence, by the definition of $b^+(\cdot)$, it follows that $b^+(t;B_1) \leq b^+(t;B)$. 
        
        An analogous argument, this time using the second implication of Lemma \ref{lem:lemsupp}\eqref{lem1r}, shows that the lower boundary $b^-(t;B)$ is also decreasing in $B$.

        \item By part (vi) of this theorem, $b^-(t;B)$ is decreasing in $B$ for a fixed $c$ and $\lambda$. Hence, $b^-(t;B) \le -b(t;0)$ for all $B$. But when $B=0$, we are in the same setting as \cite{fudenberg-strack-strzalecki2018}, so we can apply their Theorem 4 to conclude $b^-(t; 0) < 0$. 
\end{enumerate}
    
\subsection{Proof of Theorem \ref{thm:limit_experiment}}
By Assumption \ref{asm:asm-on-mu},
$$
\begin{aligned}
W_n^B((\bm{d}_n, \delta_n), \bm{h}) &= \mathbb{E}_{n,\bm{h}}[\delta_n] + \gamma \mathbb{E}_{n,\bm{h}} [\delta_n] \sqrt{n}\mu_n(\bm{h}) - \gamma(1-\alpha^\prime) \sqrt{n} \mu_n(\bm{h}) - c \mathbb{E}_{n,\bm{h}}[\tau_n] \\
&= \mathbb{E}_{n,\bm{h}}[\delta_n] + \gamma \mathbb{E}_{n,\bm{h}}[\delta_n] \mu(\bm{h}) - \gamma(1-\alpha^\prime) \mu(\bm{h}) - c \mathbb{E}_{n,\bm{h}}[\tau_n] + o(1).
\end{aligned}
$$
Similarly, 
$$
W_n^A((\bm{d}_n, \delta_n), \bm{h}) = \mathbb{E}_{n,\bm{h}}[\delta_n] \mu(\bm{h})  - (1-\alpha) \mu(\bm{h}) + o(1).
$$

Hence, the claim follows if we show that for each strategy pair $(\bm{d}_n, \delta_n)$, there is a subsequence $\{n_k\}_k$ and a strategy pair $(\bm{d} \equiv (\bm{q},\tau), \delta)$ in the limit experiment under which $\delta$ is $\mathcal{F}_{\tau}^{\bm{q}}$-measurable, and 
$$
\mathbb{E}_{n_k,\bm{h}}[\delta_{n_k}] \to \mathbb{E}_{\bm{h}}[\delta], \quad \mathbb{E}_{n_k,\bm{h}}[\tau_{n_k}] \to \mathbb{E}_{\bm{h}}[\tau]. 
$$

\subsubsection*{Step 1 (Asymptotic representations)}

The log-likelihood ratio process, conditional on the information collected up to time $t$, is expressed as:
$$
\varphi_n(\bm{h};t)	= \ln\frac{dP_{\theta_{0}^{(1)}+h^{(1)}/\sqrt{n}}^{(1)}}{dP_{\theta_{0}^{(1)}}^{(1)}}\left({\bf y}_{\left\lfloor nq_{n,1}(t)\right\rfloor }^{(1)}\right)+\ln\frac{dP_{\theta_{0}^{(0)}+h^{(0)}/\sqrt{n}}^{(0)}}{dP_{\theta_{0}^{(0)}}^{(0)}}\left({\bf y}_{\left\lfloor nq_{n,0}(t)\right\rfloor }^{(0)}\right),
$$
where, for any $a\in\{0,1\}$ and $\gamma\in[0,1]$,
$$
\ln\frac{dP_{\theta_{0}^{(a)}+h_a/\sqrt{n}}^{(a)}}{dP_{\theta_{0}^{(a)}}^{(a)}}\left({\bf y}_{\left\lfloor n\gamma\right\rfloor }^{(a)}\right):=\sum_{i=1}^{\left\lfloor n\gamma\right\rfloor }\ln\frac{dP_{\theta_{0}^{(a)}+h_a/\sqrt{n}}^{(a)}}{dP_{\theta_{0}^{(a)}}^{(a)}}\left(Y_{i}^{(a)}\right).
$$
Adusumilli (2025) shows that under Assumption \ref{asm:qmd}, 
\begin{equation}\label{eq:SLAN}
\varphi_n^{(a)}(\bm{h};t)= \sum_a \left[ h_a^\intercal I_{a}^{1/2}X_{n,a}(t)-\frac{q_{n,a}(t)}{2}h_a^\intercal I_{a}h_a \right] +o_{\mathbb{P}_{n,0}}(1)\ \textrm{uniformly over } t \in[0,T].
\end{equation}

Now, under an experimental strategy $\bm{d}$ in the limit-experiment, the likelihood-ratio process at time $t$ is given by 
\begin{equation}\label{eq:LR_process_limit_experiment}
\varphi(\bm{h};t) = \sum_a \left[ h_a^\intercal I_{a}^{1/2}X_{a}(t)-\frac{q_{a}(t)}{2}h_a^\intercal I_{a}h_a \right], 
\end{equation}
where $X_a(t):= Z_a(q_a(t))$ for $Z_a(\cdot)$ defined in (\ref{eq:limit_experiment_process}). By \citet[Theorem 1]{adusumilli-continuous-arts}, for every sequence of empirical allocation processes $\bm{q}_n$, there exists a subsequence $\{n_k\}_k$ and an allocation process $\bm{q}$ in the limit experiment such that 
\begin{equation}\label{eq:convergence_of_experiments}
\left\{ X_{n_{k},a}(\cdot),q_{n_{k},a}(\cdot)\right\} _{a}\xrightarrow[\mathbb{P}_{n,0}]{d}\left\{ X_{a}(\cdot),q_{a}(\cdot)\right\} _{a}.
\end{equation}
Combining (\ref{eq:SLAN})-(\ref{eq:convergence_of_experiments}) gives
\begin{equation}\label{eq:LR_process_convergence}
\varphi_{n_k}(\bm{h};\cdot) \xrightarrow [\mathbb{P}_{n,0}]{d} \varphi(\bm{h};\cdot)\ \textrm{for each } \bm{h}.
\end{equation}

Note that $\tau_n$ is tight since it is uniformly bounded. Combined with (\ref{eq:convergence_of_experiments}), this implies 
$$
\left ( \left\{ X_{n_{k},a}(\tau_{n_k}),q_{n_{k},a}(\tau_{n_k})\right\}_a,\tau_{n_k} \right)
$$
is also jointly tight. Furthermore, by construction, the sample paths of $\left\{ X_{n_{k},a}(\cdot),q_{n_{k},a}(\cdot)\right\}_a$ on the interval $[0,t]$ are $\mathcal{F}_t^{\bm{d}_n}$-measurable, and $\tau_n$ is an $\mathcal{F}_t^{\bm{d}_n}$-adapted stopping time. Hence, by similar arguments as in \citet[Proposition 3]{LeCam1979}, there exists a further subsequence---represented again as $\{n_k\}_k$ for ease of notation---and a $\mathcal{F}_t^{\bm{d}}$-adapted stopping time $\tau$ in the limit experiment, such that 
\begin{equation}\label{eq:convergence_of_stopping_times}
\tau_{n_k} \xrightarrow[\mathbb{P}_{n,0}]{d} \tau.
\end{equation}
In view of (\ref{eq:convergence_of_experiments}), the above further implies 
\begin{equation}\label{eq:LR_process_convergence_2}
\left ( \left\{ X_{n_{k},a}(\tau_{n_k}),q_{n_{k},a}(\tau_{n_k})\right\}_a,\tau_{n_k} \right) \xrightarrow [\mathbb{P}_{n,0}]{d} \left ( \left\{ X_{a}(\tau),q_{a}(\tau)\right\}_a,\tau \right).
\end{equation}

\subsubsection*{Step 2 (Convergence of implementation rules and stopping times)}

Since $\delta_n$ is bounded, it is tight. Hence, by (\ref{eq:LR_process_convergence_2}) and Prohorov's theorem, there exists a further subsequence---represented again as $\{n_k\}_k$ for ease of notation---and a random variable $\bar{\delta}\in [0,1]$ such that\\ $(\delta_{n_k}, \{X_{n_k,a}(\tau_{n_k}), q_{n_k,a}(\tau_{n_k})\}_a, \tau_{n_k}) \xrightarrow[\mathbb{P}_{n,0}]{d} (\bar{\delta},\{X_a(\tau),q_a(\tau)\}_a,\tau)$. Combined with (\ref{eq:SLAN}), this implies that for each $\bm{h}$,
\begin{equation}\label{eq:weak_convergence_1}
\begin{aligned}
\begin{pmatrix}
    \delta_{n_k} \\ \tau_{n_k} \\ \varphi_{n_k}(\bm{h}; \tau_{n_k}) 
\end{pmatrix} 
& \xrightarrow[\mathbb{P}_{n,0}]{d} 
\begin{pmatrix}
    \bar{\delta} \\ \tau \\ \ln Z 
\end{pmatrix}\\
\textrm{where } Z &\sim \exp \sum_a \left[ h_a^\intercal I_{a}^{1/2}X_{a}(\tau)-\frac{q_{a}(\tau)}{2}h_a^\intercal I_{a}h_a \right].
\end{aligned}
\end{equation}
Denote 
$
S(t):=\sum_{a}h_a^\intercal I_{a}^{1/2}X_{a}(t)
$
and 
$$
M(t):=\exp\sum_{a}\left\{ h_a^\intercal I_{a}^{1/2}X_{a}(t)-\frac{q_{a}(t)}{2}h_a^\intercal I_{a}h_a\right\}.
$$
By \citet[Lemma 1]{adusumilli-continuous-arts}, $S(t)$ is an $\mathcal{F}_{t}^{\bm{d}}$-martingale, and its quadratic variation is given by $\sum_{a} q_a(t)h_a^\intercal I_{a}h_a$. Hence, $M(t)$ is the stochastic/Doleans-Dade exponential of $S(t)$. As $q_{a}(t) \le T$ almost surely,
$$
E\left[\exp \left\{ \sum_{a}\frac{q_{a}(T)}{2}h_a^\intercal I_{a}h_a\right\} \right]\le\exp\left\{ \sum_{a}\frac{T}{2}h_a^\intercal I_{a}h_a\right\} <\infty.
$$
Thus, Novikov's condition is satisfied, and $M(t)$ is also an $\mathcal{F}_{t}^{\bm{d}}$-martingale. Doob's optional sampling theorem then implies $E[Z] \equiv E[M(1)]=E[M(0)]=1$. 

We now claim that 
\begin{equation}\label{eq:weak_convergence_2}
\begin{pmatrix}
\delta_{n_k} \\ \tau_{n_k}
\end{pmatrix}
\xrightarrow[\mathbb{P}_{n,\bm{h}}]{d}\mathcal{L};\ \textrm{where }\mathcal{L}(B):=E\left[\mathbb{I}\{(\bar{\delta}, \tau)\in B\}Z\right]\ \forall\ B\in\mathcal{B}([0,1] \times[0,T]).
\end{equation}
It is clear from $Z \ge 0$ and $E[Z] = 1$ that $\mathcal{L}$ is a probability measure, and that for every measurable function $f:[0,1] \times [0,T]\to\mathbb{R}$, $\int f d\mathcal{L}=E[f(\bar{\delta}, \tau)V]$. Furthermore, for any $f(\cdot)$ lower-semicontinuous and non-negative, 	
\begin{equation}\label{eq:weak_convergence_3}
\begin{aligned}
\lim\inf \mathbb{E}_{n_k,\bm{h}}\left[f\begin{pmatrix}
\delta_{n_k} \\ \tau_{n_k}
\end{pmatrix} \right]
 & \ge \lim\inf \mathbb{E}_{n_k,0}\left[
 f\begin{pmatrix}
    \delta_{n_k} \\ \tau_{n_k}
   \end{pmatrix}
\exp\{\varphi_{n_k}(\bm{h};T)\} \right] \\
& = \lim\inf \mathbb{E}_{n_k,0}\left[
f\begin{pmatrix}
\delta_{n_k} \\ \tau_{n_k}
\end{pmatrix}
\exp\{\varphi_{n_k}(\bm{h};\tau_{n_k})\} \right] \\
& \ge E\left[
f\begin{pmatrix}
\bar{\delta} \\ \tau
\end{pmatrix}Z \right].
\end{aligned}
\end{equation}
The equality in (\ref{eq:weak_convergence_3}) follows from the law of iterated expectations since $(\delta_{n_k}, \tau_{n_k})$ are $\mathcal{F}_{\tau_{n_k}}^{\bm{d_{n_k}}}$-measurable, and 
$$
\mathbb{E}_{n_k,0}\left[\exp\left\{ \varphi_{n_k}(\bm{h};T)\right\} \big\vert\mathcal{F}_{\tau_{n_k}}^{\bm{d_{n_k}}}\right] 
= \exp\left\{ \varphi_{n_k}(\bm{h};\tau_{n_k})\right\}
$$ 
as the observations are iid given $\bm{h}$. The last inequality in (\ref{eq:weak_convergence_3}) follows from applying the portmanteau lemma on (\ref{eq:weak_convergence_1}). Applying the portmanteau lemma again, in the converse direction, on (\ref{eq:weak_convergence_3}), gives (\ref{eq:weak_convergence_2}). 

Define $\delta =E\left[\bar{\delta}|\{X_{a}(\tau),q_{a}(\tau)\}_{a}\right]$. By construction, $\delta \in [0,1]$ is a valid strategy by Alice since it is  $\mathcal{F}_{\tau}^{\bm{d}}$-measurable. Furthermore, by (\ref{eq:weak_convergence_2}), 
$$
\begin{aligned}
\lim_{k \to \infty} \mathbb{E}_{n_k,\bm{h}}[\delta_{n_k}] &=E\left[\bar{\delta}e^{\sum_{a}\left\{ h_a^\intercal I_{a}^{1/2}X_{a}(\tau)-\frac{q_{a}(\tau)}{2}h_a^\intercal I_{a}h_a\right\} }\right] \\
& =E\left[\delta e^{\sum_{a}\left\{ h_a^\intercal I_{a}^{1/2}X_{a}(\tau)-\frac{q_{a}(\tau)}{2}h_a^\intercal I_{a}h_a\right\} }\right]=\mathbb{E}_{\bm{h}}[\delta],
\end{aligned}
$$
where the second equality follows by the law of iterated expectations, and the last equality follows by the Girsanov theorem. This proves that $\mathbb{E}_{n_k,\bm{h}}[\delta_{n_k}] \to \mathbb{E}_{\bm{h}}[\delta]$. 

The proof that $\mathbb{E}_{n_k,\bm{h}}[\tau_{n_k}] \to \mathbb{E}_{\bm{h}}[\tau]$ is similar.  

\subsection{Proof of Theorem \ref{thm:asymptotic_upper_bound}}

For any $\mu_1 - \mu_0 \in \mathbb{R}$, $\alpha \in [0,1]$ and $M < \infty$, define $u_M(\cdot)$ as the truncated version of $u(\cdot)$: 
$$
u_M(\mu_1 - \mu_0, \delta;\alpha) = \begin{cases}
    M & \text{if } u(\mu_1 - \mu_0, \delta;\alpha) \ge M,\\
    -M & \text{if } u(\mu_1 - \mu_0, \delta;\alpha) \le -M,\\
    u(\mu_1 - \mu_0, \delta;\alpha) & \text{otherwise}.
\end{cases}
$$
Let $W_{n,M}^A(\cdot, \cdot, \bm{h}), W_{n,M}^B(\cdot, \cdot,\bm{h})$ denote the welfare functions of Alice and Bob when $u(\cdot)$ is replaced with $u_M(\cdot)$. Due to Assumption \ref{asm:asm-on-mu} and the use of a Gaussian prior $\Gamma_0$, for any $\eta > 0$, we can find $M < \infty$ such that 
\begin{equation} \label{eq:welfare_truncation}
    \begin{aligned}
        & \int \sup_{(\bm{d}_n,\delta_n)} \left \vert W_n^B((\bm{d}_n,\delta_n),\bm{h}) - W_{n,M}^B((\bm{d}_n,\delta_n),\bm{h}) \right\vert d\Gamma_0(\bm{h}) \\
        & \le 2\gamma \int \left\vert \sqrt{n}\mu_n(\bm{h}) \mathds{1}\{\vert \sqrt{n}\mu_n(\bm{h})\vert \ge M\} \right\vert d\Gamma_0(\bm{h}) \\
        & \le \frac{2\gamma}{M} \int \left\vert \sqrt{n}\mu_n(\bm{h})  \right\vert^2 d\Gamma_0(\bm{h}) 
         \le \frac{2\gamma (\vert \dot{\mu}_1 \vert + \vert\dot{\mu}_0\vert) }{M} \int \left\vert \bm{h}  \right\vert^2 d\Gamma_0(\bm{h}) < \eta.
    \end{aligned}
\end{equation}
The same argument also shows that 
\begin{equation} \label{eq:welfare_truncation_2}
\int \sup_{(\bm{d}_n,\delta_n)} \left \vert W_n^A((\bm{d}_n,\delta_n),\bm{h}) - W_{n,M}^A((\bm{d}_n,\delta_n),\bm{h}) \right\vert d\Gamma_0(\bm{h}) < \eta.
\end{equation}

Analogously, in the limit experiment, define $W_{M}^A(\cdot, \cdot, \bm{h}), W_{M}^B(\cdot, \cdot,\bm{h})$ as the welfare functions of Alice and Bob when $u(\cdot))$ is replaced with $u_M(\cdot)$. Then, we similarly have
\begin{equation} \label{eq:welfare_truncation_3}
    \begin{aligned}
    & \int \sup_{(\bm{d},\delta)} \left \vert W^B((\bm{d},\delta),\bm{h}) - W_{M}^B((\bm{d},\delta),\bm{h}) \right\vert d\Gamma_0(\bm{h}) < \eta, \textrm{ and} \\
    & \int \sup_{(\bm{d},\delta)} \left \vert W^A((\bm{d},\delta),\bm{h}) - W_{M}^A((\bm{d},\delta),\bm{h}) \right\vert d\Gamma_0(\bm{h}) < \eta.
    \end{aligned}
\end{equation}

Let $(\bm{d}_n^*,\delta_n^*)$ denote a sequence of optimal strategy pairs in the fixed-$n$ setting, inducing the welfare $\bar{W}_n^*$. Consider any subsequence along which $\liminf_{n \to \infty} \int W_n^A((\bm{d}_n^*,\delta_n^*),\bm{h})d\Gamma_0(\bm{h})$ is attained. By Theorem \ref{thm:limit_experiment}, for such a subsequence, there exists a further subsequence, denoted by $\{k\}$, and a strategy pair $(\bm{d},\delta)$ in the limit experiment such that $W_k^A(\cdot, \bm{h}) \to W^A(\cdot, \bm{h})$ and $W_k^B(\cdot, \bm{h}) \to W^B(\cdot, \bm{h})$ for each $\bm{h}$. Combined with (\ref{eq:welfare_truncation}) and (\ref{eq:welfare_truncation_2}), an application of the dominated convergence theorem then gives 
\begin{equation} \label{eq:convergence_of_welfare}
\begin{aligned}
\bar{W}_k^* :=  \int W^B_k((\bm{d}_k^*,\delta_k^*), \bm{h}) d\Gamma_0(\bm{h}) & \to \int W^B((\bm{d},\delta), \bm{h}) d\Gamma_0(\bm{h}), \textrm{ and} \\
\int W^A_k((\bm{d}_k^*,\delta_k^*), \bm{h}) d\Gamma_0(\bm{h}) & \to \int W^A((\bm{d},\delta), \bm{h}) d\Gamma_0(\bm{h}).
\end{aligned}
\end{equation}
Since the welfare constraint requires $\int W^A_k((\bm{d}_k^*,\delta_k^*), \bm{h}) d\Gamma_0(\bm{h}) \ge V_0 -\epsilon_k$, where $\epsilon_k \to 0$, we further obtain 
$$
\int W^A((\bm{d},\delta), \bm{h}) d\Gamma_0(\bm{h}) \ge V_0.
$$

We now show that $\delta$ is a Bayes optimal strategy by Alice given $\bm{d}$. As shown in the proof of Theorem \ref{thm:limit_experiment}, we can choose the subsequence $\{k\}$ in such a manner that the likelihood ratios induced by $\bm{d}_k^*$ converge to that induced by $\bm{d}$: 
$$
\varphi_{k}(\bm{h};\tau_{k}^*) \xrightarrow [\mathbb{P}_{n,0}]{d} \varphi(\bm{h};\tau)\ \textrm{for each } \bm{h}.
$$ 
Then, we can use Lemma \ref{lem:convergence_optimal_Bayes} to conclude 
$$
\lim_{k \to \infty} \int W^A_k((\bm{d}_k^*,\delta_k^*), \bm{h}) d\Gamma_0(\bm{h})
= \int W^A((\bm{d},\delta^*_{\bm{d}}), \bm{h}) d\Gamma_0(\bm{h}),
$$
where $\delta^*_{\bm{d}}$ is any Bayes-optimal strategy relative to $\bm{d}$ in the limit experiment. Combined with (\ref{eq:convergence_of_welfare}), this, in turn, implies
$$
\int W^A((\bm{d},\delta), \bm{h}) d\Gamma_0(\bm{h}) = \int W^A((\bm{d},\delta^*_{\bm{d}}), \bm{h}) d\Gamma_0(\bm{h}).
$$
Thus, $\delta \in \Delta(\bm{d})$, where $\Delta(\bm{d})$ is the class of all Bayes-optimal strategies given $\bm{d}$.\footnote{This does not necessarily mean Alice chooses $\delta$ since she breaks ties in favor of Bob.}

To conclude, we have shown that $\delta \in \Delta(\bm{d})$, and that $(\bm{d},\delta)$ satisfies $\int W^A((\bm{d},\delta), \bm{h}) d\Gamma_0(\bm{h}) \ge V_0$. Consequently, $\bm{d}$ is a
valid strategy for Bob's primal problem (\ref{eq:Bobs_primal_limit}), where the optimal welfare is $\bar{W}^*$. Combined with (\ref{eq:convergence_of_welfare}), this proves
$$
\lim_{n \to \infty} \bar{W}_n^* = \int W^B((\bm{d},\delta), \bm{h}) d\Gamma_0(\bm{h}) \le \sup_{\delta \in \Delta(\bm{d})} \int W^B((\bm{d},\delta), \bm{h}) d\Gamma_0(\bm{h})  \le \bar{W}^*,
$$
where the first inequality accounts for the fact that Alice breaks ties in favor of Bob.
\qed

\subsection{Proof of Theorem \ref{thm:optimal-strategy-discrete-time}}

\subsubsection*{Step 1 (Stopping time approximation in limit experiment).}

Let $\bm{d}^* = (\bm{q}^*,\tau^*)$ denote the optimal strategy corresponding to the welfare constraint of $V_0$ in the limit experiment. We start by showing that approximating $\tau^*$ with 
\[
    \tau^*_{T,\xi} := \inf_{t \geq 0} \left\{m_{t} \not\in \left( b^-(t; \lambda^*), b^+(t;\lambda^*) + \xi \right)\right\} \wedge T,
\]
only leads to a negligible loss in welfare when $T$ is large enough and $\xi$ is sufficiently small. Denote the corresponding experimentation strategy by $\bm{d}^*_{T,\xi} = (\bm{q}^*,\tau^*_{T,\xi})$.

Let Alice's optimal decision rules associated with $\bm{d}^*$ and $\bm{d}^*_{T,\xi}$ be written as $\delta^* = \mathds{1}\{m_{\tau^*} \ge 0\}$ and $\delta^*_{T,\xi} =  \mathds{1}\{m_{\tau^*_{T,\xi}} \ge 0\}$, respectively. By the dominated convergence theorem, we have $\mathbb{E}_{\bm{h}}[\tau^*_{T,\xi}] \to \mathbb{E}_{\bm{h}}[\tau^*_{T,0}]$ as $\xi \to 0$, and by the monotone convergence theorem, $\mathbb{E}_{\bm{h}}[\tau^*_{T,0}] \to \mathbb{E}_{\bm{h}}[\tau^*]$ as $T \to \infty$. Thus,
\begin{equation} \label{eq:pf:Thm5:tau_convg}
\lim_{T \to \infty} \lim_{\xi \to 0} \mathbb{E}_{\bm{h}}[\tau^*_{T,\xi}] = \mathbb{E}_{\bm{h}}[\tau^*].
\end{equation}
We next show that
\begin{equation} \label{eq:pf:Thm5:delta_convg}
 \lim_{T \to \infty} \lim_{\xi \to 0} \mathbb{E}_{\bm{h}}[\delta^*_{T,\xi}] = \mathbb{E}_{\bm{h}}[\delta^*].
\end{equation}
By the monotone convergence theorem, we have $\lim_{T \to \infty}\mathbb{E}_{\bm{h}}[\delta^*_{T,0}] = \mathbb{E}_{\bm{h}}[\delta^*]$. Moreover, for any fixed $T$,
$$
\lim_{\xi \to 0} \mathbb{P}_{\bm{h}}\left(m_{\tau^*_{T,\xi}} \ge 0 \right) 
= \mathbb{P}_{\bm{h}}\left(\bigcup_{\xi > 0} \left\{m_{\tau^*_{T,\xi}} \ge 0 \right\} \right) 
= \mathbb{P}_{\bm{h}}\left(m_{\tau^*_{T,0}} \ge 0 \right).
$$
The first equality uses the nesting property $\left\{m_{\tau^*_{T,\xi_1 }} \ge 0 \right\} \subseteq \left\{m_{\tau^*_{T,\xi_2}} \ge 0 \right\}$ whenever $\xi_1 > \xi_2$. The second equality follows because, by (\ref{eq:signal_process_under_Neyman_allocation}) and (\ref{eq:posterior_dist_normal}),
$$
m_t \sim \frac{\varrho_0^{-1}}{\sigma^{-2}t + \varrho_0^{-1}} m_0 + \frac{\sigma^{-2}}{\sigma^{-2}t + \varrho_0^{-1}}(\mu(\bm{h})t + \sigma W(t) ) \quad \textrm{under $\mathbb{P}_{\bm{h}}$},
$$ 
so under $\mathbb{P}_{\bm{h}}$, $m_t$ is a continuous local martingale with strictly positive quadratic variation. Hence, it satisfies the immediate-crossing property: any sample path of $m_t$ that reaches the boundary $b^+(t) \ge 0$ will immediately cross $b^+(t) + \xi$ for some $\xi > 0$ before ever hitting $b^-(t) < 0$. Combining these facts yields \eqref{eq:pf:Thm5:delta_convg}.

Since $\int \mu(\bm{h})d\Gamma_0(\bm{h}) <\infty$ and
$$
\begin{aligned}
W^B((\bm{d}, \delta), \bm{h}) &= \mathbb{E}_{\bm{h}}[\delta] + \gamma \mathbb{E}_{\bm{h}} [\delta] \mu(\bm{h}) - \gamma(1-\alpha^\prime)  \mu(\bm{h}) - c \mathbb{E}_{\bm{h}}[\tau],\\
W^A((\bm{d}, \delta), \bm{h}) &= \mathbb{E}_{\bm{h}}[\delta] \mu(\bm{h})  - (1-\alpha) \mu(\bm{h}),
\end{aligned}
$$
an application of the dominated convergence theorem, together with \eqref{eq:pf:Thm5:tau_convg} and \eqref{eq:pf:Thm5:delta_convg}, implies 
\begin{equation} \label{eq:pf:Thm5:welfare_approx}
\begin{aligned}
\lim_{T \to \infty} \lim_{\xi \to 0} \int W^B \left( (\bm{d}_{T,\xi}^*, \delta^*_{T,\xi}), \bm{h} \right) d\Gamma_0(\bm{h})  &= \int W^B((\bm{d^*}, \delta^*), \bm{h}) d\Gamma_0(\bm{h}) := \bar{W}^*, \\
\lim_{T \to \infty} \lim_{\xi \to 0} \int W^A\left( (\bm{d}^*_{T,\xi}, \delta^*_{T,\xi}), \bm{h} \right) d\Gamma_0(\bm{h})  &= \int W^A((\bm{d^*}, \delta^*), \bm{h}) d\Gamma_0(\bm{h}).
\end{aligned}
\end{equation}

In addition, as $\int W^A((\bm{d^*}, \delta^*), \bm{h}) d\Gamma_0(\bm{h}) \ge V_0$ by construction, we also have
$$
\lim_{T \to \infty} \lim_{\xi \to 0} \int W^A\left( (\bm{d}^*_{T,\xi}, \delta^*_{T,\xi}), \bm{h} \right) d\Gamma_0(\bm{h}) \ge V_0.
$$

\subsubsection*{Step 2 (Convergence of strategies and Alice's welfare).}
Define $x_{n,a}^*(t) := \dot{\mu}_a^\intercal I_a^{-1/2}X_{n,a}(t)$, $X_a^*(t) := Z_a(q_a^*(t))$ and $x_a^*(t) := \dot{\mu}_a^\intercal I_a^{-1/2}Z_a(q_a^*(t))$, where $Z_a(\cdot)$ is defined in (\ref{eq:limit_experiment_process}). By arguments analogous to those used in the proof of \citet[Theorem 3]{adusumilli2025}, we obtain
\begin{equation} \label{eq:pf:Thm5:1}
   X_{n,a}(t) \xrightarrow[\mathbb{P}_{n,0}]{d}X_a^*(t),\quad
   x^*_{n,a}(t) \xrightarrow[\mathbb{P}_{n,0}]{d} x^*_a(t),\quad q^*_{n,a}(t) \xrightarrow[\mathbb{P}_{n,0}]{d} q^*_a(t). 
\end{equation}
It then follows from the continuous mapping theorem that
\begin{equation} \label{eq:pf:Thm5:2}
\begin{aligned}
\tilde m_{n,t} & \xrightarrow[\mathbb{P}_{n,0}]{d} m_t, \textrm{ where} \\
m_t & \sim \frac{\varrho_0^{-1}}{\sigma^{-2}t + \varrho_0^{-1}} m_0 + \frac{\sigma^{-1}}{\sigma^{-2}t + \varrho_0^{-1}}W(t) \quad \textrm{under $\mathbb{P}_0$}.
\end{aligned}
\end{equation}

Let $\mathbb{D}[0,T]$ be the metric space of real-valued functions on $[0,T]$ endowed with the sup norm. For each $z \in \mathbb{D}[0,T]$, define
\[
    \tau_{T,\xi}(z) = \inf\{t: z_t \not\in (b^-(t), b^+(t) + \xi)\} \wedge T.
\]
Because $b^-(t)$ and $b^+(t)$ are continuous, standard properties of Brownian motion imply that, with probability 1, the sample paths of $m_t$ under $\mathbb{P}_0$ lie at continuity points of the functional $\tau_{T,\xi}(\cdot)$. Thus, by the extended continuous mapping theorem \citep[Theorem 1.11.1]{van1996weak},
\begin{equation} \label{eq:pf:Thm5:3}
\tau_{n,T,\xi}^* = \tau_{T,\xi}(m_{n,t}) \xrightarrow[\mathbb{P}_{n,0}]{d} \tau_{T,\xi}^* = \tau_{T,\xi}(m_t).
\end{equation}

Take $\varphi_n^*(\cdot)$ and $\varphi^*(\cdot)$ to be the likelihood ratios induced by $\bm{d}_{n,T,\xi}^*$ and $\bm{d}_{T,\xi}^*$, respectively. Combining (\ref{eq:pf:Thm5:1})–(\ref{eq:pf:Thm5:3}) with (\ref{eq:SLAN}) yields
\begin{equation} \label{eq:pf:Thm5:4}
\varphi_n^*(\bm{h}, \tau_{n,T,\xi}^*)  \xrightarrow[\mathbb{P}_{n,0}]{d} \varphi^*(\bm{h};\tau_{T,\xi}^*) \quad \text{for each } \bm{h}.
\end{equation}
We can then use Lemma \ref{lem:convergence_optimal_Bayes} to conclude that Alice's welfare converges, i.e.,
\begin{equation} \label{eq:pf:Thm5:conv_Alice_welfare}
\lim_{n \to \infty} \int W_n^A\left((\bm{d}_{n,T,\xi}^*,\delta_{n,T,\xi}^*),\bm{h} \right)d\Gamma_0(\bm{h}) = 
\int W^A\left( (\bm{d}_{T,\xi}^*,\delta_{T,\xi}^*),\bm{h} \right)d\Gamma_0(\bm{h}),
\end{equation}
where $\delta_{n,T,\xi}^*,\delta_{T,\xi}^*$ are optimal responses of Alice given $\bm{d}_{n,T,\xi}^*$ and $\bm{d}_{T,\xi}^*$, respectively. Combining this with Step 1 establishes the first claim of the theorem, namely,
$$
\lim_{T \to \infty} \lim_{\xi \to 0} \lim_{n \to \infty} \int W_n^A\left((\bm{d}_{n,T,\xi}^*,\delta_{n,T,\xi}^*),\bm{h} \right)d\Gamma_0(\bm{h}) \ge V_0.
$$

\subsubsection*{Step 3 (Convergence of Bob's Welfare).}

Let $\{n_k\}_k$ denote a subsequence along which 
$$
\liminf_{n \to \infty} \int W_n^B\left((\bm{d}_{n,T,\xi}^*,\delta_{n,T,\xi}^*),\bm{h} \right)d\Gamma_0(\bm{h})
$$
is attained. Since $\delta^*_{n,T,\xi}$ is bounded, it is tight. Hence, by (\ref{eq:pf:Thm5:1})–(\ref{eq:pf:Thm5:3}) and Prohorov's theorem, there exists a further subsequence---represented again as $\{n_k\}_k$ for ease of notation---such that
$$
\left(\delta_{n_k,T,\xi}^*,\{X_{n_k,a}(\tau_{n_k,T,\xi}^*),q_{n_k,a}^*(\tau_{n_k,T,\xi})\}_a,\tau_{n_k,T,\xi}\right) 
\xrightarrow[\mathbb{P}_{n,0}]{d}
\left(\bar{\delta},\{X_{a}(\tau^*_{T,\xi}),q_{a}^*(\tau_{T,\xi})\}_a,\tau_{T,\xi}\right)
$$
for some random variable $\bar{\delta}\in[0,1]$. Combined with (\ref{eq:SLAN}), this implies that for each $\bm{h}$,
\begin{equation}
\begin{aligned}
\begin{pmatrix}
    \delta_{n_k,T,\xi}^* \\ \tau_{n_k,T,\xi}^* \\ \varphi_{n_k}^*(\bm{h}; \tau_{n_k,T,\xi}^*) 
\end{pmatrix} 
& \xrightarrow[\mathbb{P}_{n,0}]{d} 
\begin{pmatrix}
    \bar{\delta} \\ \tau_{T,\xi}^* \\ \ln Z 
\end{pmatrix},\\
\textrm{where } Z &\sim \exp \sum_a \left[ h_a^\intercal I_{a}^{1/2}X_{a}^*(\tau_{T,\xi}^*)-\frac{q_{a}^*(\tau_{T, \xi}^*)}{2}h_a^\intercal I_{a}h_a \right].
\end{aligned}
\end{equation}

Thus, by applying the same type of reasoning as in step 2 of the proof of Theorem \ref{thm:limit_experiment}, we can find a feasible $\mathcal{F}_{\tau^*_{T,\xi}}^{\bm{q}^*}$-measurable strategy $\delta$ such that
$$
\begin{aligned}
    \lim_{k \to \infty} W^B_{n_k}\big((\bm{d}^{*}_{n_k,T,\xi}, \delta^{*}_{n_k,T,\xi}), \bm{h}\big) &= W^B\big((\bm{d}^*_{T,\xi}, \delta), \bm{h}\big), \textrm{ and}\\
    \lim_{k \to \infty} W^A_{n_k}\big((\bm{d}^{*}_{n_k,T,\xi}, \delta^{*}_{n_k,T,\xi}), \bm{h}\big) &= W^A\big((\bm{d}^*_{T,\xi}, \delta), \bm{h}\big),
\end{aligned}
$$
for every $\bm{h}$. Together with (\ref{eq:welfare_truncation_3}) from the proof of Theorem \ref{thm:asymptotic_upper_bound}, this yields
\begin{equation} \label{eq:pf:Thm5:5}
    \lim_{k \to \infty} \int W^B_{n_k}\big( (\bm{d}^*_{n_k,T,\xi}, \delta^*_{n_k,T,\xi}), \bm{h} \big) \, d\Gamma_0(\bm{h})
    = \int W^B\big( (\bm{d}^*_{T,\xi}, \delta), \bm{h}\big) \, d\Gamma_0(\bm{h}),
\end{equation}
and
\begin{equation} \label{eq:pf:Thm5:6}
 \lim_{k \to \infty} \int W^A_{n_k}\big( (\bm{d}^*_{n_k,T,\xi}, \delta^*_{n_k,T,\xi}), \bm{h} \big) \, d\Gamma_0(\bm{h})
 = \int W^A\big((\bm{d}^*_{T,\xi}, \delta), \bm{h}\big) \, d\Gamma_0(\bm{h}).
\end{equation}

Combining (\ref{eq:pf:Thm5:conv_Alice_welfare}) with (\ref{eq:pf:Thm5:6}), we obtain
$$
\int W^A\big((\bm{d}^*_{T,\xi}, \delta), \bm{h}\big) \, d\Gamma_0(\bm{h})
= \int W^A\big( (\bm{d}_{T,\xi}^*,\delta_{T,\xi}^*),\bm{h} \big)\, d\Gamma_0(\bm{h}).
$$
For a fixed $\xi > 0$, however, Alice's Bayes-optimal policy $\delta_{T,\xi}^*$ is almost surely unique since stopping under $\tau_{T,\xi}$ only occurs when $m_t$ is strictly positive or negative. Consequently, $\delta = \delta_{T,\xi}^*$ almost surely, and by (\ref{eq:pf:Thm5:5}) we conclude
$$
\liminf_{n \to \infty} \int W^B_{n}\big( (\bm{d}^*_{n,T,\xi}, \delta^*_{n,T,\xi}), \bm{h} \big) \, d\Gamma_0(\bm{h})
= \int W^B\big( (\bm{d}^*_{T,\xi}, \delta^*_{T,\xi}), \bm{h}\big) \, d\Gamma_0(\bm{h}).
$$
The desired claim then follows by combining the above with the first part of (\ref{eq:pf:Thm5:welfare_approx}). \qed

\section{Auxiliary Results}

\subsection{Properties of $V(t,m)$}

The proof of Theorem \ref{thm:optimal-stopping-rule} makes use of the following preliminary lemma. Recall that $S_\alpha(x) \coloneq \max\{x,0\} - (1-\alpha)x$. In what follows, we occasionally index $V(t,m)$ by the parameters $\lambda, B$ or $c$ to make its dependence on these quantities explicit. 

\begin{lem}\label{lem:lemsupp}
    The value function $V(t,m)$ satisfies the following properties:
    \begin{enumerate}[(i)]
        \item On the domain $m\geq 0$, $V(t,m) - \lambda S_\alpha(m)$ is decreasing in $m$. \label{lem1m+}
        \item On the domain $m < 0$, $V(t,m) - \lambda S_\alpha(m)$ is increasing in $m$. \label{lem1m-}
        \item $V(t,m; \lambda) - \lambda S_\alpha(m)$ is increasing in $\lambda$. \label{lem1lam}
        \item $V(t,m; B) - B - \lambda S_\alpha(m)$ is decreasing in $B$ for a fixed $(c, \lambda)$. Additionally, $V(t, m; B)$ is increasing in $B$ for a fixed $(c, \lambda)$. \label{lem1r}
        \item $V(t, m)$ is decreasing in $t$.\label{lem1t}
        \item $V(t,m; \varrho_0) = V(0, m; \varrho_t)$ for all $t > 0$. \label{lem1time-change}
    \end{enumerate}
\end{lem}

\begin{proof} 
Theorem \ref{thm:Sampling_strategy}(iii) implies that, under the Neyman allocation, the posterior mean of $\mu \coloneq \mu_1 - \mu_0$ at time $t$ is determined by the process 
\begin{equation}\label{postmean}   
m_t = m_0 + \int_0^t \upsilon_{t^\prime} dW(t^\prime), \ \textrm{where } \upsilon_t \coloneq \frac{\sigma^{-1}}{\sigma^{-2}t + \varrho_0^{-1}}. 
\end{equation}

Let $\widetilde{\mathcal{T}}$ denote the class of stopping times that can be represented by a stopping region in the $(t,m)$-space. That is, each $\tau \in \widetilde{\mathcal{T}}$ admits a representation $\tau = \inf\{t \ge 0 : (t, m_t) \in \mathcal{S}_\tau\}$ for some set $\mathcal{S}_\tau \subseteq [0,\infty) \times \mathbb{R}$. For such $\tau$, define the positive and negative parts of $\mathcal{S}_\tau$ by
$$
\begin{aligned}
\mathcal{S}_\tau^{+} &\equiv \{(t,m) \in \mathcal{S}_\tau : m \ge 0\}, \\
\mathcal{S}_\tau^{-} &\equiv \{(t,m) \in \mathcal{S}_\tau : m < 0\},
\end{aligned}
$$
and the corresponding stopping times
$$
\begin{aligned}
\tau^{+} &\equiv \inf\{t \ge 0 : (t, m_t) \in \mathcal{S}_\tau^{+} \} \\
\tau^{-} &\equiv \inf\{t \ge 0 : (t, m_t) \in \mathcal{S}_\tau^{-} \}.
\end{aligned}
$$
By construction, $\tau = \tau^{+} \wedge \tau^{-}$, so each $\tau \in \widetilde{\mathcal{T}}$ is uniquely associated with a pair of stopping strategies $(\tau^{+}, \tau^{-})$.

A key result from optimal stopping theory, see, e.g., \cite{liptser-shiryaev2011}, states that the optimal stopping time for problems of the form (\ref{eq:optimal_stopping_problem}), where $m_t$ is a strong Markov process, belongs to the class $\widetilde{\mathcal{T}}$. Consequently, in what follows, it is without loss of generality to restrict the space of stopping times to $\widetilde{\mathcal{T}}$.

\begin{enumerate}[(i)]
    \item Fix some $m > 0$, and consider any pair of stopping strategies $(\tau^+, \tau^-)$ associated with some $\tau \in \widetilde{\mathcal{T}}$. The expression $V(t,m) - \lambda S_\alpha(m)$ can be written as
    \[
    \mathbb{E}\left[\mathbf{1}_{\{\tau^+ < \tau^-\}}(B + \lambda \alpha m_{\tau^+}) - \mathbf{1}_{\{\tau^+ >\tau^-\}}(\lambda (1-\alpha) m_{\tau^-}) - c(\min\{\tau^+,\tau^-\} - t)\right] - \lambda \alpha m.
    \]
    When $\tau^+ < \tau^-$, we have $m_{\tau^+} = m + \int_t^{\tau^+} \upsilon_{t^\prime} dW(t^\prime) \geq 0$ and the expression inside of the expectation reduces to 
    \[
    B + \lambda\alpha\left(m + \int_t^{\tau^+} \upsilon_{t^\prime}dW(t^\prime)\right) - c(\tau^+ - t) - \lambda\alpha m.
    \]
    Thus the expression is constant in $m$ in this case. A similar argument for the case $\tau^- < \tau^+$ shows that the expression is decreasing in $m$. As a result, the expectation overall is decreasing in $m$ and so is its supremum over all $\tau \in \widetilde{\mathcal{T}}$.
    
    \item The result follows by applying the same approach as in part $(i)$ to the case where $m < 0$.
    
    \item Consider any stopping time $\tau$. Straightforward algebra shows that the expression $V(t,m; \lambda) - \lambda S_\alpha(m)$ depends on $\lambda$ only through $\lambda(\mathbb{E}[S_\alpha(m_\tau)] - S_\alpha(m))$. The claim follows if we show that $\lambda(\mathbb{E}[S_\alpha(m_\tau)] - S_\alpha(m)) \ge 0$. Since $S_\alpha(m)$ is concave, Jensen's Inequality implies
    \[
    \mathbb{E}\left[S_\alpha\left(m + \int_t^{\tau}\upsilon_{t^\prime}dW(t^\prime)\right)\right] - S_\alpha(m) \geq S_\alpha\left(\mathbb{E}\left[m + \int_t^\tau \upsilon_{t^\prime}dW(t^\prime)\right]\right) - S_\alpha(m).
    \]
    By Doob's optional stopping theorem, the stochastic integral on the right-hand side of the expression is mean 0, so $\mathbb{E}[S_\alpha(m_\tau)] - S_\alpha(m) \geq 0$. 
    
    \item Let $\tau^*$ denote the optimal stopping strategy. From the definition of $V(t,m)$ in (\ref{eq:contvalue}), we observe that $V(t,m; B)-B-\lambda S_\alpha(m)$ depends on $B$ only through $B(\mathbb{P}[m_{\tau^*} \geq 0|t,m] - 1)$. Thus, it is weakly decreasing in $B$.

    That $V(t,m; B)$ is increasing in $B$ is trivial since Bob's utility function is increasing in B. 
 
    \item The result follows from the same time-change argument as in Lemma 2(vi) from \cite{fudenberg-strack-strzalecki2018}.

    \item Follows from Lemma 2(vii) of \cite{fudenberg-strack-strzalecki2018}.
\end{enumerate} 
\end{proof}

\subsection{Lipschitz continuity of $b^-(t)$}

To show Lipschitz continuity of $b^-(t)$, we adapt a set of lemmas from \cite{fudenberg-strack-strzalecki2018}. By Lemma \ref{lem:alpha_equivalent_bounds}, it suffices to look at the case where $\alpha = \frac{1}{2}$, so we take $S_\alpha(m) = \frac{1}{2}\lambda|m|$ without loss of generality.

In what follows, we write $V(t,m)$ and $b^-(t)$ as $V(t,m; c, \varrho_0, B, \sigma)$ and  $b^-(t; c, \varrho_0, B, \sigma)$ to make explicit the dependence of these quantities on the parameters $c, \varrho_0, B, \sigma$. The value of $\lambda$ is held fixed.

\begin{lem}\label{lemo2fss}
Under Assumptions 1-3, for any $\beta > 0$,
    \begin{align*}
    V(0,m; c\beta, \varrho_0, B, \sigma) &= \sup_\tau \beta^{-1}\mathbb{E}\left[\beta B\mathds{1}[n_\tau \geq 0] + \frac{\lambda}{2}|n_\tau| - c\tau\right], \textrm{ where} \\
    n_t &:= \beta m + \int_0^t \frac{\sigma^{-1}}{\varrho_0^{-1}+t'\beta^{-2}\sigma^{-2}}dW(t').
    \end{align*}
\end{lem}
\begin{proof}
    Recall from the proof setup of Lemma \ref{lem:lemsupp} that we can write $V(0, m; c\beta, \varrho_0, B, \sigma)$ as
    \begin{align*}
        \sup_\tau \mathbb{E}\left[B\mathds{1}\left[m + \int_0^{\tau} \upsilon_{t'}\ dW(t')\geq 0\right] + \frac{\lambda}{2}\left|m + \int_0^{\tau} \upsilon_{t'}\ dW(t')\right| - c\beta\tau\right],
    \end{align*}
    where
    \[
        \upsilon_{t} = \frac{\sigma^{-1}}{\sigma^{-2}t + \varrho_0^{-1}}. 
    \]
    Manipulating the expression above gives
    \begin{align*}
       &  V(0,m; c\beta, \varrho_0, B, \sigma) \\
       & = \beta^{-1}\sup_\tau \mathbb{E}\left[\beta B\mathds{1}\left[\beta m + \int_0^{\tau} \beta\upsilon_{t'}\ dW(t')\geq 0\right] + \frac{\lambda}{2}\left|\beta m + \int_0^{\tau} \beta\upsilon_{t'}\ dW(t')\right| - c\beta^2\tau\right].
    \end{align*}
    Because $\beta > 0$, we can introduce $\beta$ into the indicator function without issue. Having made this adjustment for $B\mathds{1}[\cdot]$, the rest of the proof follows exactly as in Lemma O.2 of \cite{fudenberg-strack-strzalecki2018}.
\end{proof}

\begin{lem}\label{lem:boundaryscaling}
    Under Assumptions 1-3, $b^-(t; c, \varrho_0, B)$ satisfies:
    \begin{enumerate}[(i)]
        \item $b^-(t;c,\varrho_0,B) = b^-(0;c,\varrho_t^*, B)$, where $\varrho_t^*$ is the posterior variance of $\mu_1 - \mu_0$.
        \item $b^-(0;c,\beta^2 \varrho_0,B) = \beta b^-(0;c\beta^{-3}, \varrho_0, \beta^{-1}B)$.
        \item $b^-(0;c\beta, \varrho_0, B) \leq \beta^{-1} b^-(0; c, \varrho_0, \beta B)$ for $\beta > 1$.
    \end{enumerate}
\end{lem}
\begin{proof}
 Part (i) follows from \citet[Lemma 3]{fudenberg-strack-strzalecki2018}. Part (ii) also follows by the same argument as in the proof of \citet[Lemma 3]{fudenberg-strack-strzalecki2018} after some straightforward adjustments. 

It remains to prove Part (iii). By Lemma \ref{lemo2fss}, $V(0,m; c\beta, \varrho_0, B)$ equals
    \begin{align*}
        \sup_\tau \beta^{-1}\mathbb{E}\left[\beta B\mathds{1}\left[\beta m + \int_0^t \frac{\sigma^{-1}}{\varrho_0^{-1} + t'\beta^{-2}\sigma^{-2}}dW(t') \geq 0\right] + \frac{\lambda}{2} \left|\beta m + \int_0^t \frac{\sigma^{-1}}{\varrho_0^{-1} + t'\beta^{-2}\sigma^{-2}}dW(t')\right| - c\tau\right].
    \end{align*}
Since $\beta > 1$, this implies    
\begin{align*}
        V(0, m; c\beta, \varrho_0, B) \geq \sup_\tau \beta^{-1}\mathbb{E}\left[\beta B\mathds{1}\left[\beta m + \int_0^t \frac{\sigma^{-1}}{\varrho_0^{-1} + t'\sigma^{-2}}dW(t') \geq 0\right]\right. \\ + \left.\frac{\lambda}{2} \left|\beta m + \int_0^t \frac{\sigma^{-1}}{\varrho_0^{-1} + t'\sigma^{-2}}dW(t')\right| - c\tau\right].
    \end{align*}
    By definition, the right hand side is equal to $\beta^{-1}V(0, \beta m; c, \varrho_0, \beta B)$. As a result,
\begin{align*}
        b^-(0; c\beta, \varrho_0, B) &:= \sup_{m < 0} \left\{\frac{\lambda}{2}|m|\geq V(0, m; c\beta, \varrho_0, B)\right\} \\
        &\leq \sup_{m < 0} \left\{\frac{\lambda}{2}|m| \geq \beta^{-1}V(0, \beta m; c, \varrho_0, \beta B)\right\} \\ 
        &= \beta^{-1}b^-(0; c, \varrho_0, \beta B).
\end{align*}        

\end{proof}



\begin{lem}\label{lem:Lip_continuity}
    Under Assumptions 1-3, $b^-(t)$ is Lipschitz continuous. 
\end{lem}
\begin{proof}
Let $\gamma_\epsilon = (1 + \epsilon\sigma^{-2}\varrho_0)^{-1/2}\leq 1$. Note that $\gamma^2_\epsilon\varrho_0$ is the posterior variance at time $\epsilon$. We apply each item of Lemma \ref{lem:boundaryscaling}, along with the fact that $b^-(\cdot)$ is weakly decreasing in $B$---see Theorem \ref{thm:optimal-stopping-rule}(vi)---to get
\begin{align*}
        b^-(\epsilon; c,\varrho_0, B) 
         &= b^-(0; c, \gamma_\epsilon^2\varrho_0, B) \\
         &= \gamma_\epsilon b^-(0; c\gamma^{-3}_\epsilon, \varrho_0, \gamma_{\epsilon}^{-1}B)  \\
         &\leq \gamma^4_\epsilon b^-(0;c,\varrho_0, \gamma_\epsilon^{-4} B) \le \gamma^4_\epsilon b^-(0;c,\varrho_0, B).
\end{align*}
In addition, Theorem \ref{thm:optimal-stopping-rule}(ii) shows that $b^-(t;\cdot)$ is weakly increasing in $t$. As a result,
\begin{align*}
             0 \le b^-(\epsilon; c, \varrho_0, B) - b^-(0;c,\varrho_0,B) &\le (\gamma_\epsilon^4 - 1)b^-(0; c,\varrho_0, B) \\
             &\le -2\epsilon\sigma^{-2}\varrho_0 b^-(0; c,\varrho_0, B),
 \end{align*}
where the last step employs the elementary inequality $1 - (1/(1+x)^2) \le 2x$ for any $x \ge 0$. By applying Lemma \ref{lem:boundaryscaling}(i), we can generalize the preceding result to obtain Lipschitz continuity for every $t$:
\begin{align*}
    0 \le b^-(t+\epsilon; c, \varrho_0, B) - b^-(t; c,\varrho_0, B) &= b^-(\epsilon;c,\rho_t^*,B) - b^-(\epsilon;c,\rho_t^*,B)\\
    & \le -2\epsilon\sigma^{-2}\varrho_t^* b^-(0; c, \varrho_t^*, B)  \\
    &= -2\epsilon\sigma^{-2}\varrho_t^{*}b^-(t; c, \varrho_0, B) \\
    &\le -2\epsilon\sigma^{-2}\varrho_0 b^-(0; c, \varrho_0, B),
\end{align*}
where the last step follows from the facts $\varrho_t^* \le \varrho_0$ and $b^-(t;\cdot)$ is weakly increasing. Thus, the Lipschitz constant for $b^-(t;\cdot)$ is $2\sigma^{-2}\varrho_0 \vert b^-(0;c,\varrho_0,B)\vert < \infty$.
\end{proof}

\subsection{Convergence of optimal Bayes risks}

The proofs of Theorem \ref{thm:asymptotic_upper_bound} and \ref{thm:optimal-strategy-discrete-time} make use of the following lemma. Recall that the likelihood ratios $\varphi_n(\bm{h};\tau_n),\varphi(\bm{h};\tau)$ are defined in the proof of Theorem \ref{thm:limit_experiment}. 

\begin{lem}\label{lem:convergence_optimal_Bayes}
    Suppose the likelihood ratios induced by a sequence of experimental strategies, $\bm{d}_n$, converge to those induced by $\bm{d}$, in the sense that 
    $$
     \varphi_{n}(\bm{h};\tau_{n}) \xrightarrow [\mathbb{P}_{n,0}]{d} \varphi(\bm{h};\tau)\ \textrm{for each } \bm{h}.
    $$
    Let $\delta^*_n$ represent Alice's Bayes-optimal strategy relative to $\bm{d}_n$. Then, under Assumptions \ref{asm:qmd}-\ref{asm:asm-on-mu} and a Gaussian prior $\Gamma_0$,
    $$
    \lim_{n \to \infty}\int W^A_n((\bm{d}_n^*,\delta_n^*), \bm{h}) d\Gamma_0(\bm{h}) = \int W^A((\bm{d}^*,\delta^*_{\bm{d}}), \bm{h}) d\Gamma_0(\bm{h}), 
    $$
    where $\delta_{\bm{d}}^*$ is a Bayes-optimal strategy by Alice given $\bm{d}$ in the limit experiment. 
\end{lem}
\begin{proof}
Since $\delta_k^*$ is Bayes optimal relative to $\bm{d}_k^*$ in the finite sample experiment,
$$
\int W^A_n((\bm{d}_n^*,\delta_n^*), \bm{h}) d\Gamma_0(\bm{h}) = \int \sup_{\delta \in [0,1]} \mathbb{E}_{n,\bm{h}} \left[W^A_n((\bm{d}_n^*,\delta), \bm{h})\big \vert \mathcal{F}_{\tau_n}^{\bm{q}_n} \right] d\Gamma_0(\bm{h}). 
$$
Similarly, 
$$
\int W^A((\bm{d}^*,\delta^*_{\bm{d}}), \bm{h}) d\Gamma_0(\bm{h}) = \int \sup_{\delta \in [0,1]} \mathbb{E}_{\bm{h}} \left[W^A((\bm{d}^*,\delta), \bm{h})\big \vert \mathcal{F}_{\tau}^{\bm{q}} \right] d\Gamma_0(\bm{h}). 
$$
With the above in place, the proof proceeds in the following steps. 
\smallskip
\paragraph{\it{Step 1 (Truncating the utility function)}}
Using (\ref{eq:welfare_truncation_2}) and (\ref{eq:welfare_truncation_3}), and recalling the definitions of $W_{n,M}^A(\cdot), W_M^A(\cdot)$ from the proof of Theorem \ref{thm:asymptotic_upper_bound}, we observe that
$$
    \begin{aligned}
       & \lim_{n \to \infty} \int \sup_{\delta \in [0,1]} \mathbb{E}_{n,\bm{h}} \left[W^A_n((\bm{d}_n^*,\delta), \bm{h})\big \vert \mathcal{F}_{\tau_n}^{\bm{q}_n} \right] d\Gamma_0(\bm{h}) \\
       &= \lim_{n \to \infty} \int \sup_{\delta \in [0,1]} \mathbb{E}_{n,\bm{h}} \left[W^A_{n, M}((\bm{d}_n^*,\delta), \bm{h})\big \vert \mathcal{F}_{\tau_n}^{\bm{q}_n} \right] d\Gamma_0(\bm{h}) + o(\eta),
    \end{aligned}
$$
and similarly,
$$
    \begin{aligned}
       &\int \sup_{\delta \in [0,1]} \mathbb{E}_{\bm{h}} \left[W^A_{M}((\bm{d},\delta), \bm{h})\big \vert \mathcal{F}_{\tau}^{\bm{q}} \right] d\Gamma_0(\bm{h}) \\
       &= \int \sup_{\delta \in [0,1]} \mathbb{E}_{\bm{h}} \left[W^A((\bm{d},\delta), \bm{h})\big \vert \mathcal{F}_{\tau}^{\bm{q}} \right] d\Gamma_0(\bm{h}) + o(\eta).
    \end{aligned}
$$
Because $\eta > 0$ is arbitrary, it is therefore enough to establish that
\begin{equation}\label{eq:Bayes_welfare_convg_1}
\begin{aligned}
& \lim_{n \to \infty} \int \sup_{\delta \in [0,1]} \mathbb{E}_{n,\bm{h}} \left[W^A_{n, M}((\bm{d}_n^*,\delta), \bm{h})\big \vert \mathcal{F}_{\tau_n}^{\bm{q}_n} \right] d\Gamma_0(\bm{h}) \\
&= \int \sup_{\delta \in [0,1]} \mathbb{E}_{\bm{h}} \left[W^A_{M}((\bm{d},\delta), \bm{h})\big \vert \mathcal{F}_{\tau}^{\bm{q}} \right] d\Gamma_0(\bm{h}).
\end{aligned}
\end{equation}

By Assumption \ref{asm:asm-on-mu},
$$
\begin{aligned}
& \int \sup_{\delta \in [0,1]} \mathbb{E}_{n,\bm{h}} \left[W^A_{n, M}((\bm{d}_n^*,\delta), \bm{h})\big \vert \mathcal{F}_{\tau_n}^{\bm{q}_n} \right] d\Gamma_0(\bm{h}) \\
& = \int \sup_{\delta \in [0,1]} u_M((\sqrt{n}\mu_n(\bm{h}),\delta;\alpha)) \, d\mathbb{P}_{n,\bm{h}} d\Gamma_0(\bm{h}) \\
& = \int \sup_{\delta \in [0,1]} u_M(\mu(\bm{h}),\delta;\alpha) \, d\mathbb{P}_{n,\bm{h}} d\Gamma_0(\bm{h}) + \epsilon_n \int \vert \bm{h}\vert^2 d\Gamma_0(\bm{h}).
\end{aligned} 
$$
Since $\epsilon_n \to 0$ and, for Gaussian $\Gamma_0$, we have $\int \vert \bm{h}\vert^2 d\Gamma_0(\bm{h}) < \infty$, it is in turn sufficient for proving (\ref{eq:Bayes_welfare_convg_1}) to show that
\begin{equation}\label{eq:Bayes_welfare_convg_2}
\begin{aligned}
\bar{W}_{n,M}^A := \lim_{n\to \infty} \int \sup_{\delta \in [0,1]} u_M(\mu(\bm{h}),\delta;\alpha) \, d\mathbb{P}_{n,\bm{h}} d\Gamma_0(\bm{h}) \\
= \int \sup_{\delta \in [0,1]} u_M(\mu(\bm{h}),\delta;\alpha) \, d\mathbb{P}_{\bm{h}} d\Gamma_0(\bm{h}) := \bar{W}_M^A,
\end{aligned}
\end{equation}
where $\mathbb{P}_{n,\bm{h}}$ and $\mathbb{P}_{\bm{h}}$ are understood as the restrictions of these probability measures to the $\sigma$-algebras $\mathcal{F}_{\tau_n}^{\bm{q}_n}$ and $\mathcal{F}_\tau^{\bm{q}}$, respectively.
\smallskip
\paragraph{\it{Step 2 (Proving \ref{eq:Bayes_welfare_convg_2} for finitely supported $\Gamma_0$)}}

We start by approximating $\Gamma_0$ with a prior $\tilde{\Gamma}_0$ whose support is finite, and the set of support points is denoted by $H$. 

For finite $n$, the measure $\mathbb{P}_{n,\bm{h}}$ is not necessarily absolutely continuous with respect to $\mathbb{P}_{n,0}$. But by the Lebesgue decomposition theorem, we can uniquely express $\mathbb{P}_{n,\bm{h}}$ as the sum of an absolutely continuous component and a singular component:
\begin{equation}
    d\mathbb{P}_{n,\bm{h}} = \varphi(\bm{h};\tau_n) d\mathbb{P}_{n,0} + d\sigma_{n,h},
\end{equation}
where $\sigma_{n,h}$ is a measure that is singular with respect to $\mathbb{P}_{n,0}$.

For any strategy $\delta_n$, we can decompose:
\begin{equation}\label{eq:pf:Bayes_conv:0}
\begin{aligned}
    & \int u_M(\mu(\bm{h}),\delta_n;\alpha)  d\mathbb{P}_{n,\bm{h}} d\tilde{\Gamma}_0(\bm{h})  \\
    & = \int \sum_{\bm{h} \in H} u_M(\mu(\bm{h}),\delta_n;\alpha)  \varphi_n(\bm{h};\tau_n) \tilde{\Gamma}_0(\bm{h}) d\mathbb{P}_{n,0} + \int \sum_{h \in H} u_M(\mu(\bm{h}),\delta_n;\alpha) \tilde{\Gamma}_0(\bm{h}) d\sigma_{n,\bm{h}}. 
\end{aligned}
\end{equation}

We first address the second term in (\ref{eq:pf:Bayes_conv:0}). The total mass of the singular part is given by $\|\sigma_{n,\bm{h}}\| = 1 - \mathbb{E}_{n,0}[\varphi_n(\bm{h};\tau_n)]$. Weak convergence,  $\varphi_{n}(\bm{h};\tau_{n}) \xrightarrow [\mathbb{P}_{n,0}]{d} \varphi(\bm{h};\tau)$, together with uniform integrability of $\varphi_n(\bm{h};\tau_n)$ with respect to $\mathbb{P}_{n,0}$ implies  $\limsup_{n \to \infty} \mathbb{E}_{n,0}[\varphi_n(\bm{h};\tau_n)] = \mathbb{E}_{0}[\varphi(\bm{h};\tau)] = 1$. Consequently, $\limsup_{n \to \infty} \|\sigma_{n,\bm{h}}\| = 0$. Since $u_M(\cdot)$ is bounded by $M$, the welfare contribution from the singular part is bounded by $M \sum_{\bm{h}} \Gamma_0(\bm{h}) \|\sigma_{n,\bm{h}}\| \to 0$. Thus, the second term in (\ref{eq:pf:Bayes_conv:0}) is negligible, and we have
\begin{equation} \label{eq:pf:Bayes_conv:1}
\bar{W}_{n,M}^A = \int \sup_{\delta \in [0,1]} \sum_{\bm{h} \in H} \left[ u_M(\mu(\bm{h}),\delta;\alpha)  \varphi_n(\bm{h};\tau_n) \tilde{\Gamma}_0(\bm{h}) \right]d\mathbb{P}_{n,0} + o(1).
\end{equation}

Define $\phi: \mathbb{R}^{|H|} \to \mathbb{R}$ as:
\begin{equation}
    \phi(\{\xi_{\bm{h}}\}_{\bm{h} \in H}) = \sup_{\delta \in [0,1]} \left( \sum_{\bm{h} \in H} u_M(\mu(\bm{h}),\delta;\alpha) \xi_{\bm{h}} \tilde{\Gamma}_0(\bm{h})  \right).
\end{equation}
Observe that $\bar{W}_M^A = \mathbb{E}_0[\phi(\varphi(\bm{h},\tau)]$. Furthermore, in view of (\ref{eq:pf:Bayes_conv:1}), $\bar{W}_{n,M}^A = \mathbb{E}_{n,0}[\phi(\varphi_n(\bm{h},\tau_n)] + o(1)$. Hence, the claim follows if we show 
\begin{equation} \label{eq:pf:Bayes_conv:2}
\lim_{n \to \infty} \mathbb{E}_{n,0}[\phi(\varphi_n(\bm{h},\tau_n)] = \mathbb{E}_0[\phi(\varphi(\bm{h},\tau)].
\end{equation}

The function $\phi(\cdot)$ is the supremum of a family of linear functions, making it convex and continuous. By the continuous mapping theorem, weak convergence of the likelihood ratios $\varphi_{n}(\bm{h};\tau_{n}) \xrightarrow [\mathbb{P}_{n,0}]{d} \varphi(\bm{h};\tau)$ implies
\begin{equation}
    \phi(\varphi_n(\bm{h},\tau_n) \xrightarrow [\mathbb{P}_{n,0}]{d} \phi(\varphi(\bm{h},\tau).
\end{equation}

To establish (\ref{eq:pf:Bayes_conv:2}), we require the sequence $\phi(\varphi_n(\bm{h},\tau_n)$ to be uniformly integrable. Since the loss is bounded by $M$, we have the inequality:
\begin{equation}
    |\phi(\varphi_n(\bm{h},\tau_n)| \leq M \sum_{\bm{h} \in H} \tilde{\Gamma}_0(\bm{h}) \varphi_{n}(\bm{h};\tau_n).
\end{equation}
The right-hand side is a linear combination of likelihood ratios, $\varphi_{n}(\bm{h};\tau_n)$, which are uniformly integrable with respect to $\mathbb{P}_{n,0}$. Since $\phi(\varphi_n(\bm{h},\tau_n)$ is dominated by a sum of uniformly integrable variables, it is itself uniformly integrable.
\smallskip
\paragraph{\it{Step 3 (Proving \ref{eq:Bayes_welfare_convg_2} for Gaussian $\Gamma_0$)}}

We now argue that approximating $\Gamma_0$ with a finitely supported $\tilde{\Gamma}_0$ leads to negligible difference in welfare. Indeed, it is always possible to choose $\tilde{\Gamma}_0$ such that $\lVert \Gamma_0 -\tilde{\Gamma}_0 \rVert_{\textrm{TV}} \le \eta$ for any $\eta > 0$. Since $u_M(\cdot)$ is bounded by $M < \infty$, this implies 
$$
\left \vert \int \sup_{\delta \in [0,1]} u_M(\mu(\bm{h}),\delta;\alpha)  d\mathbb{P}_{n,\bm{h}} d\Gamma_0(\bm{h}) - \int \sup_{\delta \in [0,1]} u_M(\mu(\bm{h}),\delta;\alpha)  d\mathbb{P}_{n,\bm{h}} d\tilde{\Gamma}_0(\bm{h})\right \vert < M\eta,
$$
and 
$$
\left \vert \int \sup_{\delta \in [0,1]} u_M(\mu(\bm{h}),\delta;\alpha)  d\mathbb{P}_{\bm{h}} d\Gamma_0(\bm{h}) - \int \sup_{\delta \in [0,1]} u_M(\mu(\bm{h}),\delta;\alpha)  d\mathbb{P}_{\bm{h}} d\tilde{\Gamma}_0(\bm{h})\right \vert < M\eta.
$$
As $\eta$ can be chosen arbitrarily small, this proves the desired claim. 
\end{proof}

\subsection{Equivalence of limit experiments}

In this section, we establish a formal connection between the Gaussian diffusion limit experiment from Section \ref{subsec:limit_approximations} and the experiment described in Section \ref{Sec:Incremental_learning}. 

\begin{lem}\label{lem:equivalence_of_limit_experiments}
  Consider the Gaussian diffusion limit experiment described in Section \ref{subsec:limit_approximations}. Assume that the prior $\Gamma_0$ on $\bm{h}$ is Gaussian and decomposes into a prior $p_0$ on $(\dot{\mu}_1^\intercal h_1, \dot{\mu}_0^\intercal h_0) \equiv (\mu_1, \mu_0)$ that satisfies Assumption \ref{asm-1}(iv), together with an independent prior $\tilde{\Gamma}_0$ on the remaining components of $\bm{h}$. Under these conditions, the optimal sampling rules and stopping times in the two experiments are identical once we identify $\sigma_a^2$ with $\dot{\mu}_a^\intercal I_a^{-1}\dot{\mu}_a$. Moreover, Bob's welfare resulting from these strategies is also the same in both experiments.
\end{lem}

\begin{proof}   
    Note that Alice and Bob's payoffs depend on $\bm{h}$ only through the linear functional $\mu(\bm{h}) := \dot{\mu}_1^\intercal h_1 - \dot{\mu}_0^\intercal h_0$. Given the Gaussian prior, we can invoke an argument analogous to that in the proof of \citet[Theorem 5]{liang-mu-syrgkanis-ecta-2022} to establish that the sampling strategy which minimizes the posterior variance of $\mu(\bm{h}) := \dot{\mu}_1^\intercal h_1 - \dot{\mu}_0^\intercal h_0$ uniformly at all times is dynamically optimal. A straightforward calculation of the variance-minimizing allocation rule under our prior assumptions then reveals that it coincides with the Neyman allocation described in the statement of this lemma.

   We now prove that the optimal stopping times coincide as well. By standard properties of Gaussian processes, for each $a$ we can decompose the signal process $Z_a(\cdot)$ into a scalar component $z_a(\cdot) := \dot{\mu}_a^\intercal I_a^{-1}Z_a(\cdot)$ and an orthogonal $(d-1)$-dimensional component $\tilde{Z}_a(\cdot)$ that is independent of $z_a(\cdot)$. Let $\bar{\mathcal{G}}_{\gamma_1, \gamma_0}$ denote the natural filtration generated by the sample paths of $Z_a(\cdot)$ on the interval $[0,\gamma_a]$, augmented with an exogenous randomization $U$, and define $\bar{\mathcal{F}}_t^{\bm{q}^*} = \bar{\mathcal{G}}_{q_1^*(t), q_0^*(t)}$, where $\bm{q}^*$ is the Neyman allocation. Likewise, let $\mathcal{F}_t^{\bm{q}^*} := \mathcal{G}_{q_1^*(t), q_0^*(t)}$, where $\mathcal{G}_{\gamma_1, \gamma_0}$ is the natural filtration generated by the sample paths of $z_a(\cdot)$ on $[0,\gamma_a]$, together with the same exogenous randomization $U$. The payoffs of Alice and Bob depend on $\bm{h}$ only via the scalar functional $\mu(\bm{h}) := \dot{\mu}_1^\intercal h_1 - \dot{\mu}_0^\intercal h_0$, and by the immersion property of these filtrations we have $\mathbb{E}[\mu(\bm{h})|\bar{\mathcal{F}}_t^{\bm{q}^*}] = \mathbb{E}[\mu(\bm{h})|\mathcal{F}_t^{\bm{q}^*}]$ almost surely for all $t$. Therefore, by standard arguments from optimal stopping theory, it is without loss of generality to restrict attention to stopping times that are $\mathcal{F}_t^{\bm{q}^*}$-adapted.

   Because the optimal sampling and stopping rules are identical, and the payoff depends on $\bm{h}$ solely through $\mu(\bm{h}) := \dot{\mu}_1^\intercal h_1 - \dot{\mu}_0^\intercal h_0$, it follows directly that Bob's welfare is the same in both experiments.
\end{proof}

\section{Numerical Methods and Additional Numerical Exercises}\label{sec:numerical_methods}

\subsection{Computing the optimal boundary}
To compute $b^+(t), b^-(t)$, we employ a time-change argument as in Lemma O.5 of \cite{fudenberg-strack-strzalecki2018}.  Let 
$$
    \psi(t) := \langle m \rangle(t) \coloneq\int_0^t \upsilon_{t^\prime}dt = \frac{\varrho_0t}{\sigma^2\varrho_0^{-1}+t}
$$
denote the quadratic variation of $m_t$, and set $\varsigma(\rho) = \psi ^{-1}(\rho) = \frac{\rho\sigma^2\varrho_0^{-1}}{\varrho_0 - \rho}$. Then, by the Dambis-Dubins-Schwartz theorem, $W(\rho) \coloneq m_{\varsigma(\rho)} - m_0$ is a standard Brownian motion for $\rho \in [0, \varrho_0)$. As a result, setting $\alpha = \frac{1}{2}$, the value function can be written as 
\[
V(t,m) = \sup_{\rho \geq \psi(t)} \mathbb{E}\left[B\mathds{1}\{W(\rho) \geq0 \} + \frac{1}{2}\lambda|W(\rho)| - c(\varsigma(\rho) - t)|W(\psi(t)) = m \right],
\]
where $\rho$ is a stopping strategy adapted to the filtration generated by $W(\cdot)$.

By introducing a time change from $t$ to $\rho$, we can express the state variables as $(\rho, m)$ and subsequently simplify the problem. We first discretize the state space $[0, \varrho_0) \times [-\bar{m}, \bar{m}]$ using a grid with step sizes $\Delta_{\rho}$ and $\Delta_m$, respectively. In most practical settings, $\bar{m}$ can be chosen sufficiently large so that it does not interfere with computing the optimal boundaries at time $0$. Next, we approximate the Brownian motion using a binomial scheme: specifically, $W(\rho + \Delta_\rho)$ takes one of the two values $W(\rho) - \Delta_m$ or $W(\rho) + \Delta_m$, each with probability $1/2$. The approximate value function $V(\rho, m)$ is then obtained by backward induction starting from the terminal condition $V(\varrho_0, \cdot) = 0$, via the following recursion:
\begin{align*}
V(\rho, m) = \max&\left\{B\mathds{1}\{m \geq 0\} + \frac{1}{2}\lambda |m|,\right. \\&\ \  \left. \frac{1}{2}V(\rho + \Delta_\rho, m + \Delta_m) + \frac{1}{2}V(\rho + \Delta_\rho, m - \Delta_m) - \varsigma'(\rho)\Delta_\rho\right\}.
\end{align*}
\subsection{Determining the Lagrange multiplier}
Given our parameters and a welfare threshold $V_0$, we apply a binary search to identify the Lagrange multiplier $\lambda$ that yields social welfare equal to $V_0$. The procedure is as follows: first, for a pair of candidate values $\lambda_l < \lambda_h$, we compute the stopping boundary for $\lambda_c = (\lambda_l + \lambda_h)/2$  in terms of the time-change variable $\rho$. Next, we simulate Brownian motion 1 million times to obtain an empirical distribution of stopping values. Using the inverse relation $t = \varsigma(\rho)$, we then recover the implied distribution of stopping times. From this distribution, we calculate $\mathbb{E}[S_\alpha(m_\tau)]$ and repeat the search with updated bound $\lambda_h = \lambda_c$ if this welfare exceeds $V_0$ and $\lambda_l = \lambda_c$ if it is less than $V_0$.

\subsection{Comparison with RCT}
We want to find the expected duration of the stopping rule that achieves the same welfare as an RCT that is always terminated at $t=1$. Let $m_0$ denote the prior mean of $\mu_1 - \mu_0$, and define $\nu^2 := \varrho_0/(1 + \sigma^2\varrho_0^{-1})$. Some straightforward calculations show that the distribution of the posterior mean at the end of the experiment is given by $m_1 \sim N\left(m_0, \nu^2\right)$. Consequently,
\[
V^*_{0} = \mathbb{E}[S_\alpha(m_1)] = m_0\Phi\left(\frac{m_0}{\nu}\right) + \nu\psi\left(-\frac{m_0}{\nu}\right) - (1-\alpha)m_0.
\]
If $m_0 = 0$, the welfare constraint is the same across all choices of $\alpha \in [0,1]$. As a result, $V_0^*$ is independent of $\alpha$, and without loss of generality, we can calibrate $\lambda$ to achieve $\mathbb{E}[S_\alpha(m_\tau)] = V^*_0$ under $\alpha = 1$. The resulting optimal stopping time is also independent of $\alpha$ due to Lemma \ref{lem:alpha_equivalent_bounds}.

\subsection{Additional numerical examples}
\label{sec:comparativeplots}
In Figure \ref{fig:bounds_nu0}, we illustrate how the boundaries vary as the prior standard deviation changes, using the baseline value $\nu_0 = 3.12$. In this setting, $V_0$ is computed separately for each choice of $\nu_0$.

Figure \ref{fig:compare_V} presents the distribution of stopping values for different choices of $V_0$. Recall that, under the parameters specified in Section \ref{sec:Calibration}, the expected welfare from an RCT is $V^*_0 \approx 1.1853$. We display distributions of the posterior means for several multiples of this benchmark. The plots show that, as the welfare constraint becomes more stringent, the stopping boundaries expand and substantially fewer experiments terminate at a posterior mean of 0.

\begin{figure}
    \centering
    \includegraphics[width=0.6\linewidth]{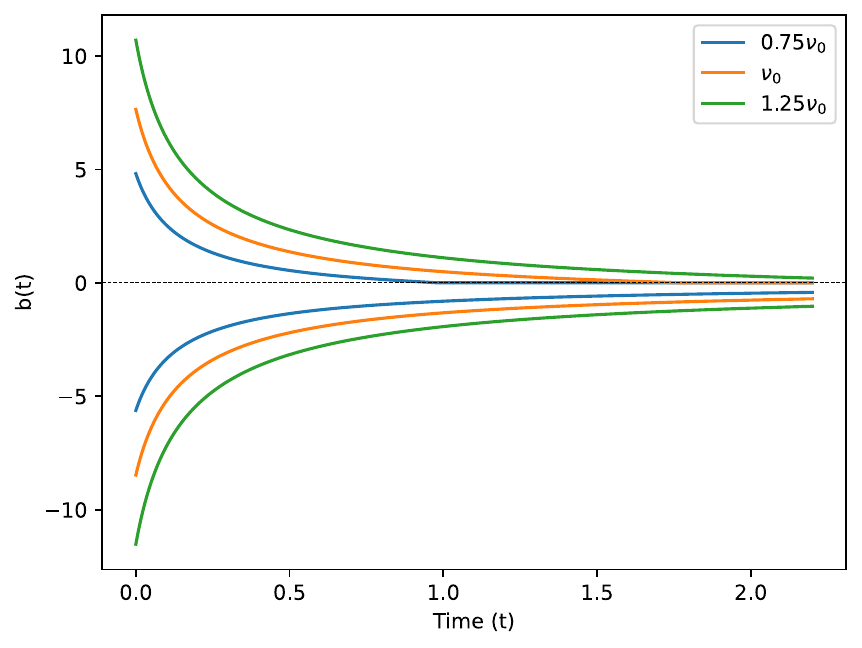}
    \caption{}
    \label{fig:bounds_nu0}
\end{figure}

\begin{figure}[htbp]
    \centering
    \begin{subfigure}[b]{0.4\textwidth}
        \centering
        \includegraphics[width=\textwidth]{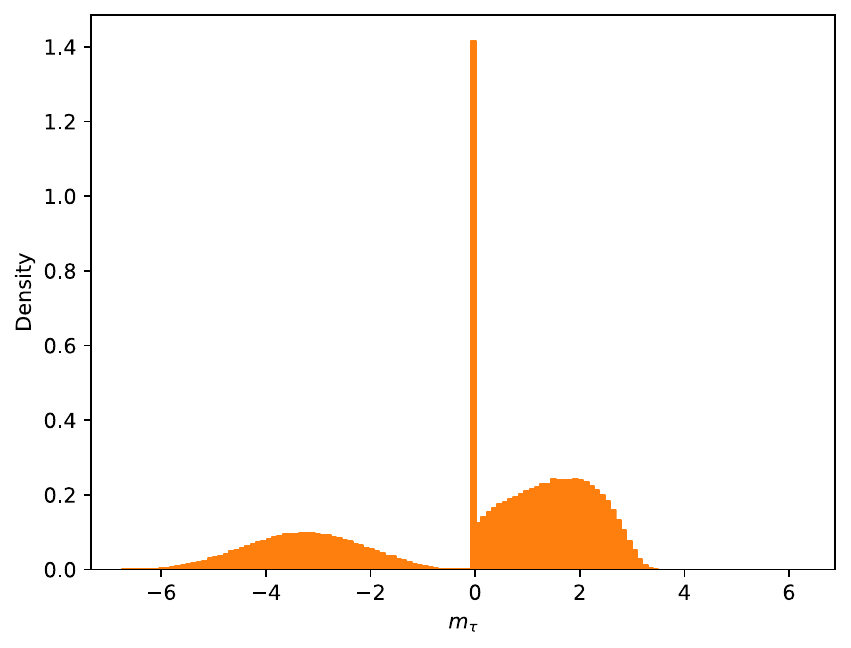}
        \caption{$0.75V^*_0$}
        \label{fig:hist_075V0}
    \end{subfigure}
    \begin{subfigure}[b]{0.4\textwidth}
        \centering
        \includegraphics[width=\textwidth]{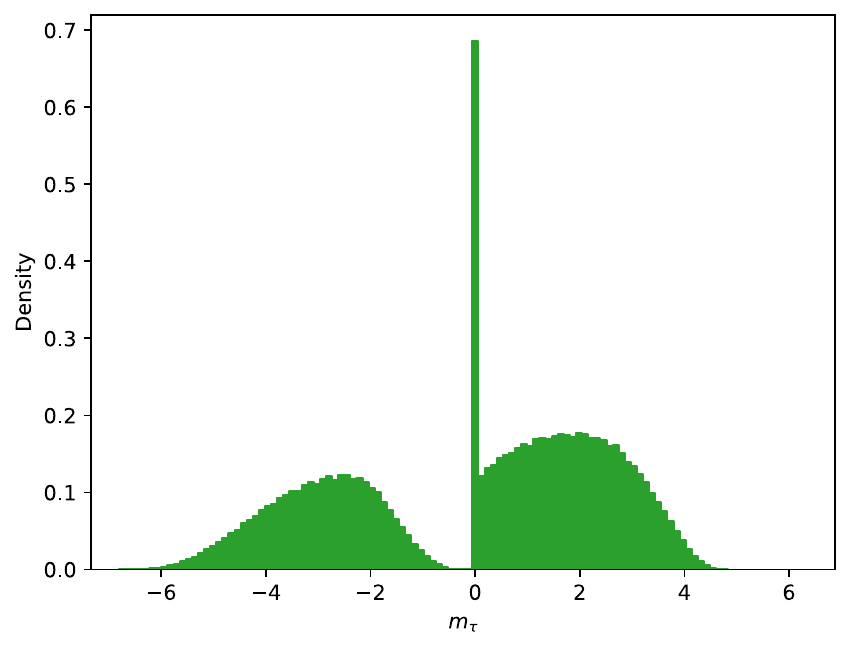}
        \caption{$0.9V^*_0$}
        \label{fig:hist_09V0}
    \end{subfigure}
    
    \vspace{0.5em}
    
    \begin{subfigure}[b]{0.4\textwidth}
        \centering
        \includegraphics[width=\textwidth]{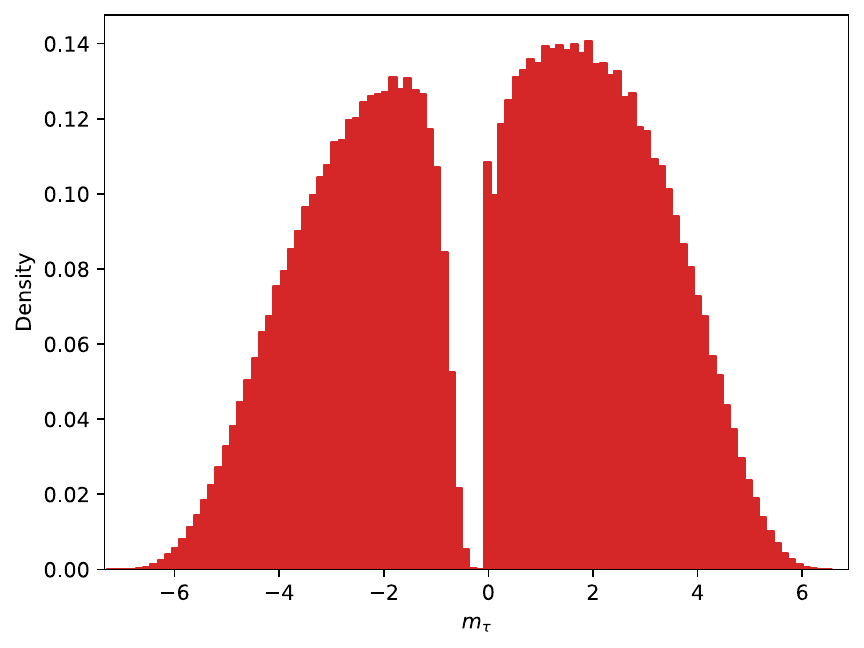}
        \caption{$V^*_0$}
        \label{fig:hist_1V0}
    \end{subfigure}
    \begin{subfigure}[b]{0.4\textwidth}
        \centering
        \includegraphics[width=\textwidth]{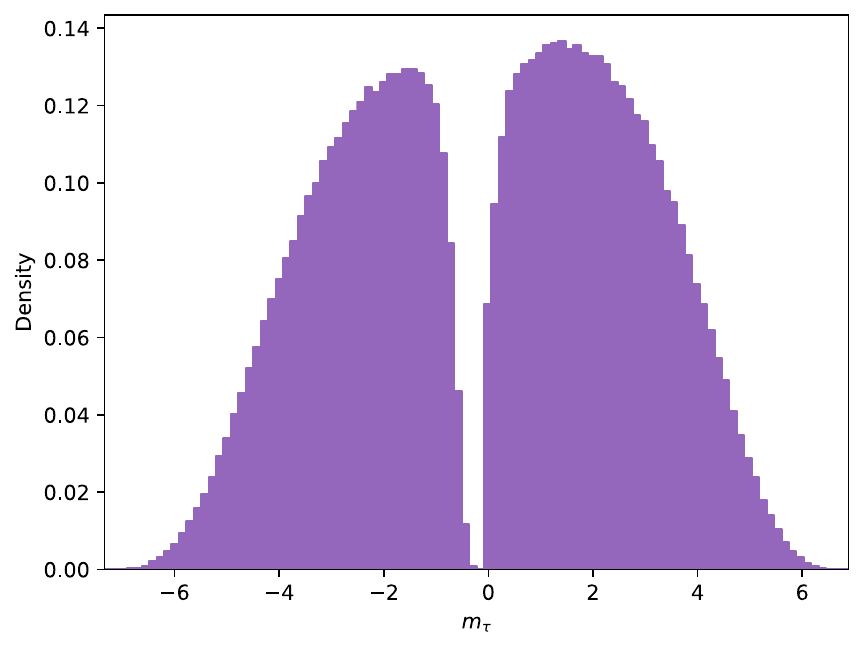}
        \caption{$1.01V^*_0$}
        \label{fig:hist_101V0}
    \end{subfigure}
    
    \caption{Distribution of Posterior Mean $m_\tau$ at Various Values of $V_0$.}
    \label{fig:compare_V}
\end{figure}

\subsection{Alternative calibrations for $B_n$}

In this section, we reproduce the analysis from Section \ref{subsec:Results}, but under alternative calibrations for $B_n$. 

\subsubsection{Approval-based benefit} In the first calibration, we set $B_n = \$802$ million, so that all profits from drug approval are interpreted as purely approval-based benefits. Although this is almost certainly an unrealistic assumption, it is still informative to see how the results change under this setup. With this assumption in place, the limit experiment produces a cost-benefit ratio of $c/B \approx 0.0153$. 

The implied stopping boundaries are shown in Figure \ref{fig:stopboundaries_calibration_1}. Relative to our baseline calibration, the acceptance boundary now declines to 0 much more rapidly. This is intuitive: because the effective cost of experimentation is lower, the pharmaceutical firm has a stronger incentive to keep experimenting when $m_t < 0$, but a weaker incentive to do so once $m_t > 0$. Figures \ref{fig:stoppingdist_calibration_1} and \ref{fig:postmeanstop_calibration_1} display the distributions of $\tau$ and $m_\tau$. Under the present calibration, we obtain $\mathbb{E}[\tau] > 1$, meaning that the average number of observations exceeds that of a conventional RCT. This increase is driven entirely by the additional experimentation when $m_t < 0$. The median stopping time is $0.606$, and 70\% of experiments terminate before $t = 1$, indicating that in most realizations, the number of samples used is below the mean. Figures \ref{fig:welfare_v_meantime_calibration_1} and \ref{fig:welfare_v_medtime_calibration_1} report the expected and median stopping times across a range of welfare levels. As in the baseline calibration, the relationship between $V_0$ and the expected sample size $\mathbb{E}[\tau]$ remains highly non-linear. 

\begin{figure}
    \centering
    \includegraphics[width=0.6\linewidth]{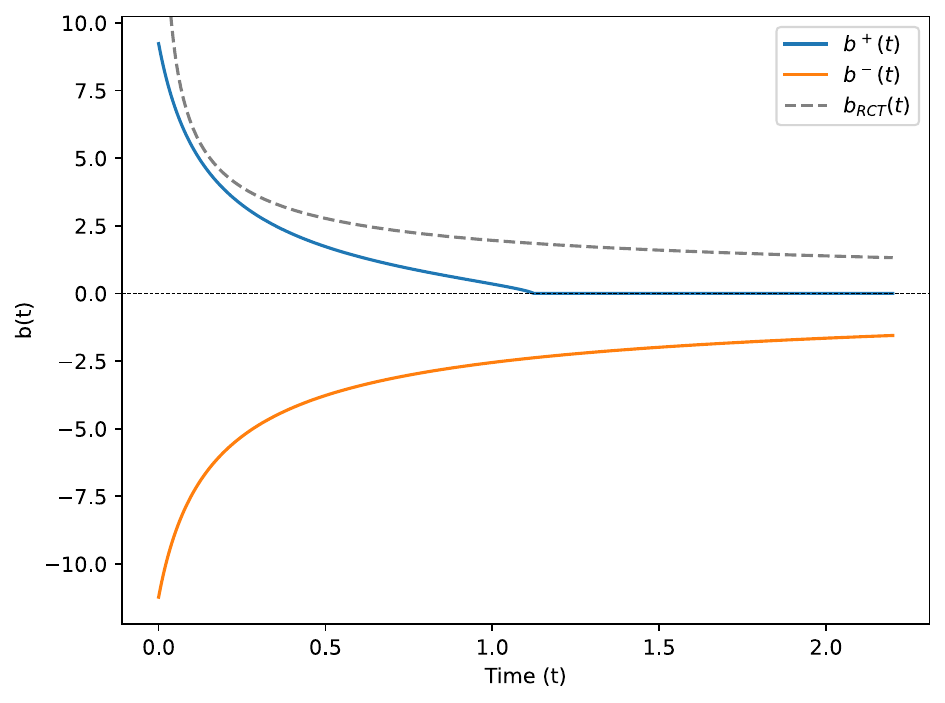}
    \caption{Stopping Boundaries $b^+(t), b^-(t)$ when $B_n =$ \$802 million.}
\label{fig:stopboundaries_calibration_1}
\end{figure}

\begin{figure}[h]
    \centering
    \begin{subfigure}[b]{0.45\linewidth}
        \centering
        \includegraphics[width=\linewidth]{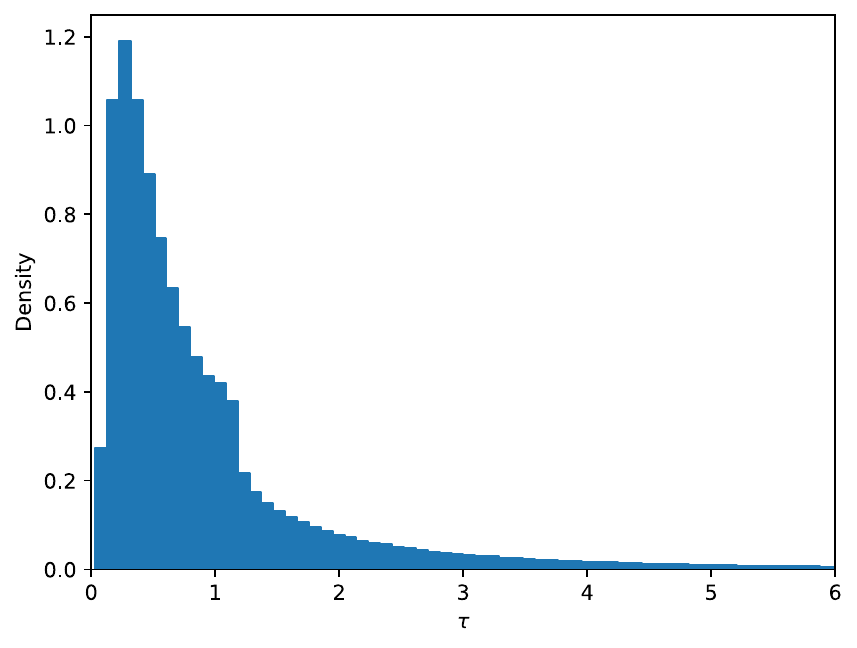}
        \caption{Distribution of $\tau$}
        \label{fig:stoppingdist_calibration_1}
    \end{subfigure}%
    \quad \quad
    \begin{subfigure}[b]{0.45\linewidth}
        \centering
        \includegraphics[width=\linewidth]{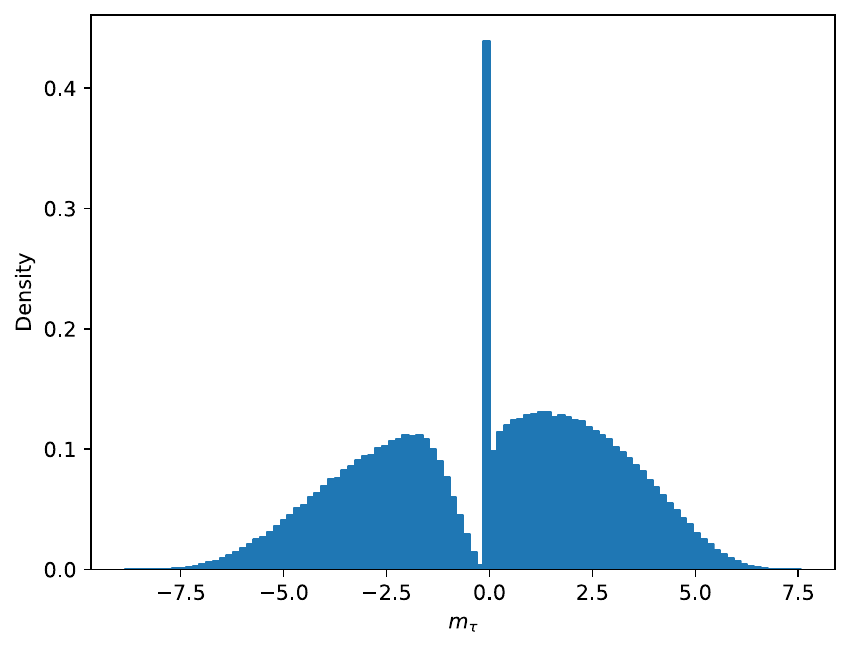}
        \caption{Distribution of $m_\tau$}
        \label{fig:postmeanstop_calibration_1}
    \end{subfigure}
    \caption{Distributions of $\tau$ and $m_\tau$ when $B_n =$ \$802 million.}
    \label{fig:stopping_and_postmean_calibration_1}
\end{figure}

\begin{figure}[h]
    \begin{subfigure}[b]{0.5\linewidth}
        \centering
        \includegraphics[width=\linewidth]{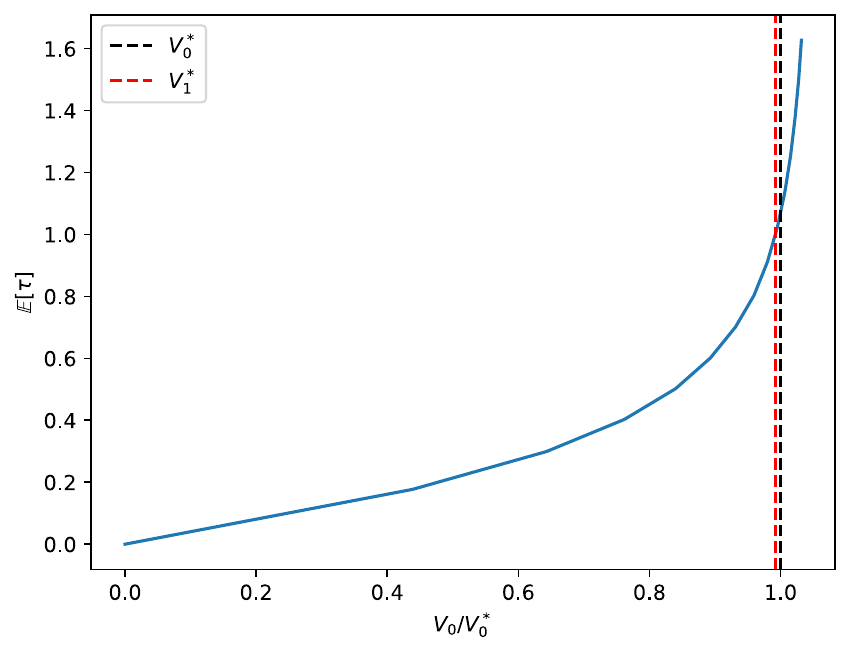}
        \caption{$\mathbb{E}[\tau^*]$ vs. $V_0/V^*_0$}
        \label{fig:welfare_v_meantime_calibration_1}
    \end{subfigure}%
    \begin{subfigure}[b]{0.5\linewidth}
        \centering
        \includegraphics[width=\linewidth]{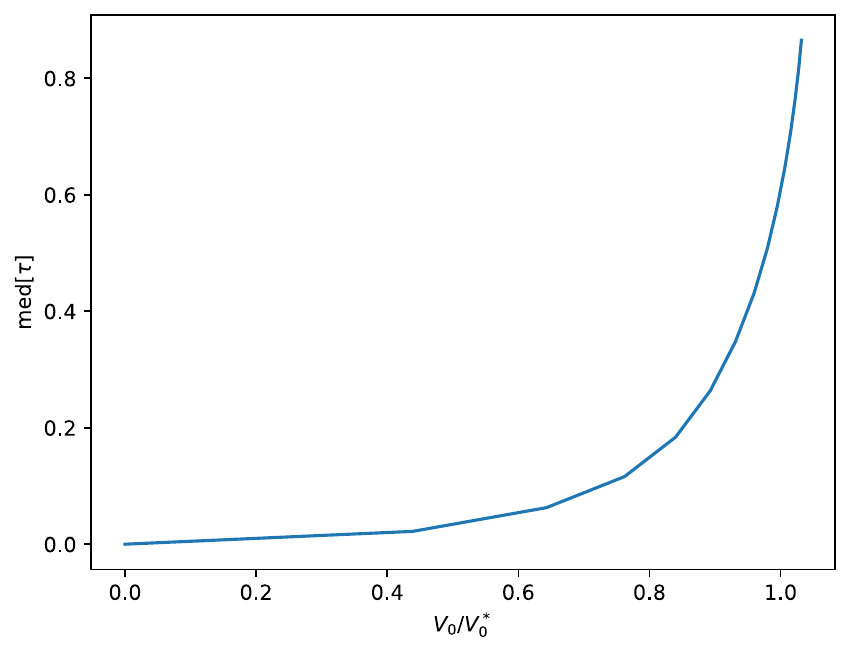}
        \caption{$\text{med}[\tau^*]$ vs. $V_0/V_0^*$}
        \label{fig:welfare_v_medtime_calibration_1}
    \end{subfigure}
    \caption{The effects of changing $V_0$ when $B_n =$ \$802 million.}
\end{figure}

\subsubsection{Welfare-based benefit} At the opposite extreme, we set Bob's approval benefit to $B = 0$, so that his incentives are fully aligned with Alice's. Under this specification, the implied asymptotic cost is $c \approx 7.980$. When $B = 0$, we return to the framework of \cite{fudenberg-strack-strzalecki2018}, but with the stopping boundary calibrated to deliver an average welfare of $V^*_0$. This boundary is plotted in Figure \ref{fig:stopboundaries_calibration_2}. As in \cite{fudenberg-strack-strzalecki2018}, the optimal stopping boundaries are symmetric and remain strictly separated from 0. Figures \ref{fig:stoppingdist_calibration_2} and \ref{fig:postmeanstop_calibration_2} display the distribution of stopping times and the posterior means at stopping, while Figures \ref{fig:welfare_v_meantime_calibration_2} and \ref{fig:welfare_v_medtime_calibration_2} plot the mean and median stopping times across a range of welfare constraint values. On average, the experiment concludes with 56\% fewer observations than the RCT, and 90\% of simulations terminate before $t = 1$. The experiment uses fewer samples than in the calibrations with $B > 0$. The reason is that the additional experimentation in the negative region more than offsets the earlier stopping in the positive region.

\begin{figure}
    \centering
    \includegraphics[width=0.6\linewidth]{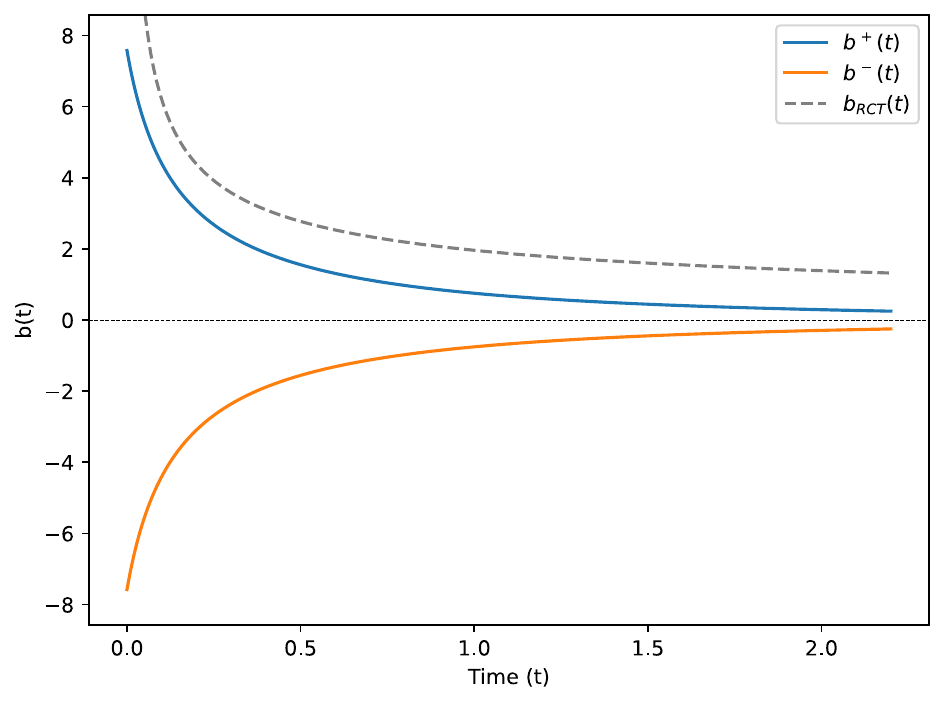}
    \caption{Stopping Boundaries $b^+(t), b^-(t)$ when $B_n =$ \$0.}
\label{fig:stopboundaries_calibration_2}
\end{figure}

\begin{figure}[h]
    \centering
    \begin{subfigure}[b]{0.45\linewidth}
        \centering
        \includegraphics[width=\linewidth]{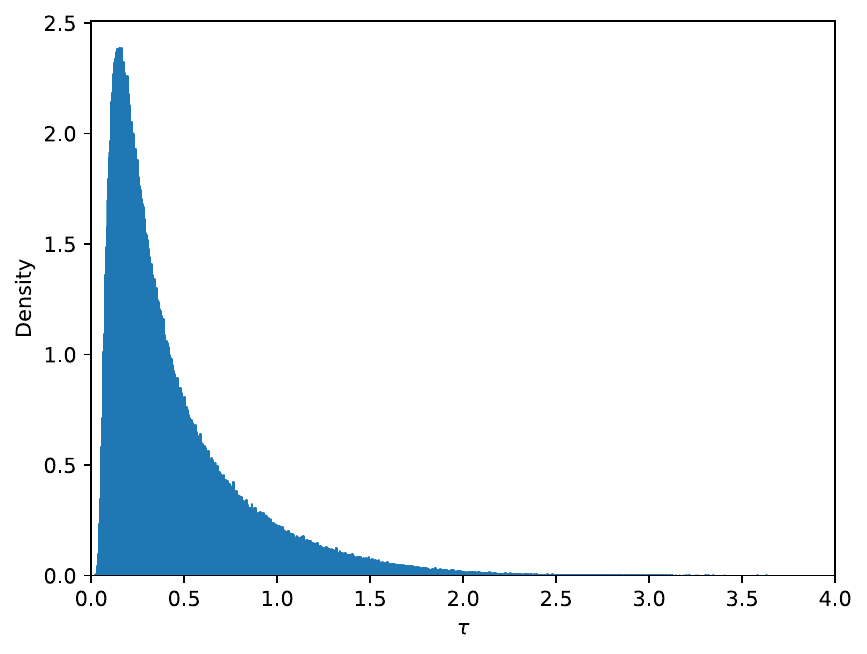}
        \caption{Distribution of $\tau$}
        \label{fig:stoppingdist_calibration_2}
    \end{subfigure}%
    \quad \quad
    \begin{subfigure}[b]{0.45\linewidth}
        \centering
        \includegraphics[width=\linewidth]{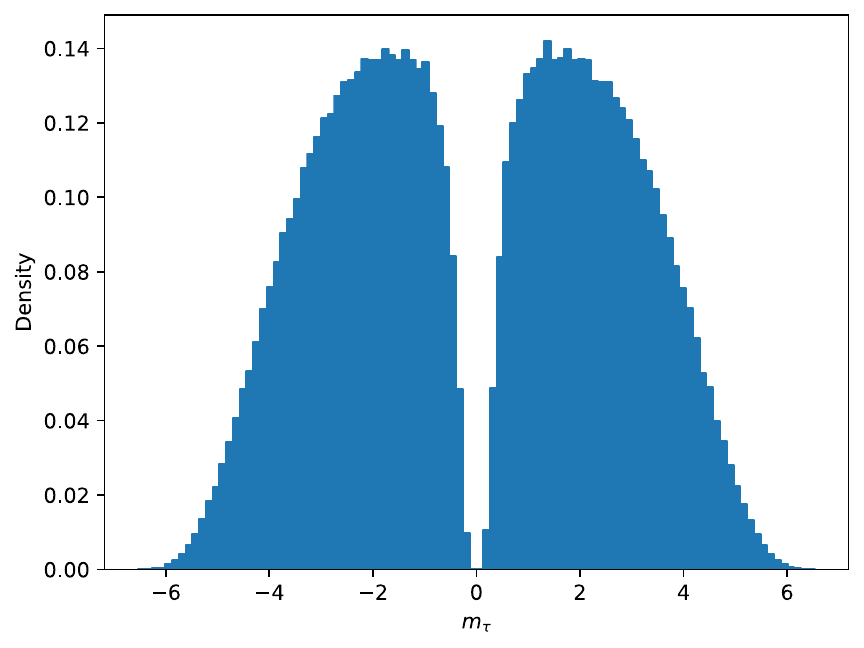}
        \caption{Distribution of $m_\tau$}
        \label{fig:postmeanstop_calibration_2}
    \end{subfigure}
    \caption{Distributions of $\tau$ and $m_\tau$ when $B_n =$ \$0.}
    \label{fig:stopping_and_postmean_calibration_2}
\end{figure}

\begin{figure}[h]
    \begin{subfigure}[b]{0.5\linewidth}
        \centering
        \includegraphics[width=\linewidth]{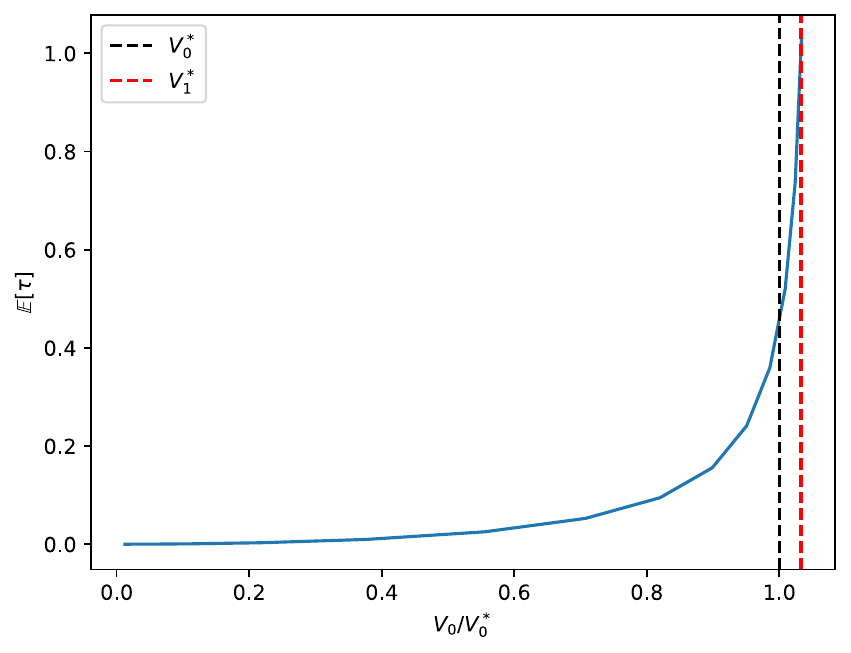}
        \caption{$\mathbb{E}[\tau^*]$ vs. $V_0/V^*_0$}
        \label{fig:welfare_v_meantime_calibration_2}
    \end{subfigure}%
    \begin{subfigure}[b]{0.5\linewidth}
        \centering
        \includegraphics[width=\linewidth]{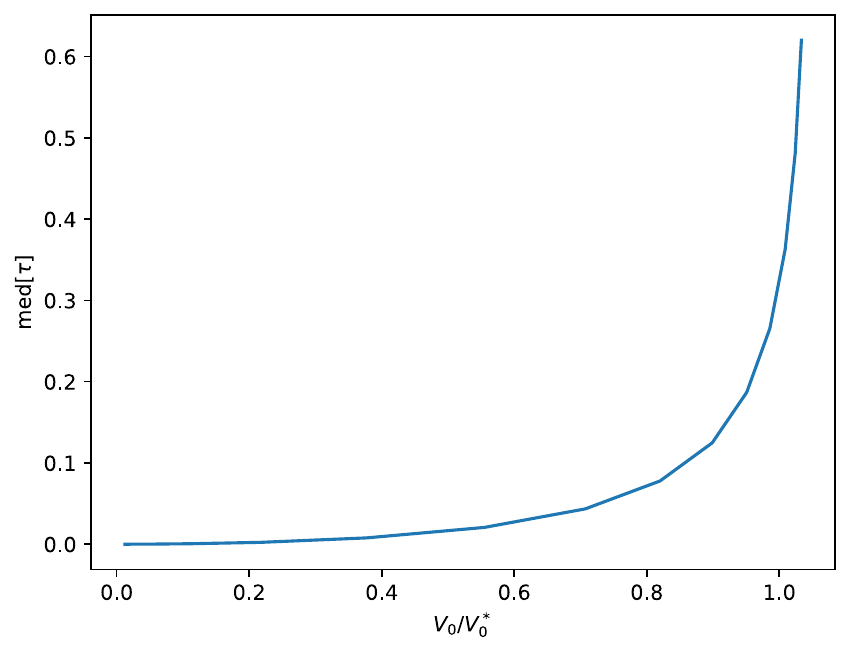}
        \caption{$\text{med}[\tau^*]$ vs. $V_0/V_0^*$}
        \label{fig:welfare_v_medtime_calibration_2}
    \end{subfigure}
    \caption{The effects of changing $V_0$ when $B_n =$ 0.}
\end{figure}
 
\end{document}